\documentclass[journal,onecolumn]{IEEEtran}

\usepackage{url,multirow,comment}
\usepackage{amsmath,amssymb,amsthm,cite}
\usepackage{algorithm}
\usepackage{graphicx}
\usepackage{subcaption}
\usepackage[compatible]{algpseudocode}
\usepackage[flushleft]{threeparttable}
\usepackage{graphicx,color}
\usepackage[figuresright]{rotating}
\usepackage[utf8]{inputenc} 
\usepackage[T1]{fontenc}    
\usepackage[dvipsnames]{xcolor}

\usepackage{array,booktabs,ragged2e}    
\usepackage{enumitem}
\usepackage{amsfonts}       
\usepackage{amssymb,mathtools}
\usepackage{nicefrac}       
\usepackage{microtype}      %
\usepackage{color}
\usepackage{pbox}

\usepackage{thmtools}
\usepackage{thm-restate}
\usepackage{enumitem}
\usepackage{bbm}
\usepackage[bookmarks=false,colorlinks=true,urlcolor=blue,citecolor=blue,linkcolor=blue]{hyperref}

\usepackage{bm}

\newcolumntype{R}[1]{>{\RaggedRight\arraybackslash}p{#1}}

\newtheorem{lem}{Lemma}

\newtheorem{defn}{Definition}

\newtheorem{example}{Example}

\newtheorem{cexample}{Canonical Example}
\newtheorem{scenario}{Scenario}

\newtheorem{rem}{Remark}

\newtheorem{candmeas}{Candidate Measure of Non-Exempt Disparity}

\newcommand{\iid}[0]{i.i.d.}

\newcommand{\paren}[1]{\left( #1 \right)}

\newcommand{\uni}[2]{\mathrm{Uni}({#1: #2})}
\newcommand{\rd}[2]{\mathrm{Red}({#1: #2})}
\newcommand{\syn}[2]{\mathrm{Syn}({#1: #2})}
\newcommand{\mut}[2]{\mathrm{I}({#1; #2})}
\newcommand{\ci}[2]{\mathrm{CCI}({#1 \rightarrow #2})}
\newcommand{\sgn}[1]{\mathrm{sgn}\paren{#1}}
\newcommand{\given}[0]{\mid}

\newcommand{\expect}[0]{\mathbb{E}}

\newcommand{\E}[2]{\expect_{{#1}}\left[{#2}\right]}
\newcommand\independent{\protect\mathpalette{\protect\independenT}{\perp}}
\def\independenT#1#2{\mathrel{\rlap{$#1#2$}\mkern2mu{#1#2}}}

\ifCLASSINFOpdf

\else

\fi

\begin{document}
\bstctlcite{IEEEexample:BSTcontrol}

\title{Fairness Under Feature Exemptions: Counterfactual and Observational Measures}

\author{Sanghamitra Dutta, Praveen Venkatesh, Piotr Mardziel, Anupam Datta, Pulkit Grover\\
Carnegie Mellon University
\thanks{Accepted for publication at the IEEE Transactions on Information Theory; Some of these results have appeared in part at AAAI 2020~\cite{dutta2020information} (oral presentation).}
\thanks{The authors are with the Department of Electrical and Computer Engineering, Carnegie Mellon University, Pittsburgh, USA. Author Contacts: S. Dutta (sanghamd@andrew.cmu.edu), P. Venkatesh (vpraveen@cmu.edu), P. Mardziel (piotrm@cmu.edu), A. Datta (danupam@cmu.edu), P. Grover (pulkit@cmu.edu).}
}

\markboth{Fairness Under Feature Exemptions: Counterfactual and Observational Measures}%
{Dutta \MakeLowercase{\textit{et al.}} 2019.}

\maketitle

\begin{abstract}
With the growing use of machine learning algorithms in highly consequential domains, the quantification and removal of disparity in decision making with respect to protected attributes, such as gender, race, etc., is becoming increasingly important. While quantifying disparity is essential, sometimes the needs of a business (e.g., hiring) may require the use of certain features that are critical in a way that any disparity that can be explained by them might need to be exempted. For instance, in hiring a software engineer for a safety-critical application, a coding-test score may be a critical feature that is weighed strongly in the decision even if it introduces disparity, whereas other features, such as name, zip code, or reference letters may be used to improve decision-making, but only to the extent that they do not add disparity. In this work, we propose a novel information-theoretic decomposition of the total disparity (a quantification inspired from counterfactual fairness) into two components: a non-exempt component which quantifies the part of the disparity that cannot be accounted for by the critical features, and an exempt component which quantifies the remaining disparity. This decomposition is important: it allows one to check if the disparity arose purely due to the critical features (inspired from the business necessity defense of disparate impact law) and also enables selective removal of the non-exempt component of disparity if desired. We arrive at this decomposition through canonical examples that lead to a set of desirable properties (axioms) that any measure of non-exempt disparity should satisfy. We then demonstrate that our proposed counterfactual measure of non-exempt disparity satisfies all of them. Our quantification bridges ideas of causality, Simpson's paradox, and a body of work from information theory called Partial Information Decomposition (PID). We also obtain an impossibility result showing that no observational measure of non-exempt disparity can satisfy all of the desired properties, which leads us to relax our goals and examine alternative observational measures that satisfy only some of these properties. We perform case studies to show how one can audit existing models as well as train new models while reducing non-exempt disparity.
\end{abstract}


%
\IEEEpeerreviewmaketitle

\section{Introduction}
\label{sec:introduction}

As artificial intelligence becomes ubiquitous, it is important to understand whether the output of a machine-learnt model is unfairly biased with respect to \emph{protected attributes} such as gender, race, etc., and if so, how we can engineer fairness into such a model. The field of fair machine learning provides several measures for fairness~\cite{zliobaite2011handling,dwork2012fairness,kamiran2013quantifying,calders2013controlling,agarwal2018reductions,hardt2016equality,calmon2017optimized,menon2018cost,kamishima2012fairness,donini2018empirical,ghassami2018fairness,zafar2015fairness,lipton2018does,varshney2019trustworthy,kusner2017counterfactual,kilbertus2017avoiding,russell2017worlds,chiappa2018path,datta2017use,liao2019learning,zemel2013learning,yeom2018hunting,speicher2018unified,wang2020split,kearns2018preventing,fairMI,kusner2018causal,xu2020algorithmic}, and uses them to reduce disparity, e.g., as a regularizer during training~\cite{agarwal2018reductions,kamishima2012fairness}. In several applications, there are some features that are \emph{critical} in a way that they are required to be weighed strongly in the decision even if they give rise to disparity. Examples of such critical features might be weightlifting ability for a firefighter's job, educational qualification for an academic job, coding skills for a software engineering job, merit and seniority in deciding salary, etc. In an attempt to preserve the importance of the critical features in the decision making, one might choose to exempt the disparity created by them. On the other hand, racial disparity in mortgage lending decisions arising due to zip code (a non-critical feature)~\cite{hinnefeld2018evaluating}, or disparity in promotion/transfer decisions arising from aptitude tests\footnote{In the landmark employment discrimination court-case of Griggs v. Duke Power~\cite{grover1995business}, the US Supreme Court deemed certain aptitude tests as not job-related and hence not business necessities, ruling against the employer.} are examples of non-exempt disparity. In this work, our goal is to formalize and quantify the \emph{non-exempt disparity}, i.e., the part of the disparity that cannot be accounted for by the critical features. This quantification is important for two reasons: (i) it allows one to check if the disparity arose purely due to the critical features (inspired from the ``business necessity defense'' in the disparate impact law, i.e., Title VII of the Civil Rights Act of 1964~\cite{barocas2016big}); and (ii) it enables selective removal of the non-exempt component if desired. 

In this work, we assume that the critical features or business necessities are known (similar to \cite{kamiran2013quantifying,kilbertus2017avoiding}; this discussion is revisited in Section~\ref{sec:conclusion}). We let $X_c$ and $X_g$ denote the critical and the non-critical (or general) features, and $X$ denote the entire set of features. We also denote the protected attribute(s) by $Z$, the true label by $Y$, and the model output by $\hat{Y}$ which is a function of the entire feature vector $X$. While we acknowledge that such categorization of features is application-dependent and might require domain knowledge and ethical evaluation, such exemptions do exist in law. E.g., the US Equal Pay Act~\cite{USEqualPay} exempts for difference in salary based on gender that can be explained by merit and seniority. Similarly, the US employment discrimination law contains a business necessity defense~\cite{grover1995business} where disparity about protected attributes may be exempted if the disparity can be justified as ``necessary to the normal operation of that particular business.'' For example, a standardized coding-test score may be a critical feature in hiring software engineers for a safety-critical application. Similarly, weightlifting ability might be a critical feature in hiring firefighters so that they are able to carry fire victims out of a burning building. The critical feature is therefore required to be weighed strongly in hiring even if it is correlated with some protected attributes. 

\emph{Why should we use the ``general'' features at all for prediction if they are not critical?} General features can improve performance metrics such as accuracy of the model, or even help reduce the candidate pool,~e.g., if 60\% applicants clear a test, but resources are available to interview only 10\%. Not using the general features at all can reduce accuracy, or produce a very large candidate pool. In this work, our proposition is to use both critical and general features in a way that maximizes accuracy (to the extent possible) while preventing non-exempt disparity. For instance (inspired from \cite{barocas2016big}), to choose a ``good'' employee, an employer could evaluate standardized test scores and also reference letters (human-graded performance reviews). All these features are ``job-related'' in that they have statistical correlation with the prediction goal, and can help improve the accuracy. However, test scores, a critical feature, may need to be weighed strongly in the decision, even if they introduce disparity, whereas, reference letters may be used only to the extent that they do not discriminate. 

This work treads a middle ground between two popular measures of fairness that do not use domain knowledge, namely, \emph{statistical parity}~\cite{agarwal2018reductions,ghassami2018fairness,fairMI,dwork2012fairness}, which enforces the criterion $Z\independent\hat{Y}$, and \emph{equalized odds}~\cite{hardt2016equality,ghassami2018fairness,fairMI}, which enforces  $Z\independent\hat{Y}|Y$ (directly or through practical relaxations).  Our selective quantification of non-exempt disparity (using domain knowledge to identify critical features) helps address one of the major criticisms against statistical parity. The criticism is that it can lead to the selection of unqualified members from the protected group~\cite{zemel2013learning,hardt2016equality}, e.g., by disregarding the critical features if they are correlated with the protected attribute $Z$. In fact, in our case study in Section~\ref{sec:case_study}, we observe that the weight of the critical feature is significantly reduced in the decision making when one uses statistical parity as a regularizer with the loss function because the critical feature is correlated with $Z$ (also see Canonical Example~\ref{cexample:exemption} in Section~\ref{subsec:contrast}). On the other hand, equalized odds suffers from label bias~\cite{hinnefeld2018evaluating,hinnefeld2018bias,yeom2018discriminative,kearns2018preventing} because it is based on agreement with the true labels. In fact, we demonstrate (Canonical Example~\ref{cexample:equalized_odds} in Section~\ref{subsec:contrast}) that if the historic labels themselves reinforce disparity from the non-critical features, then even if we obtain a perfect classifier after training on the historic data, which satisfies equalized odds, it can reinforce undesirable non-exempt disparity\footnote{Our quantification does not use the true labels for fairness (unlike equalized odds), addressing the criticism in~\cite{barocas2016big} which says that `` [...] often the best labels for different classifications will be open to debate.''}.

\begin{table*}[t]
	\centering
	\caption{Observational Measures ($M_{NE}$) of Non-Exempt Disparity (Utility and Limitations)} 
	\label{table:examples}
	\begin{tabular}{llccc}
		\toprule
		& Desirable  Properties & $\uni{Z}{\hat{Y} \given X_{c}}$ & $\mut{Z}{\hat{Y} \given X_{c}}$ & $\mut{Z}{\hat{Y} \given X_{c},X'}$ \\
		\midrule
		{\bf 1.} & No counterfactual causal influence from $Z$ to $\hat{Y}\;\Rightarrow\; M_{NE}=0.$ & Yes & Not Always & Not Always  \\
			{\bf 2.} & $M_{NE}$ detects unique information about $Z$ in $\hat{Y}$ not in $X_c$. & Yes & Yes & Not Always \\
		{\bf 3.} & $M_{NE}$ detects non-exempt masked disparity. & No & Masked by  $g(X_{c})$ & Masked by  $g(X_{c},X')$  \\

			{\bf 4.} & $M_{NE}$ equals total disparity if $X_c=\phi$ and $X_g=X$. & No & No & No  \\
							{\bf 5.} & $M_{NE}$ is non-increasing as more features are added to $X_c$ from $X_g$. & Yes & No & No  \\
		{\bf 6.} & $M_{NE}$ is $0$ (complete exemption) if $X_c=X$ and $X_g=\phi$. & Yes & Yes & Yes  \\

		\bottomrule
	\end{tabular}
\end{table*}

\subsection{Contributions}

Our main contribution in this work is the quantification of non-exempt disparity based on a rigorous axiomatic approach. As a first step towards this quantification, we propose an information-theoretic quantification (see Definition~\ref{defn:total_disparity} in Section~\ref{subsec:sys_model}) of the total disparity (exempt and non-exempt) that is $0$ if and only if the model is \emph{counterfactually fair}~\cite{kusner2017counterfactual}. Counterfactual fairness~\cite{kusner2017counterfactual,russell2017worlds} is a causal notion of fairness where the features $X$, the protected attribute $Z$ and the model output $\hat{Y}$ are assumed to be observables in a Structural Causal
Model (SCM) (defined formally in Section~\ref{sec:sys_model}; see Definition~\ref{defn:scm}). The model is deemed \emph{counterfactually fair} if $Z$ has no \emph{counterfactual causal influence} on $\hat{Y}$, i.e., $\hat{Y}$ does not change if we are able to vary $Z$ in the SCM in a manner that other independent latent factors remain constant (defined formally in Section~\ref{sec:sys_model}; see Definition~\ref{defn:cci}). 

Interestingly, note that the total disparity (in a counterfactual sense) may not exhibit itself entirely in the mutual information $\mut{Z}{\hat{Y}}$,  which is the \emph{statistically visible information}\footnote{This is a quantification of disparity inspired from statistical parity which deems a model fair if and only if $\hat{Y}\independent Z$. Note that, $\mut{Z}{\hat{Y}}=0$ if and only if $\hat{Y}\independent Z$.} about $Z$ in $\hat{Y}$, because of ``statistical masking effects'' (also relates to Simpson's paradox~\cite{peters2017elements}). Consider an example inspired from~\cite{datta2017use,kusner2017counterfactual,kearns2018preventing} where a software engineering job advertisement is shown only to a) men with coding skills above a threshold, and b) women with coding skills below a threshold. That is, the decision $\hat{Y}=Z \oplus G$  where $\oplus$ denotes XOR, $G$ is the binary variable denoting whether coding skills are above a threshold (that does not have a causal influence of $Z$ in this example), and $G,Z$ are \iid{} Bern(\nicefrac{1}{2}). This decision is biased against the high-skilled women for whom the ad is relevant, but $\mut{Z}{\hat{Y}}=0$ here, thus failing to capture this bias. Intuitively, our quantification of total disparity also extends the idea of \emph{proxy-use}~\cite{datta2017use} from \emph{white-box models}\footnote{White-box models~\cite{datta2017use} are the type of models where one can clearly explain how they behave, how they produce predictions and what the influencing variables or sub-components of the model are, e.g., decision trees, linear regression, etc.} to black-box models. Proxy-use~\cite{datta2017use} examines ``white-box'' models, i.e., models with clearly defined constituents (e.g., decision trees) and regards a model as having disparity if (i) there is a constituent that has high mutual information about $Z$ (a proxy of $Z$); and (ii) this constituent also causally influences the output $\hat{Y}$ (i.e., varying the constituent while keeping other constituents constant does not change the output). In this work, the total disparity captures the intuitive notion of a virtual constituent or proxy of $Z$ that causally influences the final output $\hat{Y}$ (this intuition is revisited to understand Scenario~\ref{scenario:cf} in Section~\ref{subsec:sys_model}). For instance, a virtual constituent $Z$ is formed in the example of masked disparity in ads that causally influences $\hat{Y}$ even though $\mut{Z}{\hat{Y}}=0$. 

Next, we quantify the \textit{non-exempt} part of this total disparity, i.e., the part that cannot be explained by the critical features $(X_c)$. Building on the extension of proxy-use~\cite{datta2017use} for black-box models as discussed above, we aim to quantify the influence of a  discriminatory virtual constituent or proxy of $Z$, if  formed inside the black-box model, on the model output $\hat{Y}$, and that cannot be attributed entirely to the critical features (this idea is revisited for an intuitive understanding of the canonical examples in Section~\ref{subsec:sys_model}.). To quantify this \textit{non-exempt disparity}, we consider toy examples and thought experiments to first arrive at a set of desirable properties (axioms) that any measure of non-exempt disparity should satisfy, and then provide a measure that satisfies them (see Theorem~\ref{thm:satisfythm}).  These desirable properties can be intuitively described as follows. If the model is counterfactually fair, e.g., if the virtual constituents or proxies of $Z$ cancel each other leading to a final model output that has no counterfactual causal influence of $Z$, then it is desirable that the non-exempt disparity is also $0$. Next, it is desirable that the measure be non-zero if $\hat{Y}$ has any ``unique'' statistically visible information about $Z$ that is not present in $X_{c}$ because then that information content is also attributed to $X_g$. However, because of statistical masking effects, even if this unique information is $0$, there may still be \emph{non-exempt masked disparity} that needs to be captured, e.g., in the aforementioned example of software-engineering-job ads (also revisited in Canonical Example~\ref{cexample:masking} in Section~\ref{subsec:rationale} where we discuss our rationale for the properties).  The next three properties are more intuitive. If all the features are in the non-critical set, then the measure should be equal to the total disparity since no disparity is exempt. For a fixed set of features $X$ and a fixed model, as more features become categorized as critical, the measure of non-exempt disparity should not increase, i.e., it either decreases or stays the same. Ultimately, if all the features are in the critical set $X_{c}$, then we require the measure of non-exempt disparity to be $0$ since then the total disparity is exempt. 

Our proposed measure of non-exempt disparity, that satisfies all these desirable properties, is \emph{counterfactual} in nature, i.e.,
it depends on the true SCM, and hence, is not \emph{observational}\footnote{Observational measures are those that can be estimated from the probability distribution of the data without knowledge of the underlying SCM.} in general. We also show the theoretical impossibility of any observational measure in satisfying all the desirable properties together (see Theorem~\ref{thm:impossibility}). We note that in some applications, counterfactual measures can be realized or approximated with assumptions on the causal model. However, for more general use in practical applications, we also propose several observational
relaxations of our measure that satisfy only some of these properties. 
Nevertheless, we believe that a counterfactual measure and
its properties are crucial in understanding the utility and the
limitations of different observational measures and informing
which measure to choose in practice (summarized in Table~\ref{table:examples}; detailed discussion in Section~\ref{sec:observational}). 

To summarize, our contributions in this work are as follows:\\
\textbf{1.~Quantification of Non-Exempt Disparity}: We propose a novel counterfactual measure of non-exempt disparity that captures the disparity that cannot be explained by the critical features. Our quantification attempts to capture the intuitive notion of whether a discriminatory virtual constituent or proxy~\cite{datta2017use} of $Z$ is formed inside the black-box model that influences the output $\hat{Y}$ and that cannot be attributed entirely to the critical features ($X_c$). We adopt a rigorous axiomatic approach where we first arrive at a set of desirable properties that any measure of non-exempt disparity should satisfy by analyzing several canonical examples (thought experiments). Next, we show that the proposed measure satisfies these properties (see Theorem~\ref{thm:satisfythm}). Our quantification leverages a body of work in information theory called Partial Information Decomposition (PID), as well as, causality.\\
\textbf{2. Overall Decomposition of Total Disparity into Statistically Visible and Masked components}: Our quantification finally leads us to an overall decomposition of the total disparity into four non-negative components, namely, exempt and non-exempt \emph{statistically visible} disparity and exempt and non-exempt \emph{masked} disparity (see Theorem~\ref{thm:measure_decomposition}). The exempt and non-exempt \emph{statistically visible} disparities add up to give $\mut{Z}{\hat{Y}}$ which is the total statistically visible disparity. \\
\textbf{3. An Impossibility Result}: We show that no purely observational measure of non-exempt disparity can satisfy all our desirable properties (see Theorem~\ref{thm:impossibility}). 
\\
\textbf{4. Observational Relaxations}: Relaxing our requirements, we obtain purely observational measures that satisfy some of the desirable properties (summarized in Table~\ref{table:examples}) and then use  them in case studies to demonstrate how to (i) audit existing models; and also (ii) train new models that selectively reduce non-exempt disparity.\\

\noindent \textbf{Our contribution in the context of related works:} Causal approaches for fairness have been explored in~\cite{kusner2017counterfactual,kilbertus2017avoiding,russell2017worlds,chiappa2018path,datta2017use,zhang2018fairness,nabi2018fair}, including impossibility results on purely observational measures \cite{kilbertus2017avoiding,datta2017use}. Our main novelty lies in using a rigorous axiomatic approach based on realistic examples and thought experiments for quantifying non-exempt and exempt  disparity separately, thereby allowing for exemptions due to critical features. The decomposition of total disparity into exempt and non-exempt components is tricky. For instance, following the ideas of path-specific counterfactual fairness~\cite{chiappa2018path}, one might be tempted to examine specific causal paths from $Z$ to $\hat{Y}$ that pass (or do not pass) through $X_{c}$, and deem those influences as the two (exempt and non-exempt) measures. However, we provide a counterexample (see Canonical Example~\ref{cexample:synergy} in Section~\ref{subsec:rationale}) to show that disparity can also arise from synergistic information about $Z$ in both $X_{c}$ and $X_{g}$, that cannot be attributed to any one of them alone, \textit{i.e.}, $\mut{Z}{X_{c}}$ and $\mut{Z}{X_{g}}$ may both be $0$ but $\mut{Z}{X_{c},X_{g}}$ may not be. Purely causal measures (that do not rely on the PID framework) can attribute such disparity entirely to $X_{c}$. We contend that such synergistic information, if influencing the decision, must be included in the \textit{non-exempt} component of disparity because both $X_{c}$ and $X_{g}$ are contributors.  We note that identifying synergy is important: synergy arises frequently in machine-learning and other related applications~\cite{tax2017partial,venkatesh2019information,peters2017elements}.

Some observational measures for quantifying non-exempt disparity have been introduced previously in \cite{kamiran2013quantifying,zliobaite2011handling} where the authors propose a decomposition of statistically visible discrimination (statistical parity) into explainable and non-explainable components (see also subsequent works  \cite{calders2013controlling,corbett2017,interventional_fairness,debiasing,xu2020algorithmic} that build on this idea). They examine the difference in the expected model output ($\hat{Y}$) for candidates of different races/genders ($Z$) after conditioning on specific subsets of features\footnote{Conditional mutual information (conditioned on the critical feature(s)) as a measure of non-exempt disparity has surfaced in \cite{debiasing} with a focus on novel estimators.} (this relates to dependence between $Z$ and $\hat{Y}$ after conditioning on specific features; also referred to as conditional statistical parity~\cite{corbett2017}). In this context, in this work, we provide simple yet relevant counterexamples showing that conditioning may not always faithfully capture non-exempt disparity. E.g., Canonical Example~\ref{cexample:college} in Section~\ref{subsec:rationale}) is deemed \emph{unfair} by conditional mutual information (or conditional statistical parity), but is  \emph{fair} by counterfactual fairness~\cite{kusner2017counterfactual,russell2017worlds}. We use these examples as motivation to  decompose conditional mutual information into unique and synergistic information using PID, separating two kinds of ``statistical dependence'' which conditioning alone fails to do (see Section~\ref{subsec:background}). We refer to Section~\ref{subsec:contrast} for more detailed discussion on existing measures that have some provision for exemption, namely, conditional statistical parity~\cite{corbett2017,debiasing}, justifiable fairness~\cite{interventional_fairness}, as well as a related causal measure of path-specific counterfactual fairness~\cite{chiappa2018path}. Our problem also differs from \emph{sub-group fairness}~\cite{kearns2018preventing} where the sub-populations in consideration are based on the protected attributes alone, e.g., $Z=(Z_1,Z_2)$ with $Z_1$ being gender, and $Z_2$ being race, and does not consider exemptions with respect to the other (non-protected) attributes. Another interesting related work is \cite{galhotra2020fair} which approaches the problem of fairness from the perspective of feature selection while allowing for a set of admissible attributes/features. In \cite{galhotra2020fair}, the authors propose conditional independence tests (observational) with respect to the admissible attributes for feature selection while using group testing to improve the complexity of the technique, and demonstrate that the proposed technique satisfies the interventional fairness definition in \cite{interventional_fairness}.

We also note that the idea of using correlation-based observational approximations of disparity (e.g., correlation between $Z$ and $\hat{Y}$ to represent statistical parity) as a regularizer during training has been proposed earlier \cite{kamishima2012fairness}. In this context, our main contribution here is on first arriving at a measure of non-exempt disparity (that happens to be non-observational), and then proposing 3 observational measures for applications in \emph{both} auditing existing models and training new models with reduced non-exempt disparity. For auditing, we use alternate non-correlation-based estimators for unique information, mutual information, and conditional mutual information from the \texttt{dit} package~\cite{dit}. For training, we rely on simplistic correlation-based approximations for mutual information and conditional mutual information along the lines of \cite{kamishima2012fairness} for ease of computation. For unique information, we introduce novel correlation-based regularizers for training in Section~\ref{sec:case_study}, leveraging a Gaussian approximation for PID \cite{barrett2015exploration}. 


\subsection{Paper Outline}
The rest of the paper is organized as follows. Section~\ref{sec:sys_model} introduces the background, system model and assumptions underlying our problem formulation, i.e., how to quantify the non-exempt disparity. Section~\ref{subsec:main_results} first states all the desirable properties that a measure of non-exempt disparity should satisfy, and then introduces our proposed counterfactual measure that satisfies all of them (Theorem~\ref{thm:satisfythm} in Section~\ref{subsec:main_results}). This is followed by a rationale behind the desirable properties through canonical examples and thought experiments in Section~\ref{subsec:rationale}. We also discuss the utility and limitations of some existing measures, namely, path-specific counterfactual fairness~\cite{chiappa2018path}, conditional statistical parity~\cite{corbett2017}, and justifiable fairness~\cite{interventional_fairness} in Section~\ref{subsec:contrast}. Next, Section~\ref{sec:measure_decomposition} provides insights on the overall decomposition of the total disparity (in a counterfactual sense) into exempt and non-exempt components, with each of them being further decomposed into \emph{statistically visible} and \emph{masked} components (Theorem~\ref{thm:measure_decomposition} in Section~\ref{sec:measure_decomposition}). Section~\ref{sec:impossibility} provides an impossibility result on observational measures, stating that no observational measure can satisfy all of the desirable properties. Nonetheless, since counterfactual measures are often difficult to realize in practice, we propose several observational relaxations of our proposed counterfactual measure in Section~\ref{sec:observational} (that only satisfy some of the desirable properties), and discuss their utility and limitations. Next, in Section~\ref{sec:case_study}, we use our observational measures to conduct case studies on both artificial and real datasets to demonstrate practical application in training. Finally, we conclude with a discussion in Section~\ref{sec:conclusion}.

\section{Preliminaries}
\label{sec:sys_model}

Here, we first provide a brief background on Partial Information Decomposition (PID)  in Section~\ref{subsec:background} to help follow the paper. Appendix~\ref{app:pid_properties} provides more details on the specific properties used in the proofs. Next, we introduce our system model and assumptions in Section~\ref{subsec:sys_model}. We use the following notations: (i) $X=(X_1,X_2,\ldots,X_n)$ denotes a tuple~\cite{tuple}, i.e., an ordered set of elements $X_1,X_2,\ldots,X_n$; (ii) $\phi$ denotes the empty tuple (no elements); (iii) For tuple with a single element, the bracket is omitted for brevity, i.e., $(X_1)=X_1$; (iv) $(X,A)$ is equivalent to the new tuple $(X_1,X_2,\ldots,X_n,A)$ formed by appending the element $A$ at the end of tuple $X$; (v) $X_1\in X$ means $X_1$ is an element of tuple $X$; (vi) $S \subseteq X$ means the set of elements in tuple $S$ form a subset of the set of elements in tuple $X$; and (vii) $X\backslash X_2$ denotes a new tuple formed by removing element $X_2$ from $X$ without changing the order of other elements, i.e., $(X_1,X_3,X_4,\ldots,X_n)$.

\subsection{Background on Partial Information Decomposition (PID)} 
\label{subsec:background}


\begin{figure*}
\begin{subfigure}[b]{0.43\linewidth}
\centering
{\centering \includegraphics[height=3.8cm]{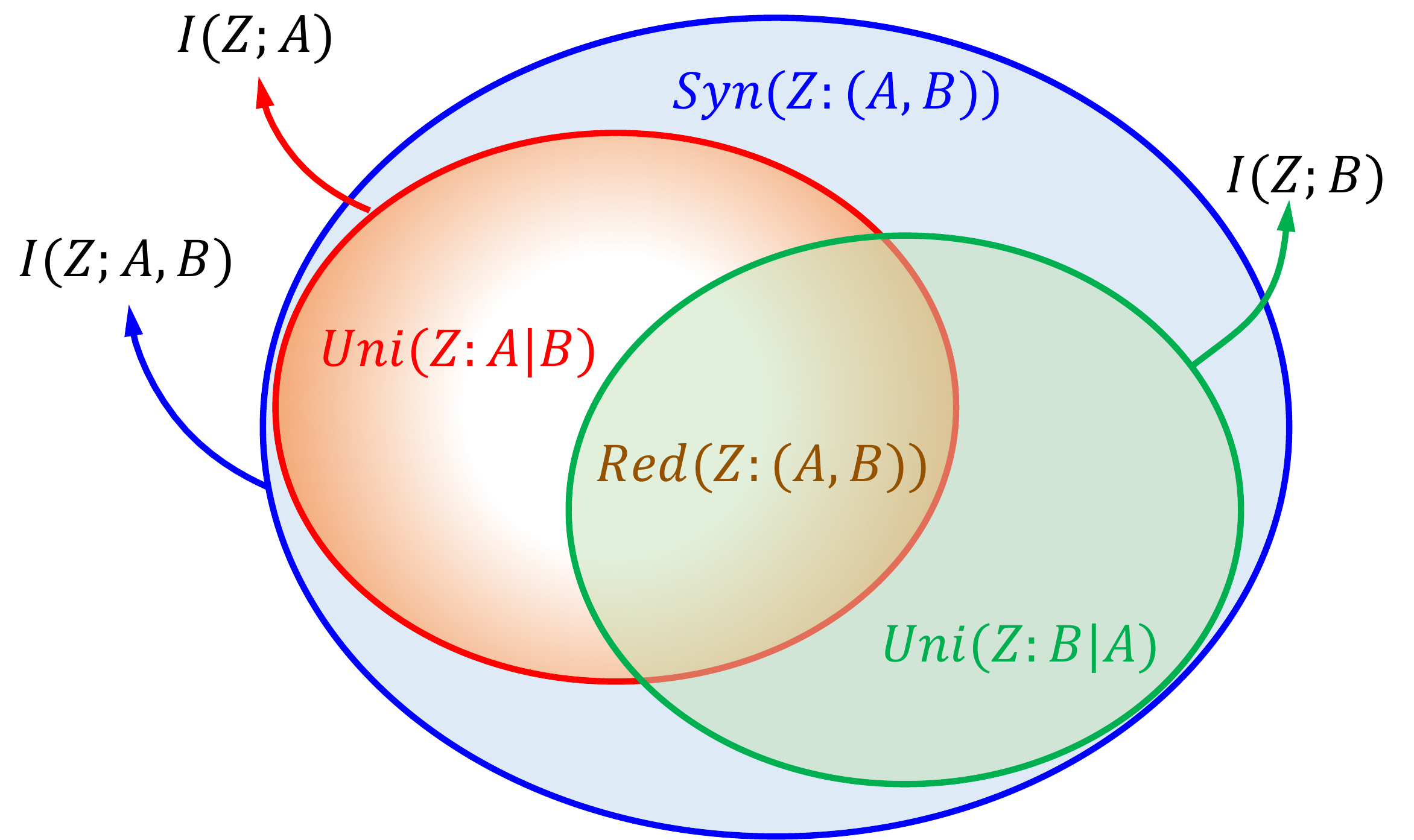}}
\caption{Venn diagram showing PID of $\mathrm{I}(Z;(A,B))$ }\label{fig:pid1}
\end{subfigure}
\begin{subfigure}[b]{0.56\linewidth}
\centering
{\centering \includegraphics[height=3.8cm]{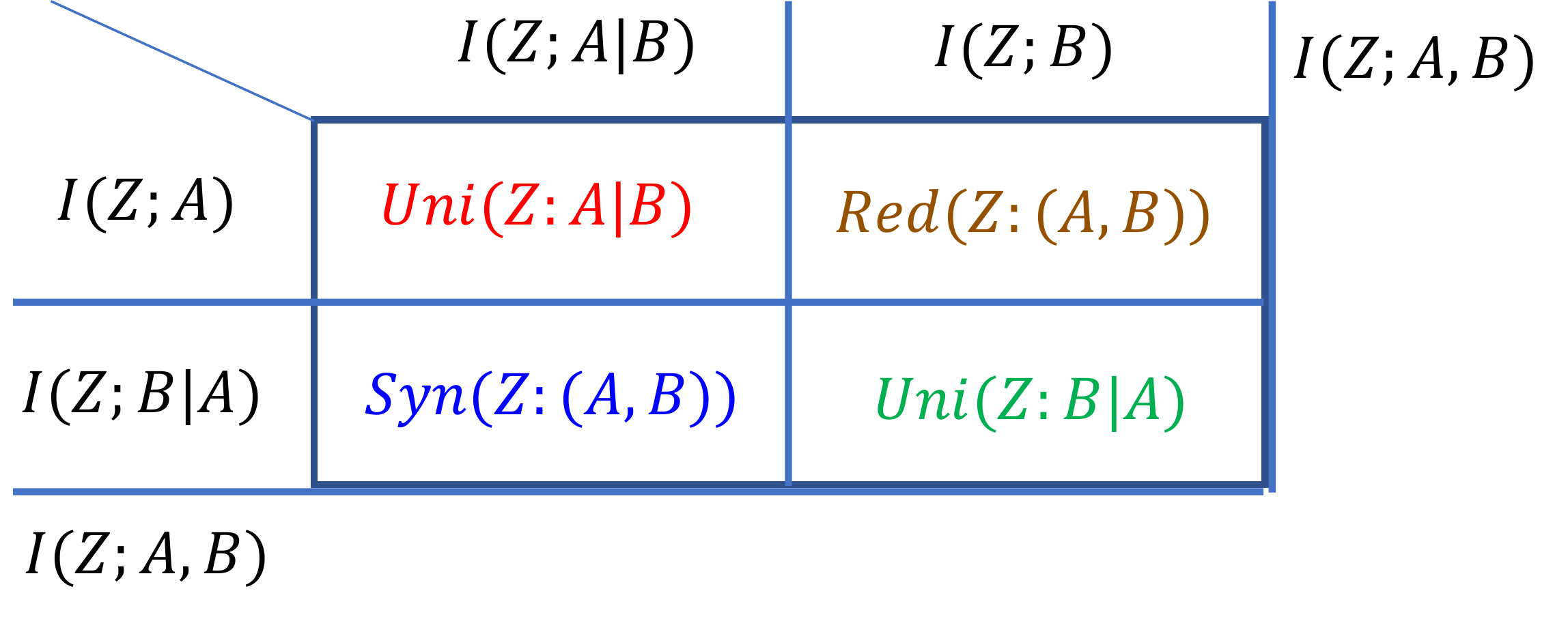}}
\caption{Tabular Representation of PID of $\mathrm{I}(Z;(A,B))$ }\label{fig:pid2}
\end{subfigure}
\caption{Mutual information $\mathrm{I}(Z;(A,B))$ is decomposed into $4$ non-negative terms, namely, $\uni{Z}{A| B}$, $\uni{Z}{B| A}$, $\rd{Z}{(A, B)}$ and $\syn{Z} {(A, B)}$. Also note that, $\mathrm{I}(Z;(A,B))=\mathrm{I}(Z;B)+\mathrm{I}(Z;A\given B),$ each of which is in turn a sum of two PID terms. $\rd{Z}{(A, B)}$ is the sub-volume between $\mathrm{I}(Z;A)$ and $\mathrm{I}(Z;B)$, and $\uni{Z}{A| B}$ is the sub-volume between $\mathrm{I}(Z;A\given B)$ and $\mathrm{I}(Z;A)$.  }\label{fig:pid}
\end{figure*}

The PID framework~\cite{bertschinger2014quantifying,williams2010nonnegative,griffith2014quantifying} decomposes the mutual information $\mut{Z}{(A,B)}$ about a random variable $Z$ contained in the tuple $(A,B)$ into four \emph{non-negative} terms as follows (also see Fig.~\ref{fig:pid}):
\begin{align}
 \mut{Z}{(A,B)}  = \uni{Z}{A| B} + \uni{Z}{B| A}  + \rd{Z}{(A, B)} + \syn{Z}{(A, B)}. \label{eq:pid1}
\end{align}
Here, $\uni{Z}{A| B}$ denotes the unique information about $Z$ that is present only in $A$ and not in $B$. Likewise, $\uni{Z}{B| A}$ is the unique information about $Z$ that is present only in $B$ and not in $A$. The term $\rd{Z}{(A, B)}$ denotes  the redundant information about $Z$ that is present in both $A$ and $B$, and $\syn{Z}{(A, B)}$ denotes the synergistic information not present in either of $A$ or $B$ individually, but present jointly in $(A,B)$. \emph{All four of these terms are non-negative. Also notice that, $\rd{Z}{(A, B)}$ and $\syn{Z}{(A, B)}$ are symmetric in $A$ and $B$.}  Before defining these PID terms formally, let us understand them through an intuitive scenario.
\begin{scenario}[Understanding Partial Information Decomposition]
Let $Z=(Z_1,Z_2,Z_3)$ with $Z_1,Z_2,Z_3\sim$ \iid{} Bern(\nicefrac{1}{2}). Let $A=(Z_1,Z_2,Z_3\oplus N)$, $B=(Z_2,N)$, $N\sim $  Bern(\nicefrac{1}{2}) is independent of $Z$. Here,  $\mathrm{I}(Z; (A,B))=3$ bits. 
\end{scenario}

The unique information about $Z$ that is contained only in $A$ and not in $B$ is effectively contained in $Z_1$ and is given by $\uni{Z}{A| B} = \mut{Z}{Z_1} = 1$ bit. The redundant information about $Z$ that is contained in both $A$ and $B$ is effectively contained in $Z_2$ and is given by $\mathrm{Red}(Z: (A, B))=\mathrm{I}(Z;Z_2)=1$ bit. Lastly, the synergistic information about $Z$ that is not contained in either $A$ or $B$ alone, but is contained in both of them together is effectively contained in the tuple $(Z_3\oplus N,N)$, and is given by $\syn{Z}{(A,B)} = \mut{Z}{(Z_3\oplus N,N)}=1 $ bit. This accounts for the $3$ bits in $\mut{Z}{ (A,B)}$. Here, $B$ does not have any unique information about $Z$ that is not contained in $A$, i.e., $\uni{Z}{B| A} =0.$

Irrespective of the formal definition of these individual terms, the following identities also hold (see Fig.~\ref{fig:pid2}):
\begin{align}
&\mut{Z}{A}=\uni{Z}{A| B} + \rd{Z}{(A, B)}. \label{eq:pid2}\\
&\mut{Z}{A \given B}=\uni{Z}{A| B} + \syn{Z}{(A, B)}.  \label{eq:pid3} 
\end{align}
\begin{rem}[An Interpretation of PID as Information-Theoretic Sub-Volumes] Equations~\eqref{eq:pid1}, \eqref{eq:pid2} and \eqref{eq:pid3} have been represented in a tabular fashion in Fig.~\ref{fig:pid2}. Notice that, $\uni{Z}{A| B}$ can be viewed as the information-theoretic sub-volume of the intersection between $\mut{Z}{A}$ and $\mut{Z}{A \given B}$. Similarly, $\rd{Z}{(A, B)}$ is the sub-volume between $\mut{Z}{A}$ and $\mut{Z}{B}$.
\label{rem:sub-volume}
\end{rem}

These equations also demonstrate that $\uni{Z}{A| B}$ and  $\rd{Z}{(A, B)}$ are the information contents that exhibit themselves in $\mut{Z}{A}$ which is the statistically visible information content about $Z$ present in $A$. Because both these PID terms are non-negative, if any one of them is non-zero, we will have $\mut{Z}{A}>0$. Similarly, $\uni{Z}{B| A}$ and  $\rd{Z}{(A, B)}$ also exhibit themselves in $\mut{Z}{B}$. On the other hand, $\syn{Z}{(A, B)}$ is the information content that does not exhibit itself in $\mut{Z}{A}$ or $\mut{Z}{B}$ individually, i.e., these terms can still be $0$ even if $\syn{Z}{(A, B)}>0$. But, $\syn{Z}{(A, B)}$ exhibits itself in $\mut{Z}{(A, B)}$. Notice that,
\begin{align}
\mut{Z}{(A, B)} &= \underbrace{\uni{Z}{A| B}  + \rd{Z}{(A, B)}}_{\mut{Z}{A}}  + \underbrace{\uni{Z}{B| A} + \syn{Z}{(A, B)}}_{\mut{Z}{B\given A}} \\
& =\underbrace{\uni{Z}{B| A}  + \rd{Z}{(A, B)}}_{\mut{Z}{B}}  + \underbrace{\uni{Z}{A| B} + \syn{Z}{(A, B)}}_{\mut{Z}{A\given B}}.
\end{align}

Given three independent equations \eqref{eq:pid1}, \eqref{eq:pid2} and \eqref{eq:pid3} in four unknowns (the four PID terms), defining any one of the terms (e.g., $\uni{Z}{A| B}$) is sufficient to obtain the other three. For completeness, we include the definition of unique information from \cite{bertschinger2014quantifying} (that also allows for estimation via convex optimization~\cite{banerjee2018computing}) with the specific properties used in the proofs in Appendix~\ref{app:pid_properties}. To follow the paper, only an intuitive understanding is sufficient.
\begin{defn}[Unique Information~\cite{bertschinger2014quantifying}] Let $\Delta$ be the set of all joint distributions on $(Z,A,B)$ and $\Delta_p$ be the set of joint distributions with the same marginals on $(Z,A)$ and $(Z,B)$ as their true distribution, \textit{i.e.},
$\Delta_p=\{ Q \in \Delta : q(z,a){=}\Pr(Z{=}z,A{=}a) \text{ and } q(z,b){=}\Pr(Z{=}z,B{=}b) \}. $  Then, 
$\uni{Z}{A| B}=\min_{Q \in \Delta_p} \mathrm{I}_{Q}(Z;A\given B),$ where $\mathrm{I}_{Q}(Z;A\given B)$ is the conditional mutual information when $(Z,A,B)$ have joint distribution $Q$. \label{defn:uni} 
\end{defn} 
\noindent The key intuition behind this definition is that the unique information should only depend on the marginal distribution of the pairs $(Z, A)$ and $(Z, B)$. This is motivated from an \textbf{operational} perspective that if $A$ has unique information about $Z$ (with respect to $B$), then there must be a situation where one can predict $Z$ better using $A$ than $B$ (more details in \cite[Section 2]{bertschinger2014quantifying}). Therefore, all the joint distributions in the set $\Delta_p$ with the same marginals essentially have the same unique information, and the distribution $Q^*$ that minimizes $\mathrm{I}_{Q}(Z;A\given B)$ is the joint distribution that has no synergistic information leading to $\mathrm{I}_{Q^*}(Z;A\given B)= \uni{Z}{A| B}$. Definition~\ref{defn:uni} also defines $\rd{Z}{(A, B)} $  and $\syn{Z} {(A, B)}$ using \eqref{eq:pid2} and \eqref{eq:pid3}.

\subsection{System Model and Assumptions}
\label{subsec:sys_model}

Here, we introduce our system model and assumptions. We start with an introduction to Structural Causal Model (SCM).

\begin{defn}[Structural Causal Model: $\mathrm{SCM}(U,V,\mathcal{F})$ \cite{peters2017elements}] A structural causal model $(U,V,\mathcal{F})$ consists of a set of latent (unobserved) and mutually independent variables $U$ which are not caused by any variable in the set of observable variables $V$, and a collection of deterministic functions (structural assignments) $\mathcal{F}=(F_1,F_2,\ldots )$, one for each $V_i \in V$, such that: $V_i=F_i(V_{pa_i},U_i).$ Here $V_{pa_i} \subseteq V\backslash V_i$ are the parents of $V_i$, and $U_i\subseteq U$. The \textit{structural assignment graph} of $\mathrm{SCM}(U,V,\mathcal{F})$ has one vertex for each $V_i$, and directed edges  to $V_i$ from each parent in $V_{pa_i}$, and is always a directed acyclic graph.
\label{defn:scm}
\end{defn}


\begin{figure}
	\centering
	\includegraphics[height=3cm]{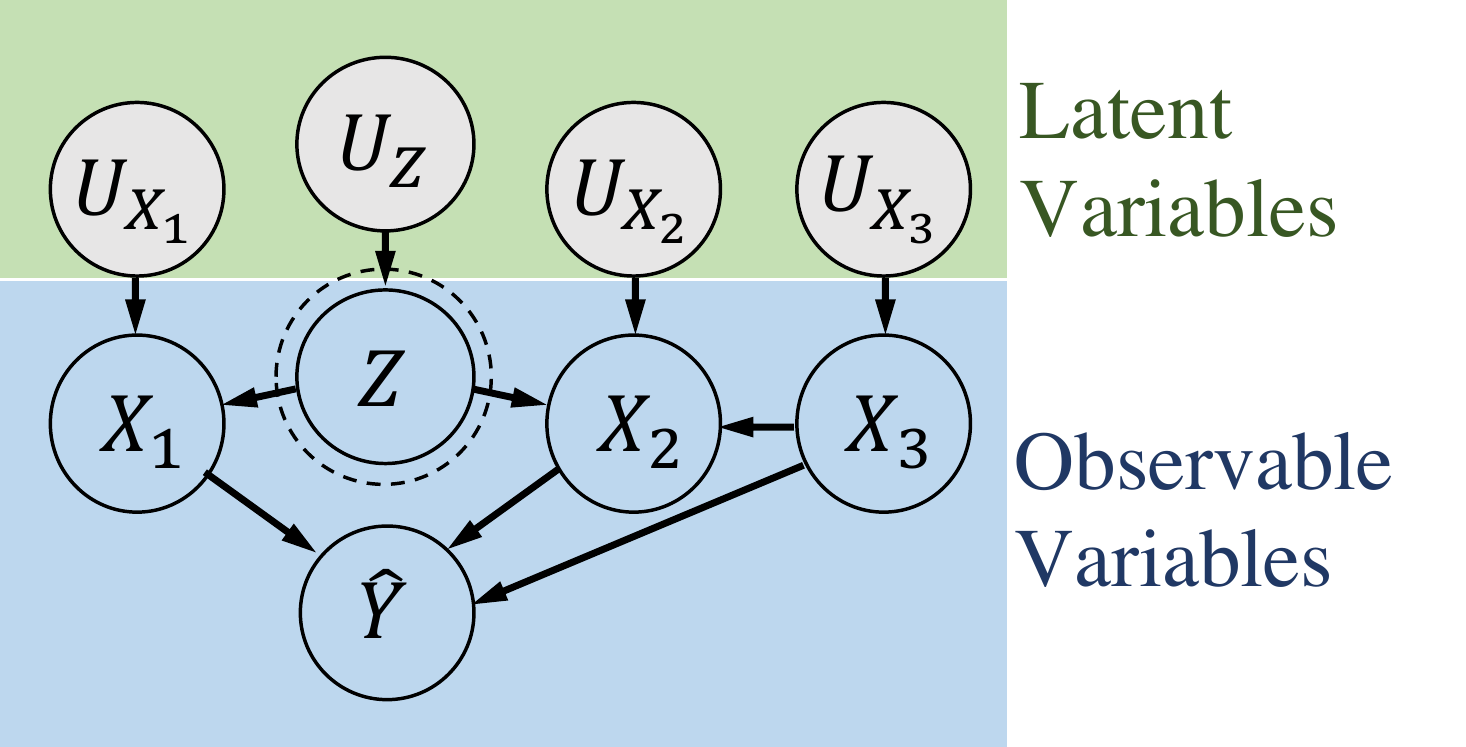}
	\caption{An SCM with protected attribute $Z$, features $X=(X_1, X_2,X_3)$, and output $\hat{Y}$. Here $X$ and $\hat{Y}$ are the observables, and $U_Z$ and $U_X=(U_{X_1},U_{X_2},U_{X_3})$ are the latent social factors. $Z$ does not have any parents in the SCM and $\hat{Y}$ is completely determined by $X=(X_1, X_2,X_3)$. \label{fig:background}}
\end{figure}

\noindent \textbf{Our System Model:} For our problem, consistent with several other works on fairness~\cite{chiappa2018path,kusner2017counterfactual,kilbertus2017avoiding}, the latent variables $U$ represent possibly unknown social factors. The observables $V$ consist of the protected attributes $Z$, the features $X$ and the output $\hat{Y}$ (see Fig.~\ref{fig:background}).  For simplicity, we assume ancestral closure of the protected attributes, \textit{i.e.},  the parents of any $V_i \in Z$ also lie in $Z$ and hence $Z$ is not caused by any of the features in $X$ ($V_i \in Z$ are source nodes in the graph). Therefore, $Z=f_z(U_Z)$ for $U_Z\subseteq U$. Any feature $X_j$ in $X$ is a function of its corresponding latent variable ($U_{X_j}$) and its parents, which are again functions of their own latent variables and parents. Therefore, each $X_j$ can also be written as $f_j(Z,U_{X})$ for some deterministic $f_j(\cdot)$, where $U_X=U\backslash U_Z$ denotes the latent factors in $U$ that do not cause $Z$ (see a formal proof in \cite[Proposition 6.3]{peters2017elements}). Here, $f_j(\cdot)$ may be constant in some of its arguments. This claim holds because the underlying graph is acyclic, and hence the structural assignments of the ancestors of $X_j$ can be substituted  recursively into one another until all observables except $Z$ are substituted by latent variables. Also note that, $Z\independent U_{X}$. A model takes $X$ (which consists of critical features $X_c$ and general features $X_g$) as its input and produces an output $\hat{Y}$ which is a deterministic function of $X$, i.e., $\hat{Y}=r(X)$ where $X$ is itself a deterministic function of $(Z,U_{X})$. Therefore,  $\hat{Y}=h(Z,U_{X})$ for some deterministic function $h(\cdot)$.

Next, we introduce the concept of Counterfactual Causal Influence (CCI) (\cite{kusner2017counterfactual,russell2017worlds,breiman2001random,datta2016algorithmic,koh2017understanding,adler2018auditing,henelius2014peek}), which will help us understand the well-known causal definition of fairness called \emph{counterfactual fairness}~\cite{kusner2017counterfactual}.

\begin{defn}[Counterfactual Causal Influence: $\ci{Z}{\hat{Y}}$] \label{defn:cci} Consider the aforementioned system model. Let $\hat{Y}=h(Z,U_{X})$ for some deterministic function $h(\cdot)$ where $U_{X}$ are latent variables that do not cause $Z$ in the true SCM. Then, 
\begin{equation}
\mathrm{CCI}(Z \rightarrow \hat{Y}) = \E{Z,Z',U_{X}}{| h(Z,U_{X})-h(Z',U_{X})| }  \text{ where }Z',Z \text{ are } \iid{}
\end{equation}
\end{defn}


\noindent Counterfactual causal influence quantifies the change in $\hat{Y}=h(Z,U_{X})$ if we only vary $Z$ while keeping the other latent factors ($U_X$) unchanged. A model is said to satisfy \emph{counterfactual fairness}~\cite{kusner2017counterfactual,russell2017worlds} if and only if the output $\hat{Y}$ has no counterfactual causal influence of $Z$ (we formally derive that $\mathrm{CCI}(Z \rightarrow \hat{Y})=0$ is equivalent to counterfactual fairness~\cite{kusner2017counterfactual} in Lemma~\ref{lem:cci_imply_cf} in Appendix~\ref{app:cci_additional}). What this means is that a model is \emph{counterfactually fair} if and only if the output $\hat{Y}=h(Z,U_{X})$ does not change with $Z$ while keeping the other latent factors ($U_X$) unchanged. It captures the intuitive notion that no virtual constituent or proxy of $Z$ influences the output (inspired from the work on proxy-use~\cite{datta2017use}). In other words, $\hat{Y}\independent Z|U_X$ (proved in Lemma~\ref{lem:cci}), i.e.,
\begin{equation}
\Pr(\hat{Y}=y|Z=z,U_X=u_x) =\Pr(\hat{Y}=y|Z=z',U_X=u_x)  \ \forall z,z',y,u_x.
\end{equation}

\noindent This notion of fairness also leads us to propose an information-theoretic quantification of total disparity (exempt and non-exempt) that is $0$ if and only if the counterfactual causal influence of $Z$ on $\hat{Y}$ is $0$ (equivalence is demonstrated in Lemma~\ref{lem:cci} with the proof in Appendix~\ref{app:cci}).

\begin{defn}[Total Disparity]
The total disparity in a model is defined as $\mathrm{I}(Z; (\hat{Y},U_X) )$. \label{defn:total_disparity}
\end{defn}

Notice that, 
\begin{equation}
\mathrm{I}(Z; (\hat{Y},U_X) )= \mathrm{I}(Z; \hat{Y}|U_X) + \underbrace{ \mathrm{I}(Z;U_X)}_{=0 \text{ since }Z \independent U_X}  = \mathrm{I}(Z; \hat{Y}|U_X).
\end{equation}

\begin{restatable}[Equivalences of CCI]{lem}{cci} \label{lem:cci} Consider the aforementioned system model. Let $\hat{Y}=h(Z,U_{X})$ for some deterministic function $h(\cdot)$ and $Z \independent U_X$. Then, $\mathrm{CCI}(Z \rightarrow \hat{Y}) = 0$ if and only if $ \mathrm{I}(Z; (\hat{Y},U_X) ) = 0 $.
\end{restatable}

\begin{figure*}
\centering
\begin{subfigure}[b]{0.30\linewidth}
\centering
\includegraphics[height=3.5cm]{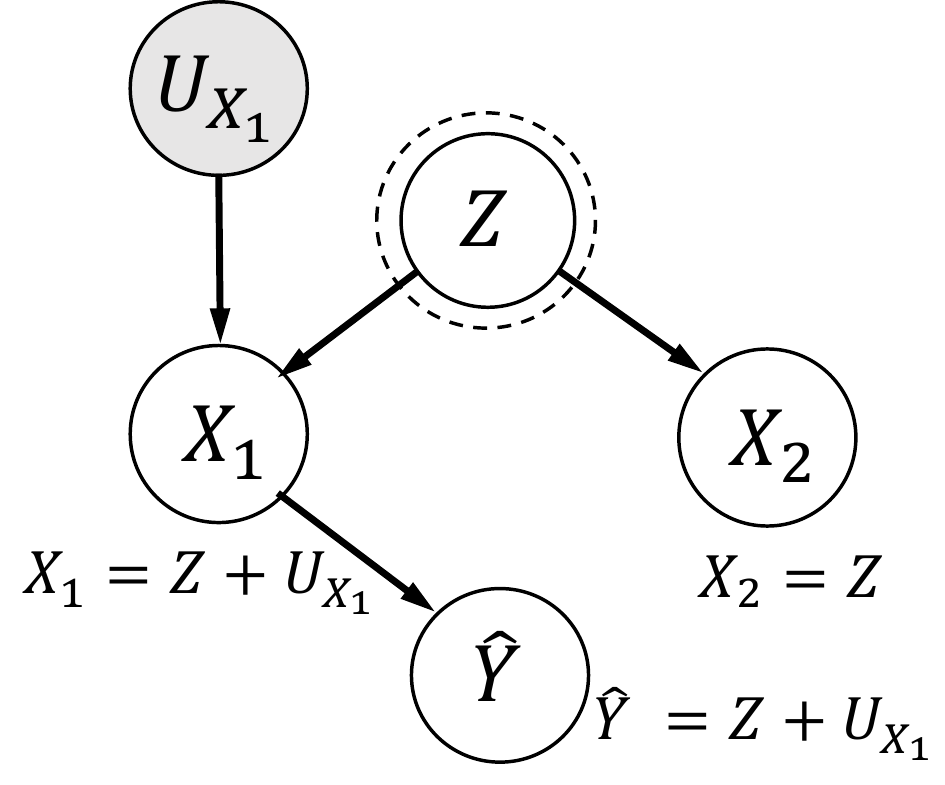}
\caption{Model is not counterfactually fair as $\hat{Y}$ has counterfactual causal influence of $Z$.}\label{fig:cf1}
\end{subfigure}
\hspace{0.1cm}
\begin{subfigure}[b]{0.30\linewidth}
\centering
\includegraphics[height=3.5cm]{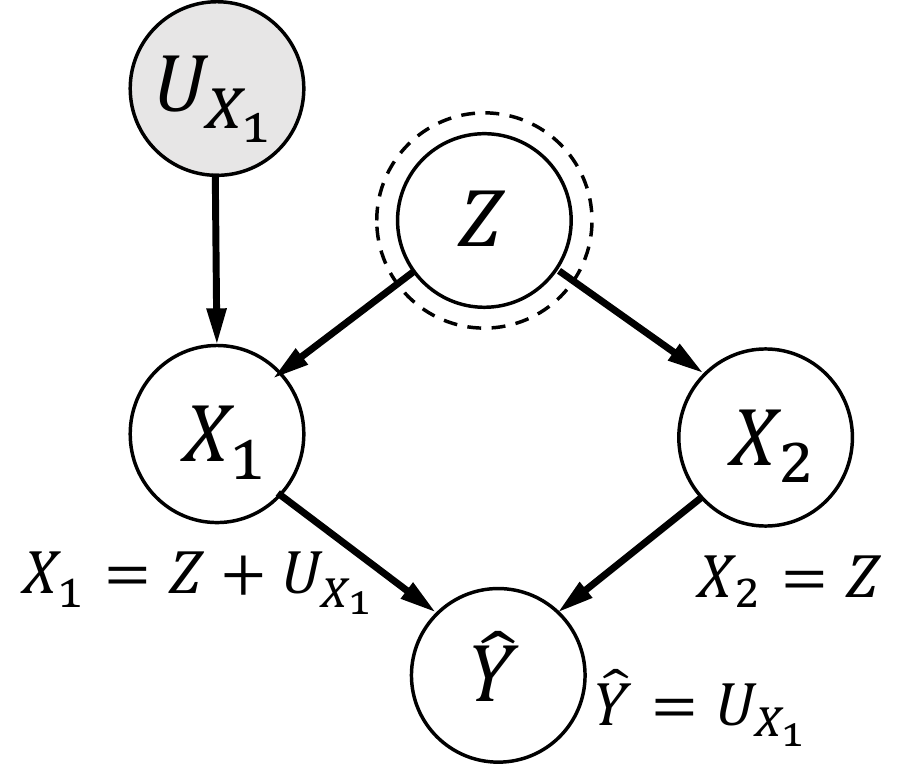}
\caption{Model is counterfactually fair after cancelling out the influence of $Z$ from $X_1$.}\label{fig:cf2}
\end{subfigure}
\hspace{0.1cm}
\begin{subfigure}[b]{0.30\linewidth}
\centering
\includegraphics[height=3.5cm]{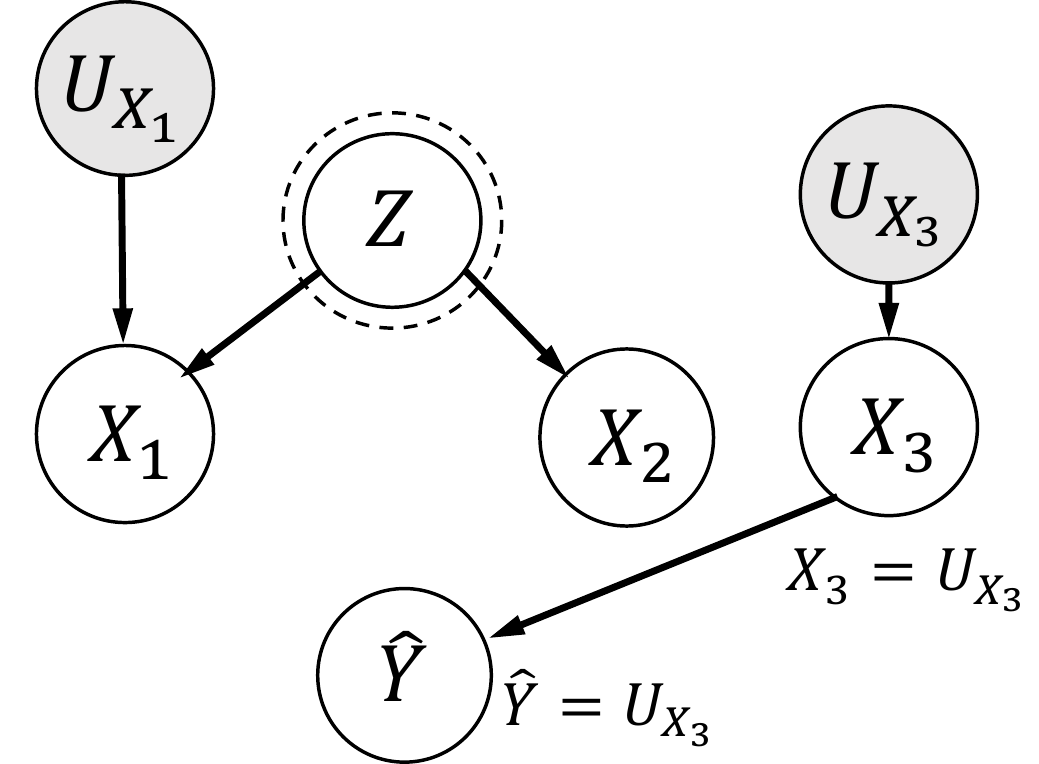}
\caption{Model is counterfactually fair even though it uses an entirely unrelated feature.}\label{fig:cf3}
\end{subfigure}
\caption{Illustration of Scenario~\ref{scenario:cf} for understanding the concept of counterfactual fairness: Different models are used to make hiring decisions on data corresponding to the same SCM with $Z$ denoting the protected attribute, $U_{X_1}$ denoting inner ability, $X_1=Z+U_{X_1}$ denoting interview score, and $X_3$ denoting an alternate feature, e.g., location. \label{fig:cf}}
\end{figure*}

\begin{rem}[Advantage of our Information-Theoretic Quantification]
\label{rem:why_info_theory}One might wonder why such an information-theoretic quantification of counterfactual causal influence (or, total disparity) is necessary. The information-theoretic quantification of total disparity enables analytical decomposition into exempt and non-exempt components that better satisfy our intuitive understanding. Our non-exempt disparity intuitively attempts to capture whether discriminatory proxies are formed inside the black-box model that cannot be entirely attributed to the critical features $X_c$.
The decomposition of counterfactual causal influence (Definition~\ref{defn:cci}) into exempt and non-exempt components is not straightforward. For instance, following the ideas of path-specific counterfactual fairness~\cite{chiappa2018path}, one might be tempted to examine specific causal paths from $Z$ to $\hat{Y}$ that pass (or do not pass) through $X_{c}$, and deem those influences as the two measures. However, as the PID literature notes, disparity can also arise from synergistic information about $Z$ in both $X_{c}$ and $X_{g}$, that cannot be attributed to any one of them alone, \textit{i.e.}, $\mut{Z}{X_{c}}$ and $\mut{Z}{X_{g}}$ may both be $0$ but $\mut{Z}{X_{c},X_{g}}$ may not be (see Canonical Example~\ref{cexample:synergy}). Purely causal measures can attribute such disparity entirely to $X_{c}$. We contend that such synergistic information, if influencing the decision, must be included in the \textit{non-exempt} component of disparity because both $X_{c}$ and $X_{g}$ are contributors to the proxy.  Information-theoretic equivalences of other existing notions of fairness, e.g., statistical parity, equalized odds, etc. have also been used in the broader literature on fairness~ \cite{ghassami2018fairness,liao2019robustness,kamishima2012fairness,calmon2017optimized,fairMI,xu2020algorithmic}.
\end{rem}

For a better understanding of counterfactual fairness, we now consider an intuitive scenario (inspired from \cite{kusner2017counterfactual}).
\begin{scenario}[Understanding Counterfactual Fairness] Suppose a company makes its decisions about hiring based on a feature $X_1$ which denotes an interview score. In the SCM, this feature $X_1=Z+U_{X_1}$ where $Z$ denotes the protected attribute and $U_{X_1}$ denotes the inner ability which is independent of $Z$. An output $\hat{Y}=X_1$ is not counterfactually fair because it has counterfactual causal influence of the protected attribute $Z$ (Fig.~\ref{fig:cf1}). The total disparity $\mut{Z}{(\hat{Y},U_X)}$ is also non-zero, capturing the intuitive notion that a proxy of $Z$ influences the output. On the other hand, suppose the model now uses another feature $X_2=Z$ and produces the output $\hat{Y}=X_1-X_2=U_{X_1}$. This model is now deemed counterfactually fair (Fig.~\ref{fig:cf2}), and its total disparity $\mut{Z}{(\hat{Y},U_X)}$ is zero. No proxy of $Z$ influences the output any longer. \label{scenario:cf} 
\end{scenario}

\begin{rem}[Accuracy vs Counterfactual Fairness] The goals of fairness and accuracy on a given dataset are not always aligned~\cite{menon2018cost,dutta2019information}. For instance, suppose the model in Scenario~\ref{scenario:cf} takes decisions only based on a new feature $X_3=U_{X_3}$ that is derived entirely from some latent factor that is unrelated with the ability to perform the job (see Fig.~\ref{fig:cf3}). Or, even worse, suppose a model is hiring based on a random coin flip. Such a model may be highly inaccurate and absurd but it is still counterfactually fair because it has no counterfactual causal influence of $Z$. In this work, we will assume that a model has absolutely no disparity (exempt or non-exempt) if and only if there is no counterfactual causal influence of $Z$ on $\hat{Y}$. We will also run into some toy examples that might have lower accuracy, but from a counterfactual-fairness-point-of-view, it will be desirable that they are deemed fair if there is no counterfactual causal influence of $Z$.
\end{rem}

\noindent Next, we propose two definitions, namely, statistically visible disparity and masked disparity. Statistically visible disparity is an information-theoretic quantification inspired from a well-known observational definition of fairness called \emph{statistical parity}~\cite{agarwal2018reductions}. 

\begin{defn}[Statistically Visible Disparity] The statistically visible disparity in a model is defined as $\mut{Z}{\hat{Y}}.$
\label{defn:visible_disparity}
\end{defn}

\noindent Statistical parity deems a model fair if and only if $Z\independent \hat{Y}$, i.e.,
$$  \Pr(\hat{Y}=y|Z=z)=\Pr(\hat{Y}=y|Z=z') \ \ \forall y,z,z'. $$ Thus, a model is said to be fair by statistical parity if and only if its statistically visible disparity $\mut{Z}{\hat{Y}}=0.$

\begin{rem} [Statistical Parity vs Counterfactual Fairness] Statistical parity (or independence) does not imply absence of causal effects. E.g., consider $\hat{Y}=Z \oplus U_{X}$ where $Z,U_{X}\sim$ \iid{} Bern(\nicefrac{1}{2}). Here,  $\hat{Y} \independent Z$, but $Z$ still has a causal effect on $\hat{Y}$. If we vary $Z$ keeping all other sources of randomness in $\hat{Y}$ constant (i.e., fixing $U_{X}=u_{x}$), then $\hat{Y}$ also varies. This is, in fact, an example of \textit{masked disparity}, where  $\mathrm{I}(Z;\hat{Y})=0$, but $Z$ has counterfactual causal influence on $\hat{Y}$. 
\label{rem:masked_cci}
\end{rem}

\begin{defn}[Masked Disparity] The masked disparity in a model is defined as $\mathrm{I}(Z;(\hat{Y},U_X))-\mathrm{I}(Z;\hat{Y}).$
\label{defn:masked_disparity}
\end{defn}

The masked disparity is the difference between the total disparity and the statistically visible disparity. Notice that, $\mathrm{I}(Z; \hat{Y},U_X)  - \mathrm{I}(Z;\hat{Y})=\mut{Z}{U_X\given  \hat{Y}}$, implying that masked disparity is non-negative. We will revisit masked disparity in Section~\ref{sec:measure_decomposition}.

\begin{table*}
\centering
\caption{Summary of Notations}
	\begin{tabular}{lll}
		\toprule
Symbol	  &   Description  & Observable or Not \\
		\midrule
		$X_c$ & Tuple of Critical features & Observable  \\
		$X_g$ & Tuple of Non-critical or general features & Observable \\
$X$	 & Tuple of all input features (critical and general) & Observable  \\	
$Z$   &  Protected attribute (s) &Observable   \\
$U_X$ (Note that, $Z \independent U_X$)& Tuple of latent social factors that do not cause $Z$ & Not observable in general\\
$\hat{Y}=r(X)=h(Z,U_X)$  & Model output & Observable \\
		\bottomrule
	\end{tabular}
	\label{table:notations}
\end{table*}

\noindent \textbf{Goal:}  In this work, $\mathrm{I}(Z; (\hat{Y},U_X) )$ will serve as our \textit{information-theoretic quantification of the total disparity (exempt and non-exempt)} as we discussed in Definition~\ref{defn:total_disparity} (also recall Lemma~\ref{lem:cci} and Remark~\ref{rem:why_info_theory}). Our \textit{goal} is to appropriately decompose the total disparity $\mathrm{I}(Z; (\hat{Y},U_X) )$ into an exempt component $(M_{E})$ and a non-exempt component $(M_{NE})$, which can and cannot be explained by the critical features $X_c$ (also see Fig.~\ref{fig:total_disparity}). Intuitively, the total disparity captures the idea of a virtual constituent or proxy of $Z$ that has a causal influence on the output $\hat{Y}$. We would like the exempt and non-exempt components of total disparity to be able to capture and mathematically quantify our intuitive notion of what part of the virtual constituent or proxy can and cannot be attributed to the critical features $X_c$ alone. 

Before proceeding further, we also clarify our terminology here. We say that there is \emph{no disparity} when $\mathrm{I}(Z; \hat{Y},U_X )=0$. Alternately, we call \emph{the disparity to be exempt} if only the non-exempt component is $0$, though $\mathrm{I}(Z;\hat{Y},U_X)$ may be zero or non-zero. Table~\ref{table:notations} summarizes all the important notations to help follow the rest of the paper.


\begin{figure*}
\centering
\includegraphics[height=3.5cm]{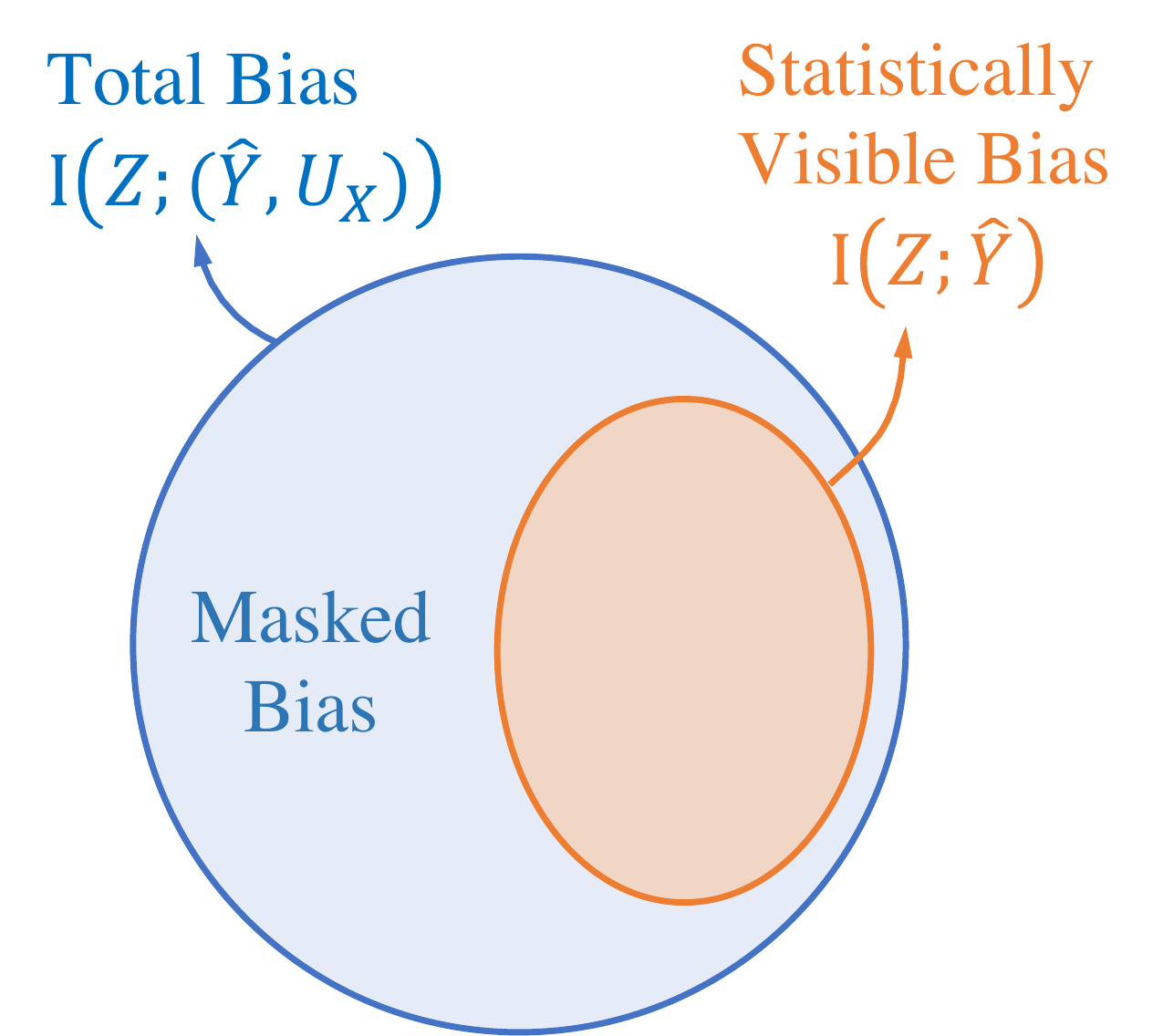}
\hspace{3cm}
\includegraphics[height=3.5cm]{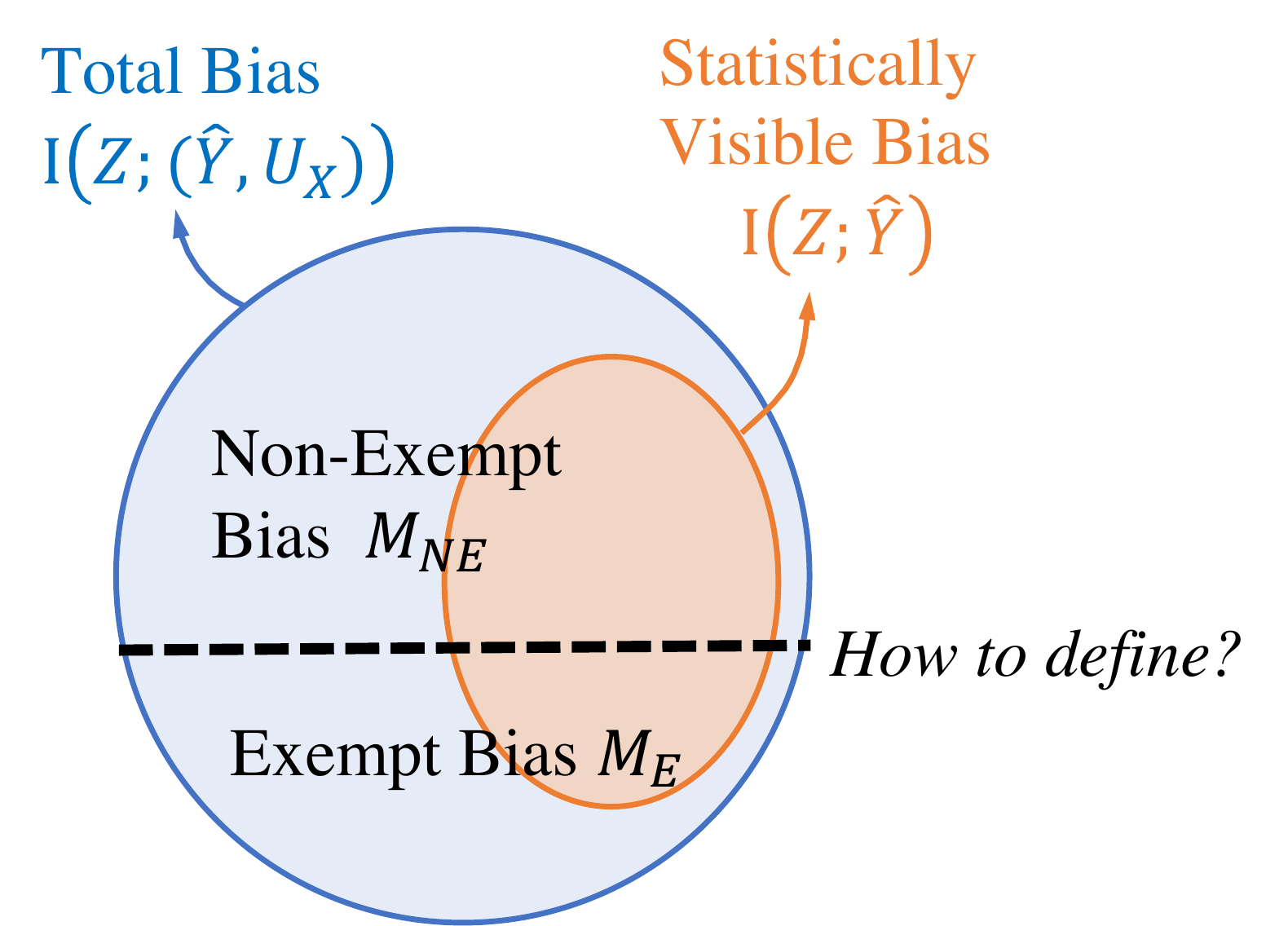}
\caption{Decomposition of Total Disparity: (Left) Total disparity (information-theoretic quantification of counterfactual causal influence) is shown in blue. The statistically visible disparity and masked disparity are two sub-components of the total disparity. (Right) Our goal is to decompose the total disparity into exempt and non-exempt components.\label{fig:total_disparity}}
\end{figure*}




\begin{figure*}
\begin{subfigure}[b]{0.31\linewidth}
\centering
\includegraphics[height=3.8cm]{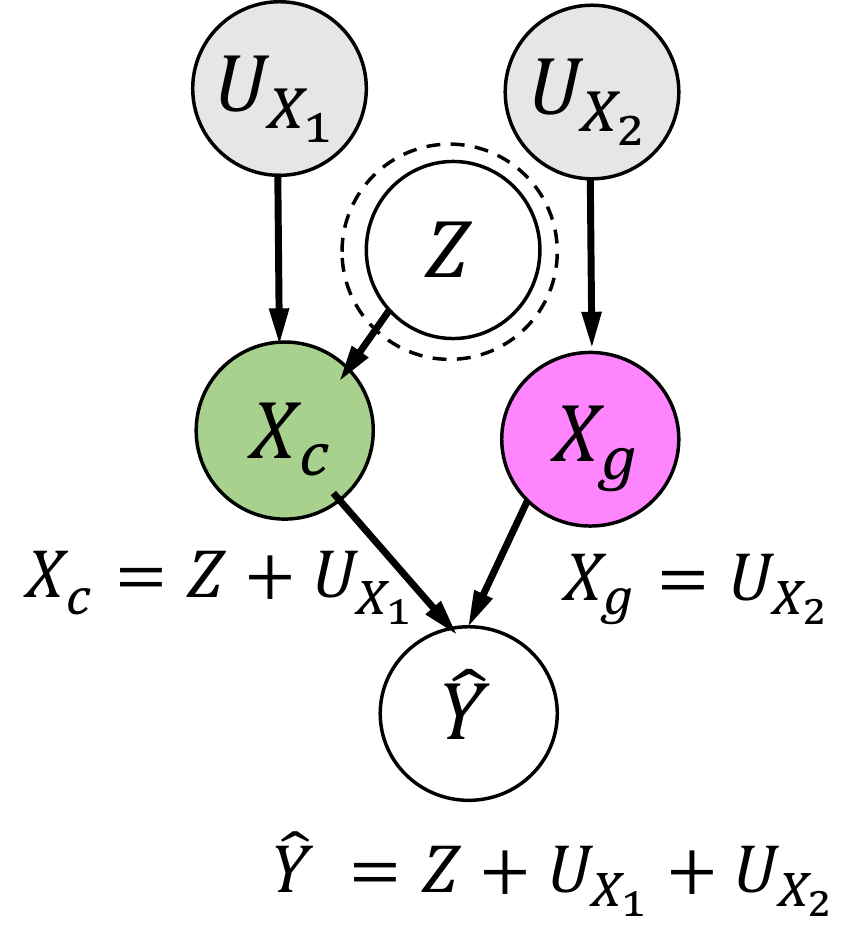}
\caption{Canonical Example~\ref{cexample:exemption}: Hiring with Biased Critical Feature (Desirable: $M_{NE}=0$)}\label{fig:firemen}
\end{subfigure}
\hspace{0.2cm}
\begin{subfigure}[b]{0.31\linewidth}
\centering
\includegraphics[height=3.8cm]{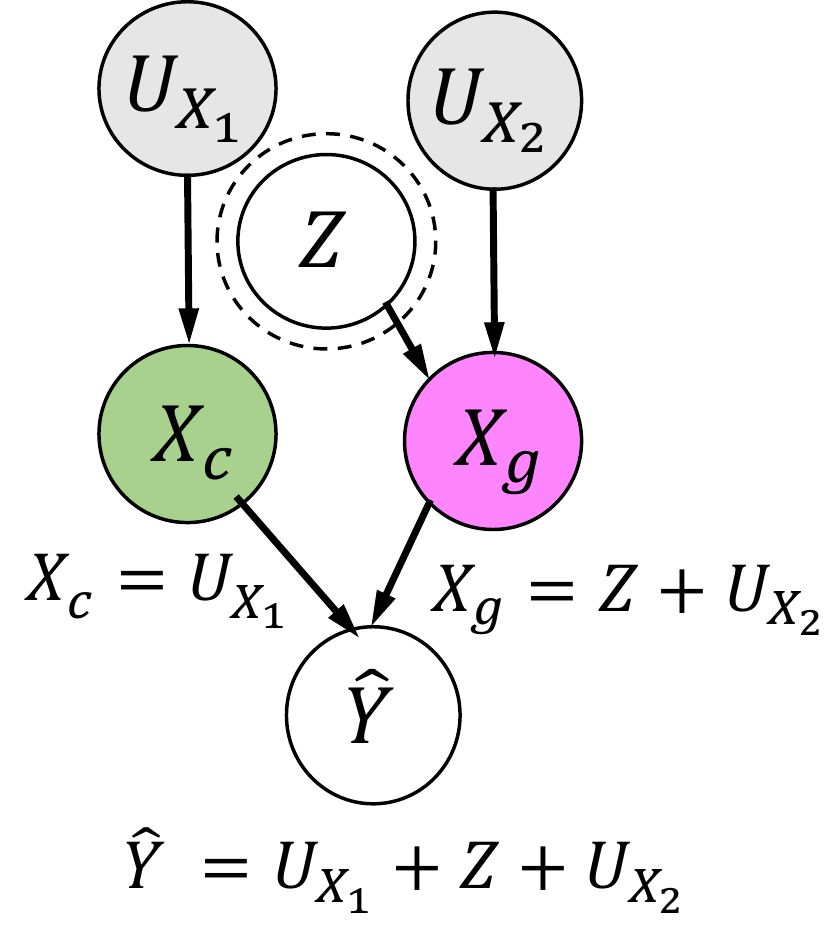}
\caption{Canonical Example~\ref{cexample:equalized_odds}: Hiring with Biased General Feature (Desirable: $M_{NE}>0$)}\label{fig:equalized_odds}
\end{subfigure}
\hspace{0.2cm}
\begin{subfigure}[b]{0.30\linewidth}
\centering
\includegraphics[height=3.8cm]{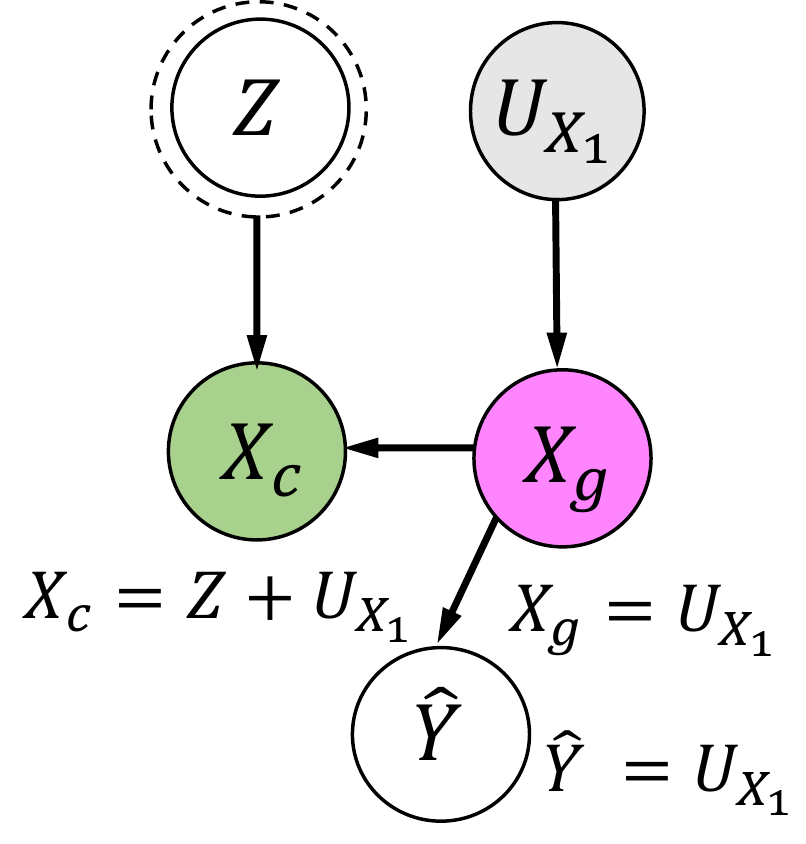}
\caption{Canonical Example~\ref{cexample:college}: Counterfactually Fair Hiring (Desirable: $M_{NE}=0$)}\label{fig:college}
\end{subfigure}

\begin{subfigure}[b]{0.31\linewidth}
\centering
\includegraphics[height=3.6cm]{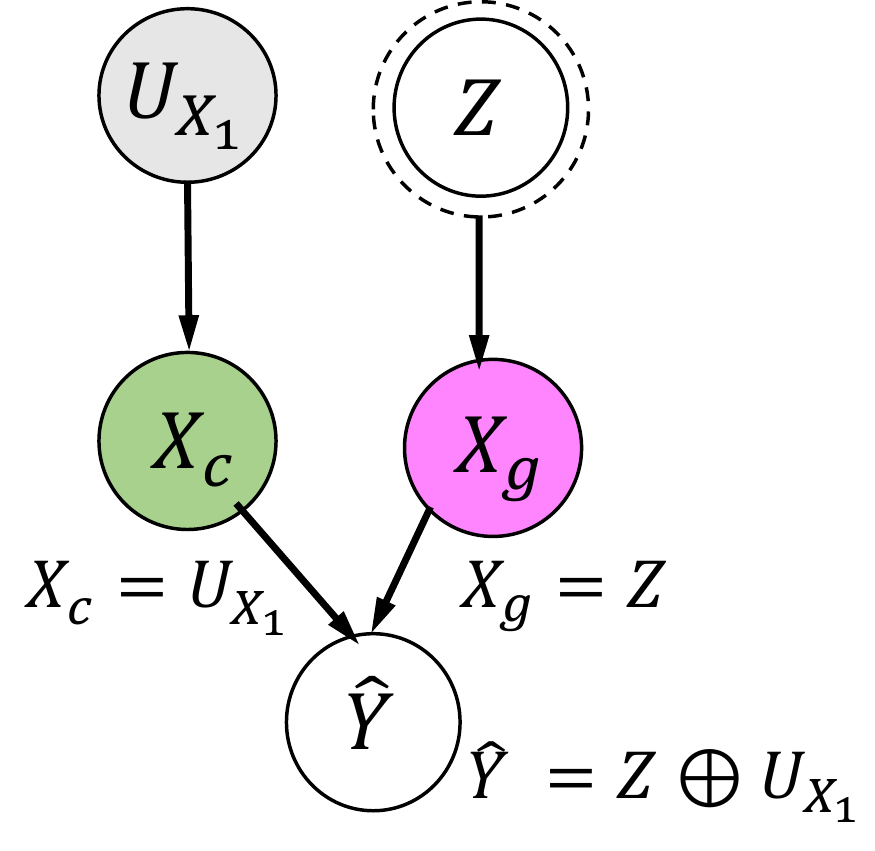}
\caption{Canonical Example~\ref{cexample:masking}: Non-Exempt Masked Disparity in Hiring Ads I (Desirable: $M_{NE}>0$)}\label{fig:masking}
\end{subfigure}
\hspace{0.2cm}
\begin{subfigure}[b]{0.30\linewidth}
\centering
\includegraphics[height=3.6cm]{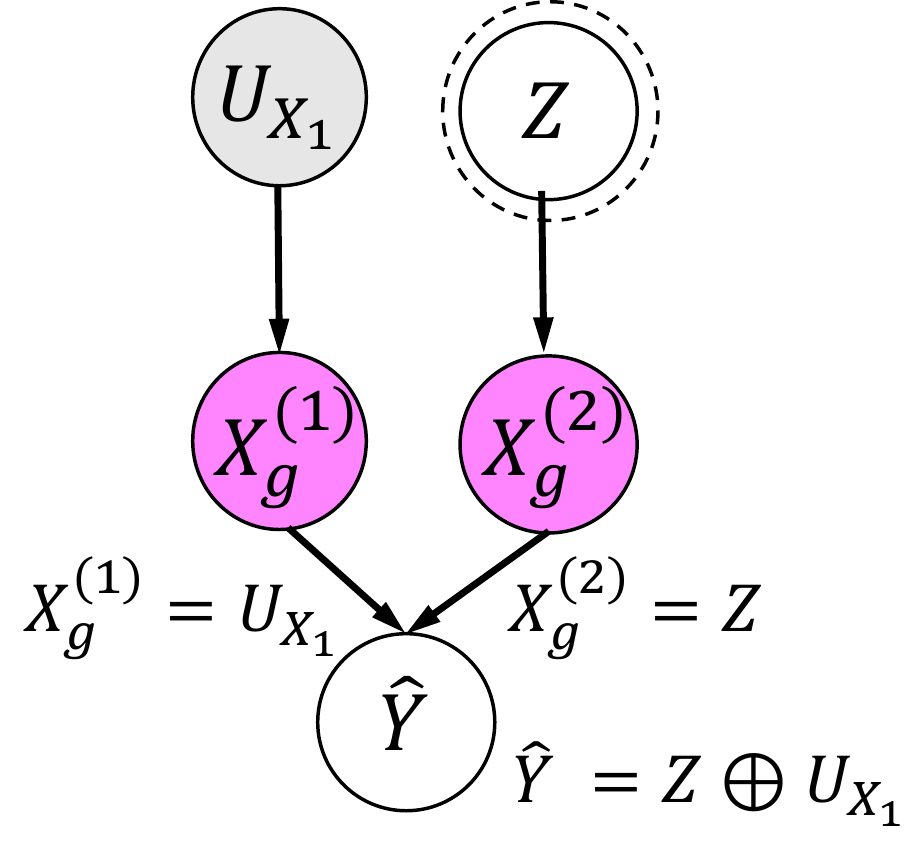}
\caption{Canonical Example~\ref{cexample:masking_general}: Non-Exempt Masked Disparity in Hiring Ads II (Desirable: $M_{NE}>0$)}\label{fig:masking_gen}
\end{subfigure}
\hspace{0.2cm}
\begin{subfigure}[b]{0.30\linewidth}
\centering
\includegraphics[height=3.8cm]{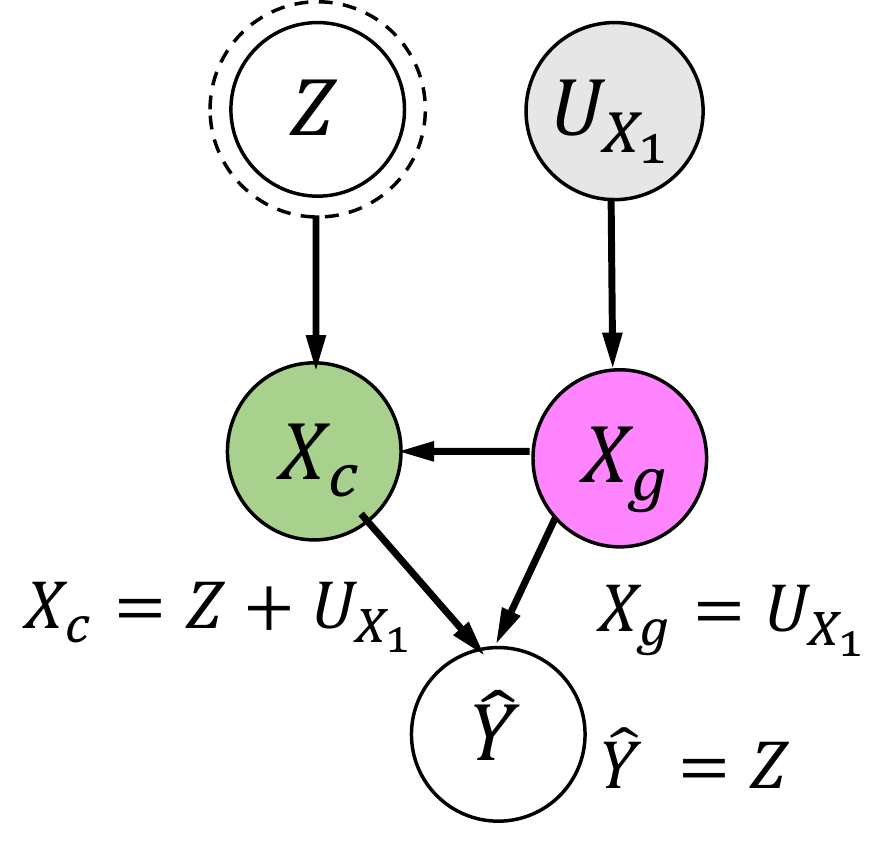}
\caption{Canonical Example~\ref{cexample:synergy}: Disparity Amplification by Unmasking (Desirable: $M_{NE}>0$)}\label{fig:synergy}
\end{subfigure}
\caption{Thought experiments to motivate desirable properties of non-exempt disparity: In all the figures, $Z$ denotes the protected attribute, e.g., gender, race, etc., and $U_{X_1},U_{X_2}$ denotes other latent social factors independent of $Z$. The critical feature is denoted by $X_c$, the non-critical/general feature is $X_g$, and the model output (hiring decision) is $\hat{Y}$. \label{fig:examples1}}
\end{figure*}

\section{Main Results}
\label{sec:properties}

In Section~\ref{subsec:main_results}, we first formally state the desirable properties that a measure of non-exempt disparity $(M_{NE})$ should satisfy. These properties were only intuitively stated in Section~\ref{sec:introduction}. Next, we introduce our proposed measure that satisfies all these  properties (Theorem~\ref{thm:satisfythm} in Section~\ref{subsec:main_results}). In Section~\ref{subsec:rationale}, we discuss in detail on how we arrive at these desirable properties through several canonical examples (summarized in Table~\ref{table:xyz} and Fig.~\ref{fig:examples1}), that helps us quantify our intuitive notion of non-exempt disparity. In Section~\ref{subsec:contrast}, we examine measures in existing literature that have some provision for exemptions, namely, path-specific counterfactual fairness~\cite{chiappa2018path},  conditional statistical parity~\cite{corbett2017}, and justifiable fairness~\cite{interventional_fairness}, and understand their limitations.

\subsection{Desirable Properties Leading to Our Proposed Measure of Non-Exempt Disparity}
\label{subsec:main_results}

It is desirable that our measure of \emph{non-exempt} disparity ($M_{NE}$) is able to capture the intuition of a virtual constituent or proxy of $Z$ being formed inside a given black-box model that a) causally influences the output $\hat{Y}$; and b) cannot be attributed to the critical features $X_c$ alone. To arrive at a set of desirable properties for a measure of \emph{non-exempt disparity} ($M_{NE}$), we examine candidate measures and examine their utility and limitations through canonical examples (see Fig.~\ref{fig:examples1}). While we discuss the rationale for each of these properties in more detail in Section~\ref{subsec:rationale}, here we state the properties and provide a brief intuition for each of them. For simplicity, assume that the protected attribute $Z$ as well as all the other independent latent variables $U_{X_1},U_{X_2},\ldots$ are \iid{} Bern(\nicefrac{1}{2}) in our canonical examples.

Our first candidate measure of non-exempt disparity is based on conditional mutual information, and is: $M_{NE}=\mut{Z}{\hat{Y}\given X_c}$ (Candidate Measure~\ref{candmeas:CMI} in Section~\ref{subsec:rationale}). Inspired from the concept of conditional statistical parity~\cite{corbett2017}, this measure assumes that there is no non-exempt disparity if and only if the hiring decision $\hat{Y}$ and the protected attribute $Z$ (e.g., gender) are independent, conditioned on the critical feature $X_c$ (e.g., coding-test score for a software engineering job). This measure might seem intuitively appealing at first. In Canonical Example~\ref{cexample:exemption} (Fig.~\ref{fig:firemen}), disparity only arises from the critical feature, namely, coding-test score in a software-engineering job, and the general/non-critical feature aptitude-test score contributes to the decision making without introducing disparity. Here, $M_{NE}=\mut{Z}{\hat{Y}\given X_c}=0$ as desired.  In Canonical Example~\ref{cexample:equalized_odds} (Fig.~\ref{fig:equalized_odds}), the disparity only arises from the general/non-critical feature aptitude-test score, which is non-exempt. Here, $M_{NE}=\mut{Z}{\hat{Y}\given X_c}>0$ as desired. 

However, this candidate measure has a limitation: it can sometimes \emph{falsely detect non-exempt disparity when there is none. E.g.}, consider a scenario where the model is counterfactually fair (Canonical Example~\ref{cexample:college} in Section~\ref{subsec:rationale};  Fig.~\ref{fig:college}), and hence there is no disparity (exempt or non-exempt). The critical feature, namely, the coding-test score for a software engineering job is biased, i.e., $X_c=Z+U_{X_1}$ with $U_{X_1}$ being the latent inner ability of a candidate. However, the model is able to distill out the latent inner ability $U_{X_1}$ using all the features and take hiring decisions entirely based on them, i.e., $\hat{Y}=U_{X_1}$. Here, $M_{NE}=\mut{Z}{\hat{Y}\given X_c}>0$ when it is desirable that $M_{NE}$ be $0$. This canonical example motivates the following property:

\begin{restatable}[Zero Influence]{propty}{cancellation}
$M_{NE}$ should be $0$ if $\mathrm{CCI}(Z\rightarrow \hat{Y})=0$ (or equivalently, $\mut{Z}{\hat{Y},U_X}=0$).\label{propty:cancellation}
\end{restatable}

This limitation of $\mut{Z}{\hat{Y}\given X_c}$ leads us to examine PID, decomposing $\mut{Z}{\hat{Y}\given X_c}$ into two components: unique information $\uni{Z}{\hat{Y}\given X_c}$ and synergistic information $\syn{Z}{(\hat{Y}, X_c)}$. The sub-component $\uni{Z}{\hat{Y}\given X_c}$ always satisfies Property~\ref{propty:cancellation} (proof in Lemma~\ref{lem:uni_cci} in Appendix~\ref{app:pid_properties}), even though $\mut{Z}{\hat{Y}\given X_c}$ sometimes may not do so because of the synergistic component (which caused false detection of non-exempt disparity in the previous scenario). This leads us to examine another candidate measure of non-exempt disparity, namely, $M_{NE}=\uni{Z}{\hat{Y}\given X_c}$ (Candidate Measure~\ref{candmeas:uni} in Section~\ref{subsec:rationale}). For example, consider hiring for a software-engineering job using coding-test score (critical feature) and aptitude-test score (non-critical/general feature). It is desirable that $M_{NE}$ be non-zero if $\hat{Y}$ has any unique information about $Z$ that is not present in $X_{c}$ (coding test) because then that information content is also attributed to $X_g$ (also see Section~\ref{subsubsec:uniq} to further motivate this property).

\begin{restatable}[Non-Exempt Statistically Visible Disparity]{propty}{synergy}
$M_{NE}$ should be strictly greater than $0$ if $\hat{Y}$ has any unique information about $Z$ not present in $X_c$. Thus,
$\uni{Z}{\hat{Y}| X_{c}} > 0$ should imply that $M_{NE}>0$.\label{propty:synergy}
\end{restatable}

However, this property alone does not capture all scenarios where $M_{NE}$ is desired to be non-zero. Statistical masking can sometimes prevent the entire non-exempt disparity from exhibiting itself in  $\uni{Z}{\hat{Y}| X_{c}}$ as demonstrated in the following scenario. Suppose an ad for a job is shown selectively to: a) men with high coding-test scores and b) women with low coding-test scores (Canonical Examples~\ref{cexample:masking} and \ref{cexample:masking_general} in Section~\ref{subsec:rationale}; see Fig.~\ref{fig:masking} and \ref{fig:masking_gen}). Such a model might seem ``statistically fair'', i.e., with no statistically visible dependence between $Z$ and $\hat{Y}$ ($\mut{Z}{\hat{Y}}=0$), but is clearly unfair to high-scoring women candidates. Since $\uni{Z}{\hat{Y}| X_{c}}\leq \mut{Z}{\hat{Y}}$ (recall \eqref{eq:pid2} in Section~\ref{subsec:background} and non-negativity of all PID terms), we have $\uni{Z}{\hat{Y}| X_{c}}=0$ for this canonical example, showing that it fails to capture such ``non-exempt masked disparity.'' In essence, $\uni{Z}{\hat{Y}| X_{c}}$ is therefore a lower bound for non-exempt disparity $M_{NE}$, i.e., $\uni{Z}{\hat{Y}| X_{c}}>0 \implies M_{NE}>0$ but not necessarily the other way round (making this candidate measure a ``lower bound'' for $M_{NE}$). The next property attempts to find an upper bound for $M_{NE}$.

Notice that, in the previous Canonical Examples~\ref{cexample:masking} and \ref{cexample:masking_general}, $\hat{Y}$ has a virtual constituent $Z$ influencing it, that is not due to the critical features $X_c$. However, the influence of $Z$ does not exhibit itself in the statistically visible disparity $\mut{Z}{\hat{Y}}$. To resolve this issue, we now consider a non-observational, causal candidate measure inspired from path-specific counterfactual fairness~\cite{chiappa2018path} that specifically examines causal paths from $Z$ to $\hat{Y}$ in the SCM (Candidate Measure~\ref{candmeas:path_specific} in Section~\ref{subsec:rationale}). This measure implies there is no non-exempt disparity if all paths from $Z$ to $\hat{Y}$ in the SCM pass through $X_c$. However, we identify scenarios where this approach can also fail to quantify non-exempt disparity, e.g., in Canonical Example~\ref{cexample:synergy} in Section~\ref{subsec:rationale} (Fig.~\ref{fig:synergy}). Here the critical feature, coding-test score is $X_c=Z+U_{X_1}$, and the non-critical feature, aptitude-test score is $X_g=U_{X_1}.$ The model amplifies the disparity in the hiring decision by cancelling $U_{X_1}$, i.e., $\hat{Y}=Z$. For this example, even though we have the causal path from $Z$ to $\hat{Y}$ passing through $X_c$, we contend that here both $X_c$ and $X_g$ jointly have information about $Z$ that cannot be attributed to $X_c$ alone. Therefore, it is desirable that we have a measure of non-exempt disparity $M_{NE}$ which is non-zero for this example ($\uni{Z}{\hat{Y}| X_{c}}$ and $\mut{Z}{\hat{Y}| X_{c}}$ are also non-zero for this example).

From a causal point of view, here $U_{X_1}$ is a ``confounder'' for both $X_c$ and $\hat{Y}$ (separately influences both $X_c$ and $\hat{Y}$ along different paths). Intuitively, a scenario when there is no non-exempt disparity would be: (i) All causal paths from $Z$ to $\hat{Y}$ in the SCM pass through $X_c$; and also (ii) No $U_{X_i}$ acts as a confounder for both $X_c$ and $\hat{Y}$ (also refer to Canonical Example~\ref{cexample:exemption} in Fig.~\ref{fig:firemen}). This leads to the intuition that to be able to say there is no non-exempt disparity, one might be able to split $U_X$ into two subsets $U_a$ and $U_b$ (further functional generalizations discussed in Section~\ref{sec:conclusion}), such that: (i) $U_a$ consists of the latent factors that do not influence $\hat{Y}$ at all, or influence it only through $X_c$ without acting as confounder; (ii) $U_b$ consists of the remaining latent factors, that only influence $\hat{Y}$ and not $X_c$; and (iii) The Markov chain $(Z,U_a)-X_c-(\hat{Y},U_b)$ holds\footnote{Notice that, this condition implies $Z-X_c-\hat{Y}$ but not the other way round.}. This leads to the following property (see Section~\ref{subsubsec:causality} to further motivate this property).

\begin{restatable}[Non-Exempt Masked Disparity]{propty}{nonexemptmasking} $M_{NE}$ should be non-zero in the canonical example of non-exempt masked disparity: $X_1=Z$, $X_2=U_X$, and $\hat{Y}=Z \oplus U_X$ with $Z,U_X\sim$ \iid{} Bern(\nicefrac{1}{2}) and $X_1 \in X_g$. However, $M_{NE}$ should be $0$ if $(Z,U_a)-X_c-(\hat{Y},U_b)$ form a Markov chain for some subsets $U_a,U_b \subseteq U_X$ such that $U_a=U_X\backslash U_b$.
\label{propty:masking}
\end{restatable}

 Properties~\ref{propty:synergy} and \ref{propty:masking} provide lower and upper bounds on our measure of non-exempt disparity, i.e., it is desirable that: 
\begin{equation}
\uni{Z}{\hat{Y}| X_{c}} \leq M_{NE}  \leq \min_{U_a,U_b\text{ s.t. } U_a=U_X\backslash U_b}\mut{(Z,U_a)}{(\hat{Y},U_b)\given X_c}.
\end{equation}
This observation is important in itself: the unique information measure, being a lower bound, never falsely detects non-exempt disparity when there is none, and thus can serve as a conservative estimate of non-exempt disparity. 

The next three properties are more intuitive. Consider the scenario where no feature is deemed critical (i.e., $X_c=\phi$) and all features are non-critical, e.g., hiring for a manager's role using aptitude-test and coding-test scores. Here, one would like $M_{NE}$ to be equal to the total disparity $\mut{Z}{(\hat{Y},U_X)}$, i.e., no disparity is exempt because no feature is deemed critical.

\begin{restatable}[Absence of Exemptions]{propty}{absence}  
If no feature is deemed critical ($X_c=\phi$), then a measure $M_{NE}$ should be equal to the total disparity, i.e., $\mut{Z}{(\hat{Y},U_X)}$.\label{propty:absence}  
\end{restatable}

Next, suppose that the same model is being used for a software-engineering role where coding-test score is deemed as a critical feature but aptitude-test score is not. For a fixed set of features and a fixed model $\hat{Y}=h(Z,U_X)$, it is desirable that $M_{NE}$ either decreases or stays the same as more features are removed from the set $X_g$ and added to $X_c$.

\begin{restatable}[Non-Increasing with More Exemptions]{propty}{nonincreasing}
For a fixed set of features $X$ and a fixed model $\hat{Y}=h(Z,U_X)$, a measure $M_{NE}$ should be non-increasing if a feature is  removed from $X_g$ and added to $X_{c}$. \label{propty:nonincreasing}
\end{restatable}

Lastly, suppose that the model is used for an even more specific role where both coding test and aptitude test are deemed as critical features. If all the features are in the exempt set $X_{c}$, we require the measure $M_{NE}$ to be $0$. 

\begin{restatable}[Complete Exemption]{propty}{completeexemption}
$M_{NE}$ should be $0$ if all features are exempt, i.e., $X_c=X$ and $X_g=\phi$.\label{propty:complete_exemption}
\end{restatable}

These six properties lead to a novel measure of non-exempt disparity that satisfies all of them (proved in Theorem~\ref{thm:satisfythm}).

\begin{defn}[Non-Exempt Disparity] Our proposed measure of non-exempt disparity is given by:
\begin{equation} M^*_{NE}= \min_{U_a,U_b} \uni{(Z,U_a)}{(\hat{Y},U_b)| X_c}    \text{such that } U_a=U_X\backslash U_b.
\end{equation}
\label{defn:non_exempt_disparity}
\end{defn}

Note that, for the rest of the paper, we use the notation $M_{NE}$ to denote any candidate measure of non-exempt disparity, and $M^*_{NE}$ to specifically denote our proposed measure in Definition~\ref{defn:non_exempt_disparity}.

\begin{restatable}[Properties]{thm}{satisfythm} Properties 1-6 are satisfied by $ M^*_{NE}=\min_{U_a,U_b} \uni{(Z,U_a)}{(\hat{Y},U_b)| X_c}   \ \text{such that } U_a=U_X\backslash U_b.$
 \label{thm:satisfythm}
\end{restatable}

\noindent \textit{Proof Sketch:} A detailed proof is provided in Appendix~\ref{app:properties}. Here, we provide a brief proof sketch. For Property~\ref{propty:cancellation},
\begin{align}
M^*_{NE} \leq \uni{Z}{\hat{Y},U_X| X_c} \leq \mut{Z}{(\hat{Y},U_X)},
\end{align}
where the last step holds as unique information is also a component of mutual information (see \eqref{eq:pid2} in Section~\ref{subsec:background}). For Property~\ref{propty:synergy}, we show that $M^*_{NE}{\geq} \uni{Z}{\hat{Y}| X_c}$ using a monotonicity property of unique information~\cite[Lemma 31]{banerjee2018unique}. Lastly, for Property~\ref{propty:masking}, we have $\mut{Z,U_a}{\hat{Y},U_b|X_c}=0$ for some $U_a,U_b$, implying that  $\uni{Z,U_a}{\hat{Y},U_b| X_c}$ is also $0$ for those $U_a,U_b$ because unique information is a component of conditional mutual information (see \eqref{eq:pid3} in Section~\ref{subsec:background}).  For Property~\ref{propty:absence}, we show that when $X_c=\phi$, we have
$M^*_{NE}=\min_{U_a,U_b\text{ s.t. } U_a=U_X\backslash U_b}\mut{Z,U_a}{\hat{Y},U_b} = \mut{Z}{(\hat{Y},U_X)}.$ Property~\ref{propty:nonincreasing} is derived using another monotonicity property of unique information~\cite[Lemma 32]{banerjee2018unique}. For Property~\ref{propty:complete_exemption}, 
\begin{align}  
M^*_{NE} \leq \uni{Z,U_X}{\hat{Y}| X}  \overset{(a)}{\leq} \mut{Z,U_X}{\hat{Y}|X} \overset{(b)}{=}0,
\end{align}
where (a) holds because unique information is a component of conditional mutual information (see \eqref{eq:pid3} in Section~\ref{subsec:background}) and (b) holds as $\hat{Y}$ is a deterministic function of $X$.

\begin{rem}[On Exhaustive Set of Properties leading to a Unique Measure]
\label{rem:uniqueness} We note that our properties do not quantify how exactly the non-exempt disparity should ``scale'' when the measure is nonzero since they are only conditions on when this disparity is nonzero, or on the monotonicity of this disparity. Hence, these properties do not lead to a unique measure. Also, note that this is an issue with all measures of fairness in that they go to zero based on an intuitive notion of fairness but their exact scaling when they are non-zero is not unique. Neither do we claim that the proposed list of desirable properties (axioms) are exhaustive. In general, it is difficult to prove that a proposed set of properties (or, axioms) is exhaustive for a problem. E.g., Shannon established uniqueness of entropy with respect to \textbf{some} properties in \cite{shannon1948mathematical} but the needs of the application can still drive the use of alternate measures. E.g. Renyi measures~\cite{renyi1961measures,liao2019learning,issa2019operational,liao2019robustness,mironov2017renyi} have been found to be useful in security and privacy applications because they weigh outliers differently. Therefore, we believe, that there may be value in the measure not being unique so that it can be tuned to the needs of the application, as well as, motivate future work in this direction. Nonetheless, our properties do capture important aspects of the problem, e.g., non-exempt masked and non-exempt statistically visible disparities, as discussed in Section~\ref{sec:measure_decomposition} and also in Remark~\ref{rem:simplicity}.
\end{rem}

\begin{rem} We note that the proposed measure is counterfactual (non-observational/causal) in nature, i.e., it requires knowledge of the true SCM. While we are able to compute the measure in our case study on artificial datasets (known SCM) in Section~\ref{sec:case_study}, we acknowledge that even after knowledge of the true SCM, there may be computational challenges if the number of latent variables is large. However, one must note that it is important to arrive at measures that satisfy all desirable properties, however hard they might be to compute: (i) It makes the shortcomings of other measures more explicit, informing \emph{which} computable/estimable definition to choose in a given situation; (ii) It opens the avenue of obtaining relaxations that may be easier to estimate; (iii) One can begin exploring research directions to reduce the difficulty/complexity (statistical and/or computational) of estimating these measures.
\end{rem}

\begin{rem}[On Simplicity of Examples] We note that, at a first glance, our examples might seem simple, and real world models will only be more complex due to a mix of causal and statistical relationships. These simple examples help us isolate many of these individual causal and statistical relationships, and examine them carefully. E.g., scenarios where only one of non-exempt masked, non-exempt visible, exempt masked or exempt visible disparity is present or none of them is present (see Fig.~\ref{fig:exhaustive}). When both non-exempt masked and non-exempt statistically visible disparities are present together, we are able to quantify both of them appropriately (discussed further in Section~\ref{sec:measure_decomposition}). Thus, developing an axiomatic understanding of such simple examples is an essential first step in understanding the complex interplay of various relationships in a real dataset. Indeed, examining toy examples (thought experiments) is a common practice in several works in existing fairness literature~\cite{yeom2018discriminative,kusner2017counterfactual,kilbertus2017avoiding,kearns2018preventing,interventional_fairness}, some of which have also inspired our examples in this work. Furthermore, our quantification of non-exempt disparity is not limited to black-box models alone, but also applies to ``white-box'' models~\cite{datta2017use}, e.g., decision trees, linear classifiers, etc., and also to non-AI-based decisions as long as the decision is as a deterministic function of the input features, i.e., $\hat{Y}=h(X)$. \label{rem:simplicity}
\end{rem}

\begin{figure*}
\centering
\includegraphics[height=4.5cm]{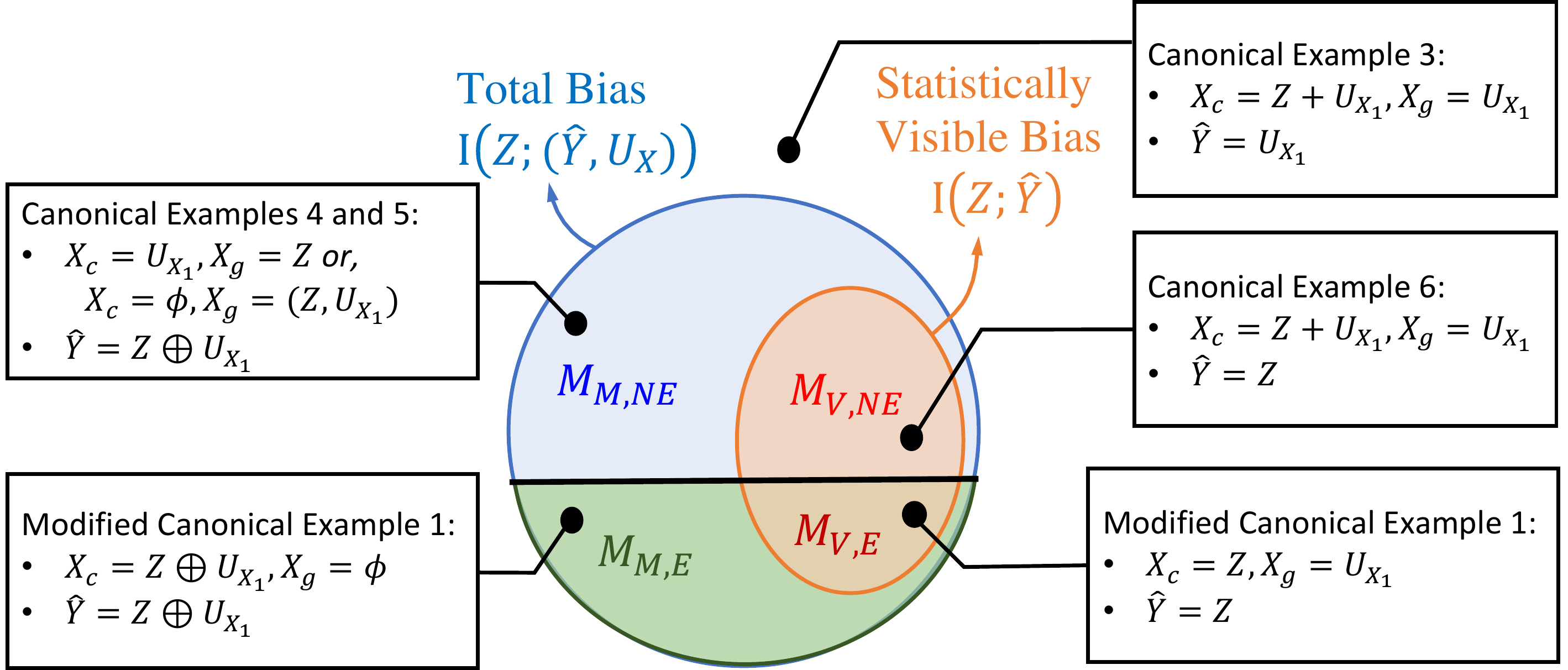}
\caption{Our examples isolate different kinds of scenarios, namely, masked non-exempt ($M_{M,NE}$),  masked exempt ($M_{M,E}$),  visible non-exempt ($M_{V,NE}$), and visible exempt ($M_{V,E}$), as well as scenarios where there is no total disparity(more in Section~\ref{sec:measure_decomposition}). \label{fig:exhaustive}}
\end{figure*}

\subsection{Detailed Rationale Behind the Desirable Properties Leading to A Measure of Non-Exempt Disparity}
\label{subsec:rationale}

Here we provide detailed rationale\footnote{Some of the arguments in this subsection have already been introduced briefly in Section~\ref{subsec:main_results}, and are being elaborated here.} behind all our desirable properties using canonical examples (summarized in Table~\ref{table:xyz}). We start by examining two canonical examples that help us motivate the basic intuition behind \emph{non-exempt disparity}. These examples also help us understand the limitations of \emph{statistical parity}~\cite{agarwal2018reductions,dwork2012fairness} and \emph{equalized odds}~\cite{hardt2016equality} which are two popular measures of fairness that do not have provision for critical feature exemptions.

\begin{table*}
	\centering
	\caption{Summary of Canonical Examples and Candidate Measures of Non-Exempt Disparity} \smallskip
	\label{table:xyz}
	\begin{tabular}{p{4.7cm}|R{1.4cm}|R{1.9cm}|R{2.2cm}|R{5cm}}
		\toprule
	Canonical Examples & Candidate Measure~\ref{candmeas:CMI}: $\mut{Z}{\hat{Y}|X_c}$ & Candidate Measure~\ref{candmeas:uni}: $\uni{Z}{\hat{Y}|X_c}$ 
	& Candidate Measure~\ref{candmeas:path_specific}: Path-Specific Causality & Proposed Measure: $\min_{U_a,U_b} \uni{(Z,U_a)}{(\hat{Y},U_b)| X_c}$ such that $U_a=U_X\backslash U_b.$ \\
		\midrule
	 \ref{cexample:exemption}. Hiring with Biased Critical Feature
	 \begin{itemize}
	 \item $X_c=Z+U_{X_1}$ and $X_g=U_{X_2}$.
	 \item $\hat{Y}=Z+U_{X_1}+U_{X_2}$.  
	 \end{itemize}
	 \textbf{Desirable: $M_{NE}=0$}
	 & \checkmark & \checkmark 
	 & \checkmark & \checkmark  \\
	 		\midrule 
	 \ref{cexample:equalized_odds}. Hiring with Biased General Feature
	 \begin{itemize}
	 \item $X_c=U_{X_1}$ and $X_g=Z+U_{X_2}$. 
	 \item $\hat{Y}=Z+U_{X_1}+U_{X_2}$. 
	 \end{itemize} 
	 \textbf{Desirable: $M_{NE}>0$}  & \checkmark 
	 & \checkmark & \checkmark & \checkmark  	\\
	 	\midrule 
	 		 \ref{cexample:college}. Counterfactually Fair Hiring 
	 		 \begin{itemize} \item $X_c=Z+U_{X_1}$ and $X_g=U_{X_1}$.
	 		 \item $\hat{Y}=U_{X_1}$. 
	 		 \end{itemize}
	 		 \textbf{Desirable: $M_{NE}=0$}   & $\times$ & \checkmark 
	 		 & \checkmark & \checkmark 	\\ 
	 		 \midrule 
		 \ref{cexample:masking}. Non-Exempt Masked Disparity in Hiring Ads I
		 \begin{itemize}
		 \item $X_c=U_{X_1}$ and $X_g=Z$.
		 \item $\hat{Y}=Z \oplus U_{X_1}.$ 
		 \end{itemize}
		 \textbf{Desirable: $M_{NE}>0$}  & \checkmark & $\times$ 
		 & \checkmark & \checkmark    \\
		 \midrule
		  \ref{cexample:masking_general}. Non-Exempt Masked Disparity in Hiring Ads II
		  \begin{itemize}
		  \item $X_c=\phi$ and $X_g=(Z,U_{X_1})$. 
		  \item $\hat{Y}=Z \oplus U_{X_1}$. 
		  \end{itemize}
		  \textbf{Desirable: $M_{NE}>0$} & $\times$ & $\times$ 
		  & \checkmark & \checkmark    \\
		 \midrule
	 \ref{cexample:synergy}. Disparity Amplification by Unmasking
	 \begin{itemize}
   \item	 $X_c=Z + U_{X_1}$ and $X_g=U_{X_1}$.
   \item  $\hat{Y}= Z$.
   \end{itemize}
   \textbf{Desirable: $M_{NE}>0$}  & \checkmark & \checkmark  
   & $\times$ & \checkmark  \\
		\bottomrule
	\end{tabular} 
\end{table*}

\subsubsection{Limitations of Statistical Parity}

As discussed in Section~\ref{sec:sys_model}, a model is deemed fair by statistical parity if $Z\independent \hat{Y}$, i.e., $\mut{Z}{\hat{Y}}=0$. However, the following example exposes some of its limitations.

\begin{cexample}[Hiring with Biased Critical Feature]
Let $X_c=Z+U_{X_1}$ be a coding-test score\footnote{The influence of $Z$ on score in the SCM can arise due to various factors, e.g., historical lack of opportunities or sampling bias due to candidates of one protected group not applying enough etc. For instance, there may be a hidden node representing opportunity such that $Z$ influences the score only though that hidden node, and the score becomes independent of $Z$ given opportunity. We adopt a simplistic representation here for ease of understanding (also see \cite{kilbertus2019sensitivity}).} and $X_g=U_{X_2}$ be an aptitude-test score. Here the protected attribute $Z\sim$ Bern(\nicefrac{1}{2}) denotes gender, $U_{X_1}\sim$ Bern(\nicefrac{1}{2}) denotes inner ability to code and $U_{X_2}\sim$ Bern(\nicefrac{1}{2}) denotes knowledge. An algorithm is deciding whether to hire software engineers based on a score $\hat{Y}=Z + U_{X_1}+ U_{X_2}$. This is shown in Fig.~\ref{fig:firemen}. Here $+$ denotes addition (not to be confused with the binary OR).
\label{cexample:exemption} 
\end{cexample}

First notice that this model will be deemed \emph{unfair} by both statistical parity and counterfactual fairness. Statistical parity is violated because $Z$ and $\hat{Y}$ are not independent, i.e., the statistically visible disparity 
$\mut{Z}{\hat{Y}}>0.$ Consequently, the total disparity $\mut{Z}{(\hat{Y},U_X)}$ is also non-zero since $\mut{Z}{(\hat{Y},U_X)}\geq \mut{Z}{\hat{Y}}>0$, violating counterfactual fairness. However, for this example, the coding-test score is a critical feature (bonafide requirement) for the job. Therefore, one may feel that any disparity in $\hat{Y}$ that is explainable by the coding-test score may be exempted. An attempt to ensure statistical parity for such an example, e.g., by reducing the importance (weight) of the critical feature in the decision making, violates the bonafide requirement of the job. Intuitively, even though the virtual constituent or proxy of $Z$, namely, $Z+U_{X_1}$, influences the output $\hat{Y}$, it is entirely explainable by $X_c$. Thus, for such an example, it is desirable that a measure of discrimination (non-exempt disparity $M_{NE}$) be $0$. 

\subsubsection{Limitations of Equalized Odds}

Equalized odds~\cite{hardt2016equality,ghassami2018fairness} is another popular measure of fairness that attempts to address this limitation of statistical parity by using the true labels (or true final-decision scores) to represent the job requirements. Equalized odds states that a model is fair if
\begin{equation} \Pr(\hat{Y}=y|Z=z,Y=\tilde{y})= \Pr(\hat{Y}=y|Z=z',Y=\tilde{y})  \forall z,z',y,\tilde{y}.
\end{equation}
This criterion is also equivalent to $\hat{Y}\independent Z|Y$, or, $\mut{Z}{\hat{Y}\given Y}=0$. Indeed, in the previous example (Canonical Example~\ref{cexample:exemption}), if the true final-decision scores already incorporate this critical requirement in them, e.g., $Y=Z+U_{X_1}+U_{X_2}$, then $\mut{Z}{\hat{Y}\given Y}=0$, and the model is deemed \emph{fair} by equalized odds. While equalized odds is a reasonable quantification in scenarios where the true label (or true final-decision score) is indeed a justified representation of the job requirements, the measure $\mut{Z}{\hat{Y}\given Y}$ has often been criticized to be affected by label bias, as we demonstrate through this example.

\begin{cexample}[Hiring with Biased General Feature]
Let $X_c=U_{X_1}$ denote the coding-test score and $X_g = \begin{cases} U_{X_2}+1, & Z=0 \\
U_{X_2}, & Z=1
\end{cases}$ denote the aptitude-test score (biased). This can be rewritten as $X_g= Z(U_{X_2}+1) + (1-Z)U_{X_2}= Z+U_{X_2}$, where $Z\sim$ Bern(\nicefrac{1}{2}) denotes gender, $U_{X_1}\sim$ Bern(\nicefrac{1}{2}) denotes the inner ability to code and $U_{X_2}\sim$ Bern(\nicefrac{1}{2}) denotes knowledge. Now suppose, the historic dataset has true decision scores given by $Y=U_{X_1}+Z+U_{X_2}$. This is shown in Fig.~\ref{fig:equalized_odds}. 
\label{cexample:equalized_odds}
\end{cexample}

In this scenario, suppose we choose a perfect predictor, i.e., $\hat{Y}=Y=U_{X_1}+Z+U_{X_2}$. The perfect predictor always satisfies equalized odds because $\mut{Z}{\hat{Y}\given Y}=0$ if $\hat{Y}=Y$. However, if examined deeply, this model is propagating disparity from aptitude-test score, a non-critical/general feature, which is discriminatory and non-exempt. Intuitively, a virtual constituent or proxy of $Z$, i.e., $Z+U_{X_2}$, is being formed from $X_g$ that is influencing the output $\hat{Y}$. For such an example\footnote{The example can be made more realistic if $U_{X_1},U_{X_2}$ are \iid{} $\mathcal{N}(0,1)$. Now suppose, the historic dataset has true labels given by $Y=\sgn{Z+U_{X_1}+U_{X_2}-0.5}$ which is binary. A perfect classifier $\hat{Y}=Y$, that satisfies equalized odds, is still discriminatory because it is influenced by $Z$ in its decision, that is arising from a non-critical feature.}, it is desirable that a measure of discrimination (non-exempt disparity $M_{NE}$) is not zero.

\subsubsection{Motivation for Conditional Mutual Information and its Limitations}

Next, we start out with the aim of finding a suitable measure of non-exempt disparity ($M_{NE}$) that resolves both these canonical examples. Notice that, both these examples can be resolved by a notion of \emph{conditional statistical parity}~\cite{corbett2017}, which deems a model as fair if and only if $Z\independent \hat{Y}|X_c$, i.e., \begin{equation}
\Pr(\hat{Y}=y|X_c=x_c,Z=z) \\ =\Pr(\hat{Y}=y|X_c=x_c,Z=z')  \ \forall y,x_c,z,z'.
\end{equation}
This idea also connects with Simpson's paradox~\cite{peters2017elements} which refers to a statistical trend that appears in several different groups of data but disappears or reverses when these groups are combined. In Canonical Example~\ref{cexample:exemption}, $Z$ and $\hat{Y}$ are not independent but they become so when conditioned on $X_c$, i.e., $\mut{Z}{\hat{Y}}> \mut{Z}{\hat{Y}\given X_c}$. In Canonical Example~\ref{cexample:equalized_odds}, $\mut{Z}{\hat{Y}}< \mut{Z}{\hat{Y}\given X_c}$.
This notion of \emph{conditional statistical parity} leads us to propose the following quantification of non-exempt disparity ($M_{NE}$). 

\begin{candmeas} $M_{NE}=\mut{Z}{\hat{Y} \given X_{c}}$.
\label{candmeas:CMI}
\end{candmeas}
This measure resolves both Canonical Examples~\ref{cexample:exemption} and \ref{cexample:equalized_odds}. However, the following example exposes some of its limitations.

\begin{cexample}[Counterfactually Fair Hiring] Let $Z\sim$ Bern(\nicefrac{1}{2}) be gender, $U_{X_1}\sim$ Bern(\nicefrac{1}{2}) be the inner ability of a candidate, and  $X_{c}=\begin{cases} U_{X_1}, & Z=0 \\
U_{X_1}+1, & Z=1
\end{cases}$ be the coding-test score (critical feature). This can be rewritten as $X_c= Z(U_{X_1}+1) + (1-Z)U_{X_1}= Z+U_{X_1}.$ However, instead of only using the biased test score, suppose the company chooses to conduct thorough evaluation of their online code samples, leading to another score that distills out their inner ability, i.e., $X_{g}=U_{X_1}$. Suppose the model for hiring that maximizes accuracy turns out to be $\hat{Y}=X_{g}=U_{X_1}$. This is shown in Fig.~\ref{fig:college}.
\label{cexample:college}
\end{cexample}

Notice that, this model is deemed \emph{fair} by counterfactual fairness because the total disparity $\mut{Z}{(\hat{Y},U_X)}=0.$ This means that the output $\hat{Y}$ has no counterfactual causal influence of $Z$. Even though the disparity from $X_c$ is legally exempt, the trained black-box model happens to base its decisions on another available non-critical/general feature that has no counterfactual causal influence of $Z$. Thus, there is no disparity in the outcome $\hat{Y}$ (this is true even if the features in $X_c$ were not exempt). Therefore, it is desirable that the non-exempt disparity $M_{NE}$ is also $0$. This is also consistent with the intuition that here no virtual constituent or proxy of $Z$ influences the output. However, the candidate measure $\mut{Z}{\hat{Y}\given X_c}=\mut{Z}{U_{X_1}\given Z+U_{X_1}}$ is non-zero here, leading to a false positive conclusion in detecting non-exempt disparity.

\begin{rem}[Cancellation of Paths]
A similar situation arises if $X_{c}= Z+U_{X_1}$, $X_g=Z$ and $\hat{Y}=X_c-X_g=U_{X_1}$. Even though the disparity from $X_c$ may be exempt, the trained model ends up removing the counterfactual causal influence of $Z$ from the decisions to make them counterfactually fair in a manner similar to the example of interviews (recall Scenario~\ref{scenario:cf} in Section~\ref{sec:sys_model}; also shown in Fig.~\ref{fig:cf2}). The influences of $Z$ along two different causal paths cancel each other in the final output, so that $\mathrm{CCI}(Z\rightarrow \hat{Y})=0$ (and, $\mut{Z}{(\hat{Y},U_X)}=0$). Since the total disparity $\mut{Z}{(\hat{Y},U_X)}=0$, the question of non-exempt or exempt disparity does not arise. However, the candidate measure $\mut{Z}{\hat{Y}\given X_c}$ is non-zero here. \label{rem:cancellation}
\end{rem}

This example also serves as a rationale for the property of zero influence, i.e., Property~\ref{propty:cancellation} which states that $M_{NE}$ should be $0$ if the total disparity is $0$. We aim to find a measure that resolves all of these examples (summarized in Fig.~\ref{fig:examples1}).

\subsubsection{Motivation for Unique Information and its Limitations}
\label{subsubsec:uniq}
We notice that conditioning on the critical feature $X_c$ can increase or decrease mutual information. For instance, in Canonical Example~\ref{cexample:exemption}, we have $\mut{Z}{\hat{Y}}>0$ but  $\mut{Z}{\hat{Y}\given X_c}=0$. In Canonical Example~\ref{cexample:college}, $\mut{Z}{\hat{Y}\given X_c}>0$ but $\mut{Z}{\hat{Y}}=0$. For both these examples, it is desirable that $M_{NE}=0.$ This motivates us to consider another candidate measure of non-exempt disparity that is equal to the information-theoretic sub-volume of intersection between $\mut{Z}{\hat{Y}}$ and $\mut{Z}{\hat{Y}\given X_c}$ (recall Fig.~\ref{fig:pid2}), that goes to $0$ when any one of them is $0$. This is a quantity that is derived from the PID literature, and is called the \emph{unique information} of $Z$ in $\hat{Y}$ that is not present in $X_c$.

\begin{candmeas} $M_{NE}=\uni{Z}{\hat{Y} | X_{c}}$. 
\label{candmeas:uni}
\end{candmeas}

This measure resolves the examples discussed so far, namely, Canonical Example~\ref{cexample:exemption} (Fig.~\ref{fig:firemen}), Canonical Example~\ref{cexample:equalized_odds} (Fig.~\ref{fig:equalized_odds}), Canonical Example~\ref{cexample:college} (Fig.~\ref{fig:college}) and a (similar) example in Remark~\ref{rem:cancellation}. We start with Canonical Example~\ref{cexample:exemption} (hiring with biased critical feature), where $\hat{Y}=Z+U_{X_1}+U_{X_2}$ and $X_c=Z+U_{X_1}$. Recall that the mutual information can be decomposed as follows: $\mut{Z}{\hat{Y}}= \uni{Z}{\hat{Y} | X_{c}}+ \rd{Z}{(\hat{Y}, X_c)} \text{ (from \eqref{eq:pid2} in Section~\ref{subsec:background})}.$
For this example, we notice that even though $\mut{Z}{\hat{Y}}>0$, we have $\uni{Z}{\hat{Y} | X_{c}}=0$. This is because, $\mut{Z}{\hat{Y}\given X_c}=\uni{Z}{\hat{Y} | X_{c}} + \syn{Z}{(\hat{Y},X_c)}\ \text{(from \eqref{eq:pid3} in Section~\ref{subsec:background})},$  and $\mut{Z}{\hat{Y}\given X_c}=0$ for Canonical Example~\ref{cexample:exemption}. In Canonical Example~\ref{cexample:exemption}, the entire statistically visible disparity $\mut{Z}{\hat{Y}}$ is essentially redundant information between $\hat{Y}$ and $X_c$ which is exempted. 

Next, we revisit Canonical Example~\ref{cexample:equalized_odds} ($\hat{Y}=U_{X_1}+Z+U_{X_2}$ and $X_c=U_{X_1}$) where it is intuitive that the measure of non-exempt disparity should be non-zero. $\uni{Z}{\hat{Y} | X_{c}}$ is non-zero here (see Supporting Derivation 1 in Appendix~\ref{app:properties_supporting}), consistent with our intuition. As a proof sketch, recall the tabular representation in Fig.~\ref{fig:pid2}. $\rd{Z}{(\hat{Y}, X_c)}$ is the sub-volume of intersection between $\mut{Z}{X_c}$ and $\mut{Z}{\hat{Y}},$ and hence goes to zero because $\mut{Z}{X_c}=0$. This leads to $\uni{Z}{\hat{Y} | X_{c}}=\mut{Z}{\hat{Y}}$ which is non-zero here.

Lastly, $\uni{Z}{\hat{Y} | X_{c}}$ is also $0$ in Canonical Example~\ref{cexample:college} (counterfactually fair hiring) and the (similar) example of cancellation of paths in Remark~\ref{rem:cancellation}. More importantly, we note that, while conditional mutual information $\mut{Z}{\hat{Y} \given X_{c}}$ may be non-zero even if the the total disparity or counterfactual causal influence is $0$ (as in Canonical Example~\ref{cexample:college}), unique information is not. \emph{In Lemma~\ref{lem:uni_cci} in Appendix~\ref{app:pid_properties}, we show that $\uni{Z}{\hat{Y} | X_{c}}$ is always zero if the total disparity or counterfactual causal influence is $0$, i.e.,} $\mut{Z}{(\hat{Y},U_X)}=0$. In fact, $\uni{Z}{\hat{Y} | X_{c}}$ is a sub-volume or component of the previous candidate measure $\mut{Z}{\hat{Y} \given X_{c}}$, that is guaranteed to be $0$ if the total disparity is zero.

These examples serve as our rationale for the property of non-exempt statistically visible disparity, i.e., Property~\ref{propty:synergy} which states that $M_{NE}$ should be $0$ if $\uni{Z}{\hat{Y} | X_{c}}>0$. 
$\uni{Z}{\hat{Y} | X_{c}}$, however, is not sufficient as a candidate measure as it fails to capture \emph{non-exempt masked disparity}, as we will demonstrate in Canonical Example~\ref{cexample:masking}. Thus, Property~\ref{propty:synergy} is only a lower bound, i.e., sometimes $M_{NE}$ may still need to be non-zero even when $\uni{Z}{\hat{Y} | X_{c}}=0$.  Property~\ref{propty:synergy} only captures the non-exempt  \emph{statistically visible disparity} that cannot be accounted for by $X_c$ alone.


\begin{cexample}[Non-Exempt Masked Disparity in Hiring Ads I] An ad for a software-engineering job is only presented to men $(Z=1)$ with a coding-test score above a threshold $(U_{X_1}=1)$, and to women $(Z=0)$ with a coding-test score below a threshold $(U_{X_1}=0)$ with $Z$ and $U_{X_1}$ being \iid{} \ Bern(\nicefrac{1}{2}). Here, $X_c= U_{X_1}$ and $X_g=Z$. The model output is given by $\hat{Y}=Z \oplus U_{X_1}$. This example is shown in Fig.~\ref{fig:masking}.
\label{cexample:masking}
\end{cexample}

This model discriminates against half of the population (high-scoring women) for whom the ad may be relevant. This is also supported by the fact that that the total disparity $\mut{Z}{(\hat{Y},U_X)}> 0$. Intuitively, here a virtual constituent or proxy ($Z$) is formed inside the black-box model that influences the output and that is derived entirely from $X_g$. For such an example, it is desirable that the non-exempt disparity $M_{NE}$ should not be $0$. In fact, this example demonstrates that there may be non-exempt disparity even when the statistically visible disparity $\mut{Z}{\hat{Y}}=0$. Here, $\uni{Z}{\hat{Y} | X_{c}}$ fails to capture the masked disparity because it has to be zero whenever $\mut{Z}{\hat{Y}}=0$ (using \eqref{eq:pid2} in Section~\ref{subsec:background}).

Let us revisit the candidate measure $\mut{Z}{\hat{Y}\given X_c}$. This measure resolves all the examples discussed so far (\ref{cexample:exemption}-\ref{cexample:masking}) except giving a false positive conclusion in Canonical Example~\ref{cexample:college}. Notice that, $\mut{Z}{\hat{Y}\given X_c}$ is zero if and only if $Z-X_c-\hat{Y}$ form a Markov chain. While the Markov chain $Z-X_c-\hat{Y}$ may not always hold even when it is desirable for $M_{NE}$ to be zero as in Canonical Example~\ref{cexample:college}, we have seen that in all the examples so far (\ref{cexample:exemption}-\ref{cexample:masking}) where the Markov chain $Z-X_c-\hat{Y}$ holds, it has been desirable that $M_{NE}$ be zero (possible one-way implication). Assuming that the Markov chain $Z-X_c-\hat{Y}$ is a sufficient condition for $M_{NE}$ to be zero, we proposed the following property of non-exempt masked disparity in our prior work~\cite{dutta2020information}. 

\noindent \textit{$M_{NE}$ should be non-zero in the example of non-exempt masked disparity, i.e., Canonical Example~\ref{cexample:masking} even if $\mut{Z}{\hat{Y}}=0$. But, $M_{NE}$ should be $0$ if the Markov chain $Z-X_c-\hat{Y}$ holds.}

\begin{rem}[Relation to our prior work~\cite{dutta2020information}]In our prior work~\cite{dutta2020information}, this property, in conjunction with Properties \ref{propty:cancellation}, \ref{propty:synergy} and~\ref{propty:complete_exemption}, leads to a measure that quantifies only a sub-volume of $\mut{Z}{\hat{Y}\given X_c}$ that no longer gives false positive conclusion in Canonical Example~\ref{cexample:college} while still resolving all the other examples discussed so far. The measure proposed in \cite{dutta2020information} is essentially the information-theoretic sub-volume of the intersection between $\mut{Z}{\hat{Y}\given X_c}$ and total disparity $\mut{Z}{(\hat{Y},U_X)}$, which goes to $0$ whenever either of them is $0$ (details are provided in Appendix~\ref{app:properties_others})\footnote{One might also wonder why a measure of the form of a product, i.e., $M_{NE}= \mut{Z}{\hat{Y}\given X_c}\times \mut{Z}{(\hat{Y},U_X)}$ does not work instead. We discuss a counterexample for such a product measure in \cite{dutta2020information} that we also include in Appendix~\ref{app:properties_others} here for completeness.}. 
\end{rem}


The property of non-exempt masked disparity stated in \cite{dutta2020information} is built on the rationale that in the example of non-exempt masked disparity in hiring ads (Canonical Example~\ref{cexample:masking} where $\hat{Y}=Z\oplus U_{X_1}$), instead of $U_{X_1}$ being the coding-test score, if $U_{X_1}$ is a random coin flip used to randomize the race, then this scenario may not necessarily be regarded as non-exempt. Then, we would have $X_c=\phi$ and $X_g=(Z,U_{X_1})$, and the Markov chain $Z-X_c-\hat{Y}$ would hold, deeming this example as \emph{exempt}. In \cite{dutta2020information}, the goal was to only account for non-exempt masked disparity in $M_{NE}$ when the ``mask'' is either a critical feature or arises exclusively from the critical features, e.g., Canonical Example~\ref{cexample:masking} while any mask from the non-critical/general features were viewed more like these random coin flips. But what if the user wishes to also account for masked disparity if the mask is arising from $X_g$ as well, as demonstrated in the following modified version of the example?

\begin{cexample}[Non-Exempt Masked Disparity in Hiring Ads II] 
An ad for a job is only presented to men $(Z=1)$ with a coding-test score above a threshold $(U_{X_1}=1)$, and to women $(Z=0)$ with a coding-test score below a threshold $(U_{X_1}=0)$ with $Z$ and $U_{X_1}$ being \iid{} \ Bern(\nicefrac{1}{2}). The model output is given by $\hat{Y}=Z \oplus U_{X_1}$. Here, $Z \in X_g$ but $U_{X_1}$ is not be a critical feature for the job.\label{cexample:masking_general} 
\end{cexample}

Canonical Example~\ref{cexample:masking_general} with $X_c=\phi$ and $X_g=(Z,U_{X_1})$ will be deemed  \emph{exempt} by \cite{dutta2020information} because the Markov chain $Z-X_c-\hat{Y}$ holds. However, here the virtual constituent or proxy $Z$ is arising from $X_g$ and is being masked by another feature of $X_g$, i.e., $U_{X_1}$. If $U_{X_1}$ denotes coding-test score and $\hat{Y}$ denotes the decision of showing hiring ads, then the model is again unfair to high-scoring women. This argument is also supported by the fact that the total disparity is non-zero (not counterfactually fair). Since $X_c=\phi$, no disparity is exempt, and a measure of non-exempt disparity should ideally capture the total disparity in this model.

In this work, we would like to arrive at an alternate criterion (modification of the property of non-exempt masked disparity in \cite{dutta2020information}) that can capture non-exempt masked disparity irrespective of whether the ``mask'' arises from the critical or general features. What this means is that any scenario deemed exempt by the property of non-exempt masked disparity in \cite{dutta2020information} will also be deemed exempt by our modified property\footnote{We show in Lemma~\ref{lem:markov_chain} that the Markov chain in our modified property, i.e., $(Z,U_a)-X_c-(\hat{Y},U_b)$ also implies $Z-X_c-\hat{Y}$, but the opposite implication is not true.} but it is desirable that our modified property also accounts for scenarios, such as Canonical Example~\ref{cexample:masking_general}, that is sometimes deemed exempt by the former property even though intuitively, it may not be reasonable to do so.


\subsubsection{Leveraging Latent Variables to Understand Non-Exempt Masked Disparity}
\label{subsubsec:causality}

One commonality that we notice in the examples so far (\ref{cexample:exemption}-\ref{cexample:masking_general}) is that whenever it is desirable that $M_{NE}$ be zero, either there is no counterfactual causal influence of $Z$ on $\hat{Y}$ (i.e., $\mathrm{CCI}(Z\rightarrow \hat{Y})=0$) or the influence of $Z$ on $\hat{Y}$ has propagated \emph{only} along paths that pass through $X_c$. In scenarios where $\mathrm{CCI}(Z\rightarrow \hat{Y})\neq0$, one may choose to define another candidate measure of non-exempt disparity that is inspired from the notion of \emph{path-specific counterfactual fairness}~\cite{chiappa2018path} (also see \cite{kusner2017counterfactual,kilbertus2017avoiding}). This candidate measure for quantifying non-exempt disparity is a causal, path-specific quantification by varying $Z$ only along the paths through $X_{g}$ that do not pass through $X_c$ and comparing if it causes any change in the model output (also see Fig.~\ref{fig:path_specific}).  

\begin{figure*}
\centering
\begin{subfigure}[b]{0.7\linewidth}
\includegraphics[height=3.5cm]{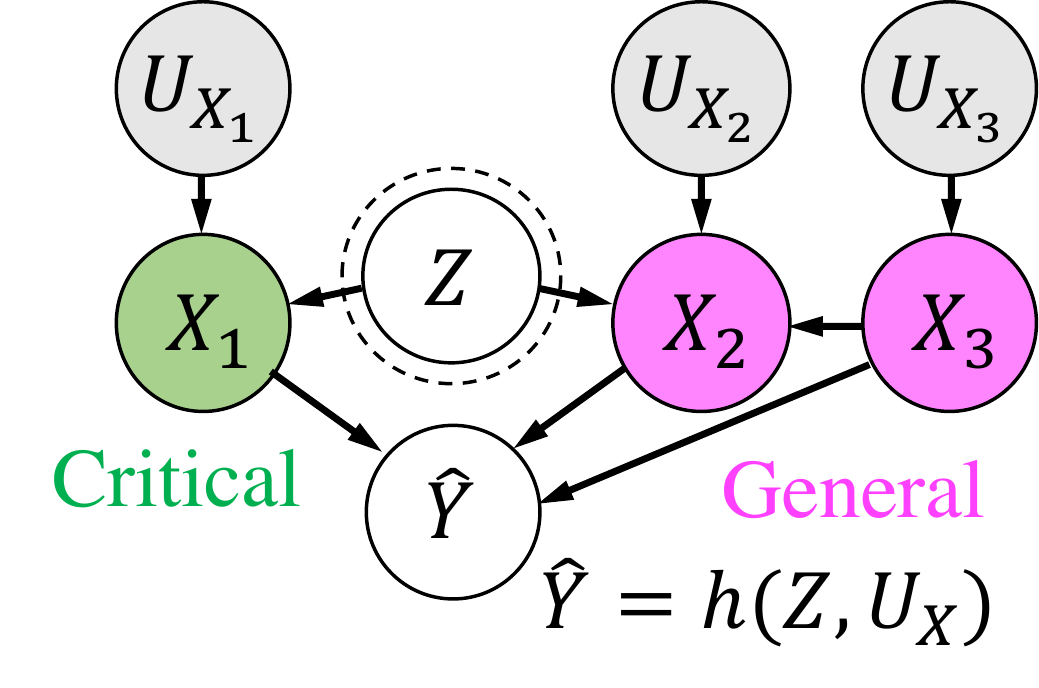}
\hspace{0.1cm}
\includegraphics[height=3.5cm]{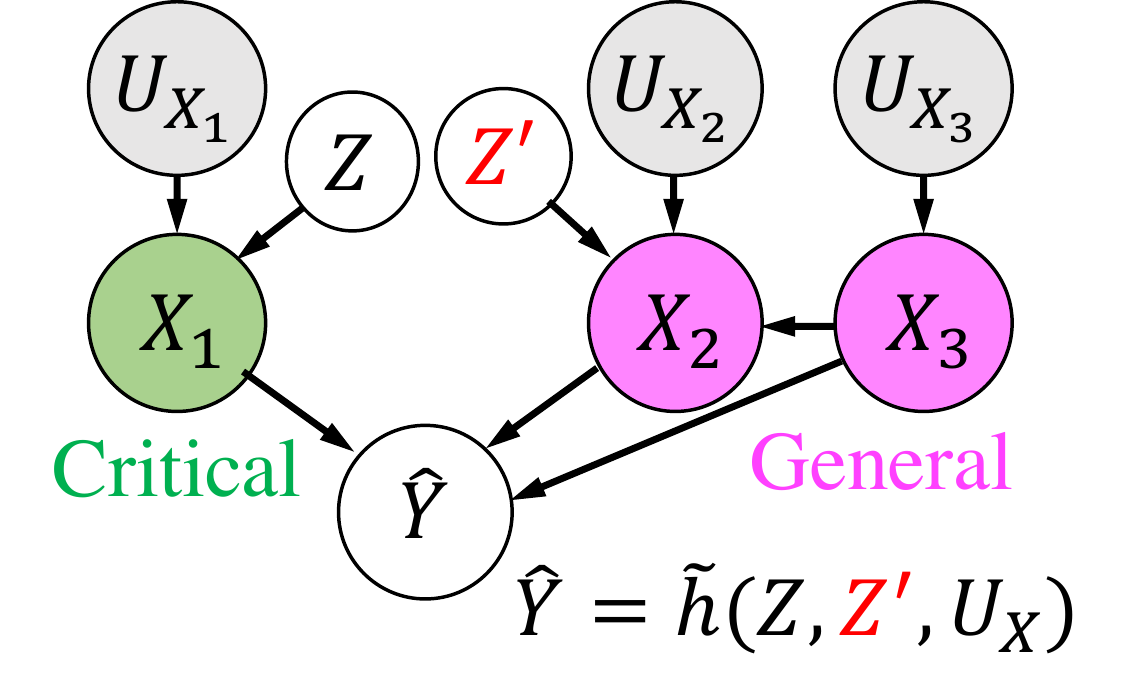}
\caption{Path-specific quantification of non-exempt disparity: (Left) Original model with output $h(Z,U_X)$. (Right) $Z$ is varied to $Z'$ along the direct paths through $X_g$ that do not pass through $X_c$ resulting in output $\tilde{h}(Z,Z',U_X)$. Candidate measure~\ref{candmeas:path_specific} quantifies the expected value of the change in output due to path-specific variation in $Z$.} \label{fig:path_specific}
\end{subfigure}
\hspace{0.2cm}
\begin{subfigure}[b]{0.25\linewidth}
\includegraphics[height=4cm]{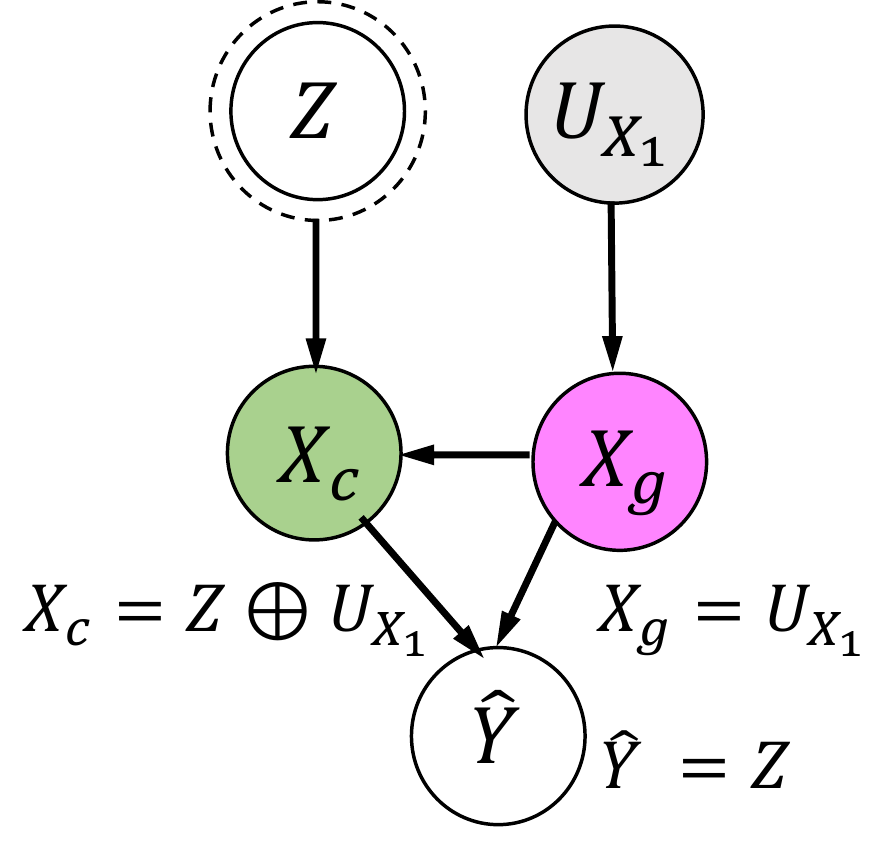}
\caption{Extreme Case of Disparity Amplification by Unmasking}\label{fig:path_specific_cexample}
\end{subfigure}
\caption{Path-specific quantification of non-exempt disparity (Candidate Measure~\ref{candmeas:path_specific}) and its limitation \label{fig:path}}
\end{figure*}

\begin{candmeas}
Let $\hat{Y}=h(Z,U_{X})$ in the true causal model. Assume a new causal graph with a new source node $Z'$ having an independent and identical distribution as $Z$ where we replace all relevant direct edges from $Z$ to $X_{g}$ with an edge from $Z'$ to $X_{g}$. Let $\hat{Y}=\tilde{h}(Z,Z',U_{X})$ in the new causal graph. A candidate measure is
$M_{NE} = \E{Z,Z',U_{X} }{|h(Z,U_{X}) - \tilde{h}(Z,Z',U_{X})| }.$
\label{candmeas:path_specific}
\end{candmeas}

This measure, when used in conjunction with $\mathrm{CCI}(Z\rightarrow \hat{Y})=0$, resolves the examples so far (\ref{cexample:exemption}-\ref{cexample:masking_general}). For Canonical Example~\ref{cexample:exemption}, it is zero and for Canonical Example~\ref{cexample:equalized_odds}, it is non-zero, as desired. For Canonical Example~\ref{cexample:college}, $\mathrm{CCI}(Z\rightarrow \hat{Y})=0$, and hence there is no need for a path-specific examination. For the example of non-exempt masked disparity (Canonical Examples~\ref{cexample:masking} and \ref{cexample:masking_general}), \emph{this measure is $0$ in spite of the statistically visible disparity $\mut{Z}{\hat{Y}}$ being $0$.} However, the following example exposes some of its limitations.

\begin{cexample}[Disparity Amplification by Unmasking]
Let $U_{X_1}$ be the inner ability of a candidate, and suppose that $X_{c}=Z + U_{X_1}$ denote the coding test score. Also let $X_{g}=U_{X_1}$ be the aptitude-test score where $Z$ and $U_{X_1}$ are \iid{} \  Bern(\nicefrac{1}{2}). Let the hiring decision be based on $\hat{Y}=X_{c}  - X_{g}=Z$. This is shown in Fig.~\ref{fig:synergy} with a more extreme modification in Fig.~\ref{fig:path_specific_cexample}. 
 \label{cexample:synergy}
 \end{cexample}

The disparity in this example will be deemed \emph{exempt} by a causal path-specific examination. However, this model has statistically visible disparity ($\mut{Z}{\hat{Y}}>0$) that cannot be attributed to $X_c$ alone. Following the PID literature, here $X_{c}$ and $X_{g}$ have synergistic information about $Z$ that ultimately appears in $\hat{Y}$ which in itself is the virtual constituent or proxy of $Z$ being formed in this model. This synergistic information cannot be attributed to $X_{c}$ alone because $\mut{Z}{X_c}$ is much smaller that $\mut{Z}{\hat{Y}}$. This is further supported by the argument that $X_{g}$ and $X_{c}$ together lead to a better estimate of $Z$ than $X_{c}$ alone which means $X_{g}$ is definitely a contributor to the disparity. Thus, $M_{NE}$ should be greater than $0$. Also, note that, here $\uni{Z}{\hat{Y} | X_{c}}>0$ (Supporting Derivation 2 in Appendix~\ref{app:properties_supporting}) because it is this ``joint'' information about $Z$ in $(X_c,X_g)$ that ultimately appears in $\hat{Y}$ that cannot be attributed to $X_c$ alone.

Ideally, we would like a property and a measure that captures the intuition in this example. From a causal perspective, here $U_{X_1}$ is a confounder~\cite{peters2017elements} to both $X_c$ and $\hat{Y}$, i.e., an extraneous variable that influences both of them along separate paths. A scenario when there is no non-exempt disparity would be: (i) All causal paths from $Z$ to $\hat{Y}$ in the SCM pass through $X_c$; and also (ii) No $U_{X_i}$ acts as a confounder for both $X_c$ and $\hat{Y}$. This leads to the intuition that to be able to say $M_{NE}=0$, one might be able to divide $U_X$ into two subsets $U_a$ and $U_b$ (further functional generalizations discussed in Section~\ref{sec:conclusion}), such that: (i) $U_a$ consists of the latent factors that do not influence $\hat{Y}$ at all, or influence it only through $X_c$ without acting as confounder; (ii) On the other hand, $U_b$ consists of the remaining latent factors, that only influence $\hat{Y}$ and not $X_c$; and (iii) The Markov chain $(Z,U_a)-X_c-(\hat{Y},U_b)$ holds.

To understand this better, we again revisit Canonical Example~\ref{cexample:exemption} (visualization in Fig.~\ref{fig:firemen}). Intuitively, the total disparity in this example is exempt because $Z$ was already masked by $U_{X_1}$ in $X_c$, and the mask remained untampered in the final output $\hat{Y}$ with only additional independent masks added inside the black-box model. Here, neither $Z-X_c-(\hat{Y},U_{X})$ nor $(Z,U_X)-X_c-\hat{Y}$ hold, but $(Z,U_{X_1})-X_c-(\hat{Y},U_{X_2})$ does. A Markov chain of the form $(Z,U_a)-X_c-(\hat{Y},U_b)$ also implies both the criterion $(Z,U_a)-X_c-\hat{Y}$ and $Z-X_c-(\hat{Y},U_b)$ (see Lemma~\ref{lem:markov_chain} with proof in Appendix~\ref{app:properties}). One can interpret $U_a$ as the latent variables that either do not influence $\hat{Y}$ at all or already mask $Z$ in $X_c$ and remain untampered in the final output $\hat{Y}$. On the other hand, $U_b$ consists of the remaining latent variables that contribute to ``additional masking inside the black-box model.''

This leads us to propose the following criterion for $M_{NE}$ that also serves as our main rationale for Property~\ref{propty:masking}: \textit{$M_{NE}$ should be $0$ if $(Z,U_a)-X_c-(\hat{Y},U_b)$ form a Markov chain for some subsets $U_a,U_b \subseteq U_X$ such that $U_a=U_X\backslash U_b$.}

\begin{restatable}{lem}{markovchain}
\label{lem:markov_chain}
The Markov chain $(Z,U_a)-X_c-(\hat{Y},U_b)$ implies that the following Markov chains also hold: (i) $Z-X_c-\hat{Y}$;  (ii)  $(Z,U_a)-X_c-\hat{Y}$; and (ii) $Z-X_c-(\hat{Y},U_b)$.
\end{restatable}


The Markov chain $(Z,U_a){-}X_c{-}(\hat{Y},U_b)$ holding implies $M_{NE}{=}0$, but the Markov chain not holding for all $U_a,U_b$ such that $U_a=U_X\backslash U_b$ does not necessarily imply that $M_{NE}\neq 0$. This criterion $(Z,U_a)-X_c-(\hat{Y},U_b)$ implying $M_{NE}=0$ only attempts to provide an upper bound on $M_{NE}$, i.e., it is desirable that 
$M_{NE}\leq \min_{U_a,U_b\text{ s.t. } U_a=U_X\backslash U_b} \mut{(Z,U_a)}{(\hat{Y},U_b)\given X_c} $ such that $U_a=U_X\backslash U_b.$ The measure $\min_{U_a,U_b\text{ s.t. } U_a=U_X\backslash U_b} \mut{(Z,U_a)}{(\hat{Y},U_b)\given X_c} $ does not suffice in itself as a measure of non-exempt disparity because it again does not satisfy Property~\ref{propty:cancellation}. To see this, notice that $\min_{U_a,U_b\text{ s.t. } U_a=U_X\backslash U_b} \mut{(Z,U_a)}{(\hat{Y},U_b)\given X_c} \geq \mut{Z}{\hat{Y}\given X_c} $ (see proof of Lemma~\ref{lem:markov_chain}), and thus, it also gives a false positive conclusion about non-exempt disparity in Canonical Example~\ref{cexample:college} (counterfactually fair hiring). Instead, $\uni{(Z,U_a)}{(\hat{Y},U_b)\given X_c}$ is a sub-component of $\mut{(Z,U_a)}{(\hat{Y},U_b)\given X_c}$ that satisfies Property~\ref{propty:cancellation}. Our desirable properties ultimately leads us to our proposed measure of non-exempt disparity, given by:
\begin{equation}
M^*_{NE}= \min_{U_a,U_b} \uni{(Z,U_a)}{(\hat{Y},U_b)| X_c}   \\ \text{such that } U_a=U_X\backslash U_b.
\end{equation}



\subsubsection{Our Proposed Measure Resolves all the Canonical Examples}

To develop intuition on what our proposed measure captures, we will now discuss how this measure resolves all of the examples in this work. We group ``similar'' examples together.

\begin{itemize}[leftmargin=*]

\item \textbf{Scenarios where total disparity $\mut{Z}{(\hat{Y},U_X)}$ is zero:}
This applies to Canonical Example~\ref{cexample:college} and the related example in Remark~\ref{rem:cancellation}. Because $\min_{U_a,U_b\text{ s.t. } U_a=U_X\backslash U_b} \uni{(Z,U_a)}{(\hat{Y},U_b)| X_c}\leq \uni{Z}{(\hat{Y},U_X)| X_c} \leq \mut{Z}{(\hat{Y},U_X)}$ (see proof of Theorem~\ref{thm:satisfythm} in Appendix~\ref{app:properties}), it satisfies Property~\ref{propty:cancellation} and goes to $0$ whenever total disparity is $0$.

\item \textbf{Scenarios where $Z$ is already masked in $X_c$ and remains so in the output (with or without additional independent masks):}
This applies to Canonical Example~\ref{cexample:exemption}. We will examine the value of $\uni{(Z,U_a)}{(\hat{Y},U_b)| X_c} $ for different choices of $U_a \subseteq U_X$ to find the minimum. First notice that, if $U_a=\phi$ (and $U_b=U_X$), we have 
\begin{equation}
\uni{(Z,U_a)}{(\hat{Y},U_b)| X_c}=\uni{Z}{(\hat{Y},U_X)| X_c}\\ \overset{(a)}{\geq} \uni{Z}{Z| X_c} > 0
\end{equation}
(see Supporting Derivation 3 in Appendix~\ref{app:properties_supporting}; (a) holds from a monotonicity property of unique information because $Z$ can be obtained from deterministic local operations on $(\hat{Y},U_X)$). This is in agreement with the intuition that $U_{X_1}$ should not belong to the set of candidate masks ($U_b$) that need to be accounted for. Next, if $U_a=U_{X_1}$ (and $U_b=U_{X_2}$), we have $\uni{(Z,U_a)}{(\hat{Y},U_b)| X_c}=0$ (implied from the Markov chain $(Z,U_{X_1})-X_c-(\hat{Y},U_{X_2})$). Since unique information is non-negative, we therefore have 
$\min_{U_a,U_b\text{ s.t. } U_a=U_X\backslash U_b} \uni{(Z,U_a)}{(\hat{Y},U_b)| X_c}=0.$ In essence, the pair $(U_a^*,U_b^*)$ that minimizes $\uni{(Z,U_a)}{(\hat{Y},U_b)| X_c}$ is such that $U^*_a= U_{X_1}$, and the candidate masks that need to be accounted for, i.e., $U^*_b=U_{X_2}.$ 


Now, what happens to the value of $\uni{(Z,U_a)}{(\hat{Y},U_b)| X_c}$ if the accountable mask $U_{X_2}$ is instead in $U_a$? We have 
\begin{align}
\uni{(Z,U_a)}{(\hat{Y},U_b)| X_c}  \overset{(a)}{\geq}  \uni{U_{X_2}}{\hat{Y}| X_c} \overset{(b)}{=}\mut{U_{X_2}}{\hat{Y}},
\end{align}
which is strictly greater than $0$. This agrees with the intuition that $U_{X_2}$ should belong to the candidate set of masks that one should account for ($U_b$). Here (a) holds using two monotonicity properties of unique information (see Properties~\ref{lem:monotonicity_bob} and \ref{lem:monotonicity_alice} in Appendix~\ref{app:pid_properties}) and (b) holds because $\mut{U_{X_2}}{X_c}=0,$ leading to $\rd{U_{X_2}}{(\hat{Y},X_c)}=0.$

\item \textbf{Scenarios where non-exempt statistically visible disparity is present, i.e., $\uni{Z}{\hat{Y} | X_c}>0$:}
This applies to Canonical Example~\ref{cexample:equalized_odds} and Canonical Example~\ref{cexample:synergy}. Because $\uni{(Z,U_a)}{(\hat{Y},U_b)| X_c} \geq \uni{Z}{\hat{Y} | X_c}$ (see proof of Theorem~\ref{thm:satisfythm} in Appendix~\ref{app:properties}), our proposed $M_{NE}^*$ satisfies Property~\ref{propty:synergy}, and is thus non-zero whenever $\uni{Z}{\hat{Y} | X_c}>0$.

\item \textbf{Scenarios where non-exempt masked disparity is present:} This applies to Canonical Example~\ref{cexample:masking} and Canonical Example~\ref{cexample:masking_general}. In the proof of Theorem~\ref{thm:satisfythm} in Appendix~\ref{app:properties}, we show that the proposed measure satisfies Property~\ref{propty:masking} (non-exempt masked disparity), and is thus non-zero for these canonical examples of non-exempt masked disparity.

We note that Canonical Example~\ref{cexample:equalized_odds} is an interesting case where both non-exempt statistically visible disparity and non-exempt masked disparity are present. Here, $M_{NE}^*$ is strictly greater than the non-exempt statistically visible disparity $(\uni{Z}{\hat{Y} | X_c}),$ and this difference can be interpreted as a quantification of the non-exempt masked disparity. First notice that,
\begin{align}
 \uni{Z}{\hat{Y} | X_c}  & \overset{(a)}{=} \mut{Z}{\hat{Y}} = \mathrm{H}(Z)-\mathrm{H}(Z|\hat{Y}) = \mathrm{H}(Z)-\mathrm{H}(Z|U_{X_1}+Z+U_{X_2}) = 1-\frac{3}{4}h_b(\nicefrac{1}{3}) \text{ bits}.
\end{align}
The full derivation is in Supporting Derivation 4 in Appendix~\ref{app:properties_supporting}. Here $h_b(\cdot)$ is the binary entropy function~\cite{cover2012elements} given by $h_b(p)=-p\log_2(p)-(1-p)\log_2(1-p)$ and (a) holds because $\mut{Z}{U_{X_1}}=0$, implying $\rd{Z}{(\hat{Y},U_{X_1})}=0$ as well. Now, we will examine the value of $\uni{(Z,U_a)}{(\hat{Y},U_b)| X_c}$ for different choices of $U_a$ to find the minimum. The full derivation for all of these cases is in Supporting Derivation 4 in Appendix~\ref{app:properties_supporting}. Here, we only mention the key step.
Let $U_a=\phi$ (and $U_b=U_X$). Then, 
\begin{align}
 \uni{(Z,U_a)}{(\hat{Y},U_b) | X_c} 
&= \uni{Z}{(\hat{Y},U_{X_1},U_{X_2}) | U_{X_1}} \overset{(a)}{=} \mut{Z}{U_{X_1}+Z+U_{X_2},U_{X_1},U_{X_2}}  = 1 \text{ bit}.
\end{align}
 Here (a) holds again because $\mut{Z}{U_{X_1}}=0$, implying the redundant information is $0$ as well (using \eqref{eq:pid2} in Section~\ref{subsec:background}). Next, for $U_a=U_{X_2}$ (and $U_b=U_{X_1}$), we have,
\begin{align}
\uni{(Z,U_a)}{(\hat{Y},U_b) | X_c} 
& = \uni{(Z,U_{X_2})}{(\hat{Y},U_{X_1}) | U_{X_1}}  \overset{(a)}{=} \mut{(Z,U_{X_2})}{(\hat{Y},U_{X_1})} = \nicefrac{3}{2} \text{ bit}.
\end{align}
Here (a) holds again because $\mut{(Z,U_{X_2})}{U_{X_1}}=0$, implying the redundant information is $0$ as well. Next, for $U_a=U_{X_1}$ (and $U_b=U_{X_2}$), we have,
\begin{align}
\uni{(Z,U_a)}{(\hat{Y},U_b) | X_c}
&= \uni{(Z,U_{X_1})}{(\hat{Y},U_{X_2}) | U_{X_1}} \overset{(b)}{=} \mut{(Z,U_{X_1})}{(\hat{Y},U_{X_2})\given U_{X_1}} = 1 \text{ bit}.
\end{align}
Here (b) holds because $\syn{(Z,U_{X_1})}{(A,B)}=0$ if one of the terms $A$ or $B$ is a deterministic function of $(Z,U_{X_1})$ (using Lemma~\ref{lem:zero_syn} in Appendix~\ref{app:pid_properties}) and hence unique information becomes equal to the conditional mutual information (see \eqref{eq:pid3} in Section~\ref{subsec:background}).  Lastly, for $U_a=U_{X}$ (and $U_b=\phi$), we have,
\begin{align}
\uni{(Z,U_a)}{(\hat{Y},U_b) | X_c} 
&= \uni{(Z,U_{X_1},U_{X_2})}{\hat{Y} | U_{X_1}}  \overset{(b)}{=} \mut{(Z,U_{X_1},U_{X_2})}{\hat{Y} \given U_{X_1}}= \nicefrac{3}{2} \text{ bit}.
\end{align}
Here (b) holds again using Lemma~\ref{lem:zero_syn} in Appendix~\ref{app:pid_properties}. Thus, we obtain that, 
\begin{equation}
M_{NE}^*= \min_{U_a,U_b\text{ s.t. } U_a=U_X\backslash U_b}\uni{(Z,U_a)}{(\hat{Y},U_b) | X_c} =1 \text{ bit},
\end{equation}
which is strictly greater than $\uni{Z}{\hat{Y} | X_c} =1-\frac{3}{4}h_b(\nicefrac{1}{3}) \text{ bits},$ accounting for both non-exempt statistically visible and non-exempt masked disparities.
\end{itemize}

As noted in Remark~\ref{rem:uniqueness}, our properties are insufficient to arrive at a unique functional form for the measure of non-exempt disparity. It is easiest to understand this issue by contrasting it with Shannon's discussion on entropy as a measure for uncertainty. First, we do not have a counterpart of ``additivity'' of entropy (see Property 3 in Section 6 of~\cite{shannon1948mathematical}) which allows Shannon to arrive at the logarithmic scaling in entropy. Second, we also do not provide an operational meaning for this measure (such as that provided by the lossless source coding theorem for entropy~\cite{cover2012elements}), which further supports the logarithmic scaling.  
 This is a direction of meaningful future work (further functional generalizations discussed in Section~\ref{sec:conclusion}). We note that this is the case with almost all existing measures of fairness (with the notable exceptions of~\cite{liao2019learning,issa2019operational,liao2019robustness}). Exploring more deeply the desirable attributes of the influence of a virtual constituent or proxy of $Z$ that influences the model output and that cannot be attributed to the critical features $X_c$ alone (inspired from the work on proxy-use~\cite{datta2017use}) could be a starting point towards deriving an exact operational meaning for our proposed measure. Nonetheless, our measure does satisfy all six desirable properties, and also captures important nuances of the problem, e.g., both non-exempt masked disparity and non-exempt statistically visible disparity when they are present together (revisited in Section~\ref{sec:measure_decomposition}). Our examples also help us understand the utility and limitations of some existing measures that have some provision for exemptions, as we discuss next.

\subsection{Understanding Existing Measures of Fairness with Provision for Exemptions}
\label{subsec:contrast}

\noindent \textbf{Conditional Statistical Parity:} This definition~\cite{corbett2017,debiasing} is equivalent to $\mut{Z}{\hat{Y}\given X_c}=0$. Therefore, it has similar utility and limitations as Candidate Measure~\ref{candmeas:CMI} ($\mut{Z}{\hat{Y}\given X_c}$). It resolves some limitations of both statistical parity and equalized odds. However, it gives a false positive conclusion in detecting non-exempt disparity in Canonical Example~\ref{cexample:college} (the example of counterfactually fair hiring), where there is no causal influence of $Z$ on $\hat{Y}$ but $\mut{Z}{\hat{Y}\given X_c}>0$. Because this is an observational measure, it is not able to distinguish between scenarios where there is causal influence of $Z$ on $\hat{Y}$ (non-exempt masked disparity in hiring ads; Canonical Example~\ref{cexample:masking}) and where there is not (Canonical Example~\ref{cexample:college}), even if $\mut{Z}{\hat{Y}\given X_c}>0$ in both (elaborated further in relation to our impossibility result in Remark~\ref{rem:impossibility} Section~\ref{sec:impossibility}). It also fails to capture non-exempt masked disparity when the mask arises from the general features as in Canonical Example~\ref{cexample:masking_general}.

\noindent  \textbf{Justifiable Fairness:} A model is said to be justifiably fair~\cite{interventional_fairness} if $\mut{Z}{\hat{Y}\given X_{s}}=0$ for all sets $X_s \subseteq X$ such that $X_c\subseteq X_s.$ This measure addresses several concerns of the previously stated measures, including capturing several forms of non-exempt masked disparity. However, it also gives false positive conclusion in Canonical Example~\ref{cexample:college} (counterfactually fair college admissions), which shows no causal influence of $Z$ on $\hat{Y}$ but $\mut{Z}{\hat{Y}\given X_c}>0$. Because this is an observational measure, it is not able to distinguish between scenarios where there is causal influence of $Z$ on $\hat{Y}$ and where there is not, even if $\mut{Z}{\hat{Y}\given X_c}>0$ in both (elaborated further in relation to our impossibility result in Remark~\ref{rem:impossibility} Section~\ref{sec:impossibility}).

Another limitation of such an individual feature-based conditioning arises when the causal effects of both $Z$ and an independent latent factor are present in the same feature, e.g., different digits of a zip-code, and it is not known in advance whether to condition on the entire zip-code or its sub-portions like the individual digits.

\begin{scenario}[Special Case of Canonical Example~\ref{cexample:masking_general}] Let $X_g=[Z,U_{X_1}]$ be a single multivariate feature, e.g., two bits of a number and $X_c=\phi$, and the output be $\hat{Y}=Z\oplus U_{X_1}$ where $Z$ and $U_{X_1}$ are \iid{} Bern(\nicefrac{1}{2}). 
\end{scenario}

In this example, as long as one treats $X_g$ as a single feature, the model will be deemed \emph{justifiably fair} because $\mut{Z}{\hat{Y}\given X_{g}}=0$ and $\mut{Z}{\hat{Y}}=0$. But, this is a case of non-exempt masked disparity. It is necessary to have an advance suspicion of this possible nature of the true SCM to be able to condition on the two bits of $X_g$ separately. This definition captures the non-exempt masked disparity in this example if the sub-portions of any single feature are defined in advance.

\noindent  \textbf{Path-Specific Counterfactual Fairness:} Path-specific counterfactual fairness~\cite{chiappa2018path} is a purely causal notion of fairness which exempts the causal influence of $Z$ along selected paths. Based on this idea, we proposed Candidate Measure~\ref{candmeas:path_specific} in Section~\ref{subsec:rationale}. However, Canonical Example~\ref{cexample:synergy} (the example of discrimination by unmasking) captures some of its limitations, when there is synergistic or joint information about $Z$ present in $X_c$ and $X_g$ that appears in $\hat{Y}$ that cannot be attributed to any one of them alone. Furthermore, sometimes the influence of $Z$ can cancel along two paths so that the final output has no influence of $Z$, e.g., the example in Remark~\ref{rem:cancellation}. For such scenarios, this measure alone can lead to false positive conclusions about non-exempt disparity, and might need to be used in conjunction with a measure of total disparity (e.g., $\ci{Z}{\hat{Y}}$).


\section{Understanding the overall decomposition}
\label{sec:measure_decomposition}

In this section, we demonstrate how our proposed quantification enables a \emph{non-negative} information-theoretic decomposition of the total disparity $\mut{Z}{(\hat{Y},U_X)}$ into four components, that can be interpreted as: statistically visible non-exempt disparity, statistically visible exempt disparity, masked non-exempt disparity and masked exempt disparity (also see Fig.~\ref{fig:decomp}). 
\begin{figure*}
\centering
\begin{subfigure}[b]{0.40\linewidth}
\centering
\includegraphics[height=4cm]{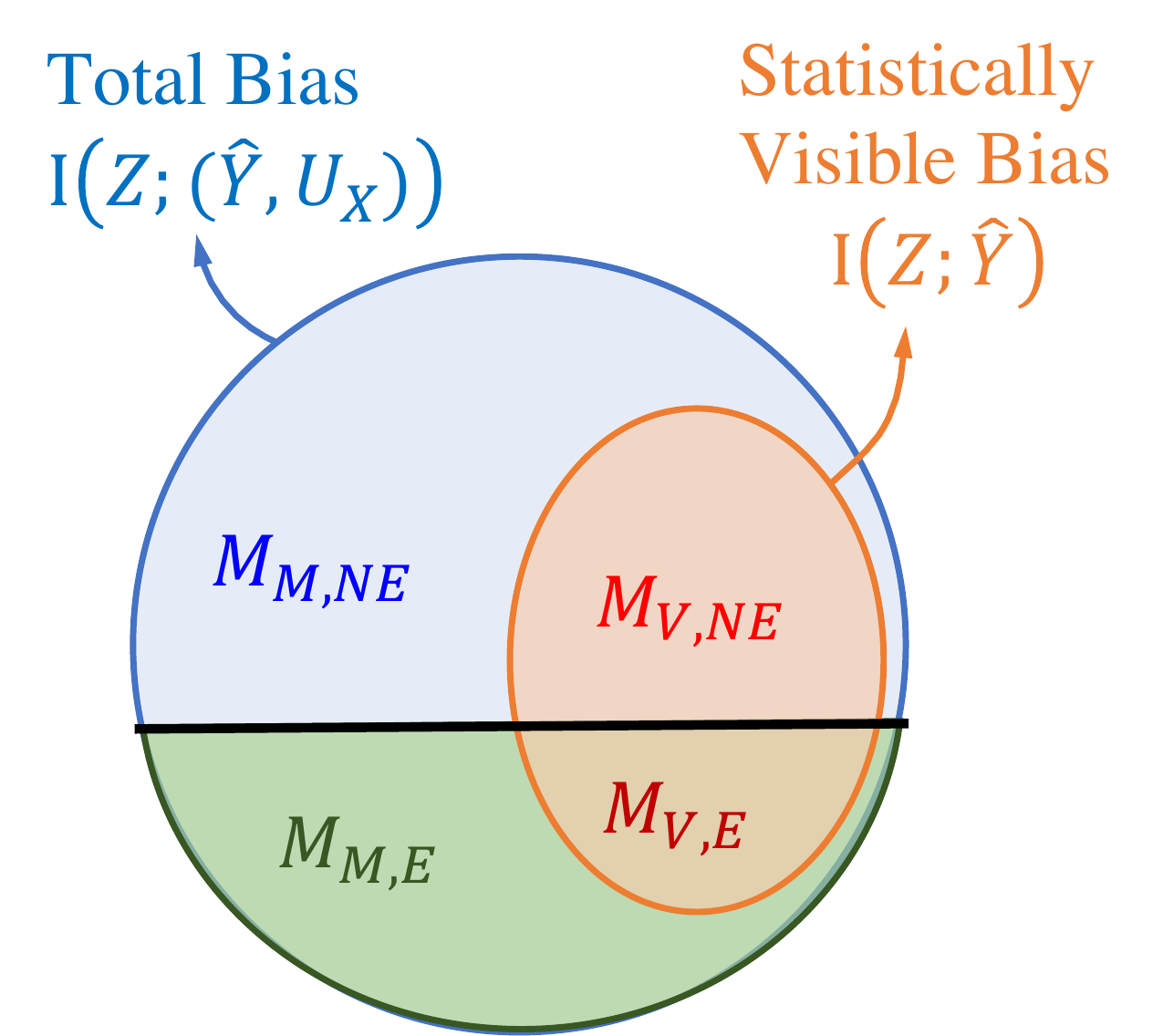}
\caption{Venn diagram representation of overall decomposition}
\end{subfigure}
\hspace{0.1cm}
\begin{subfigure}[b]{0.55\linewidth}
\centering
\includegraphics[height=4cm]{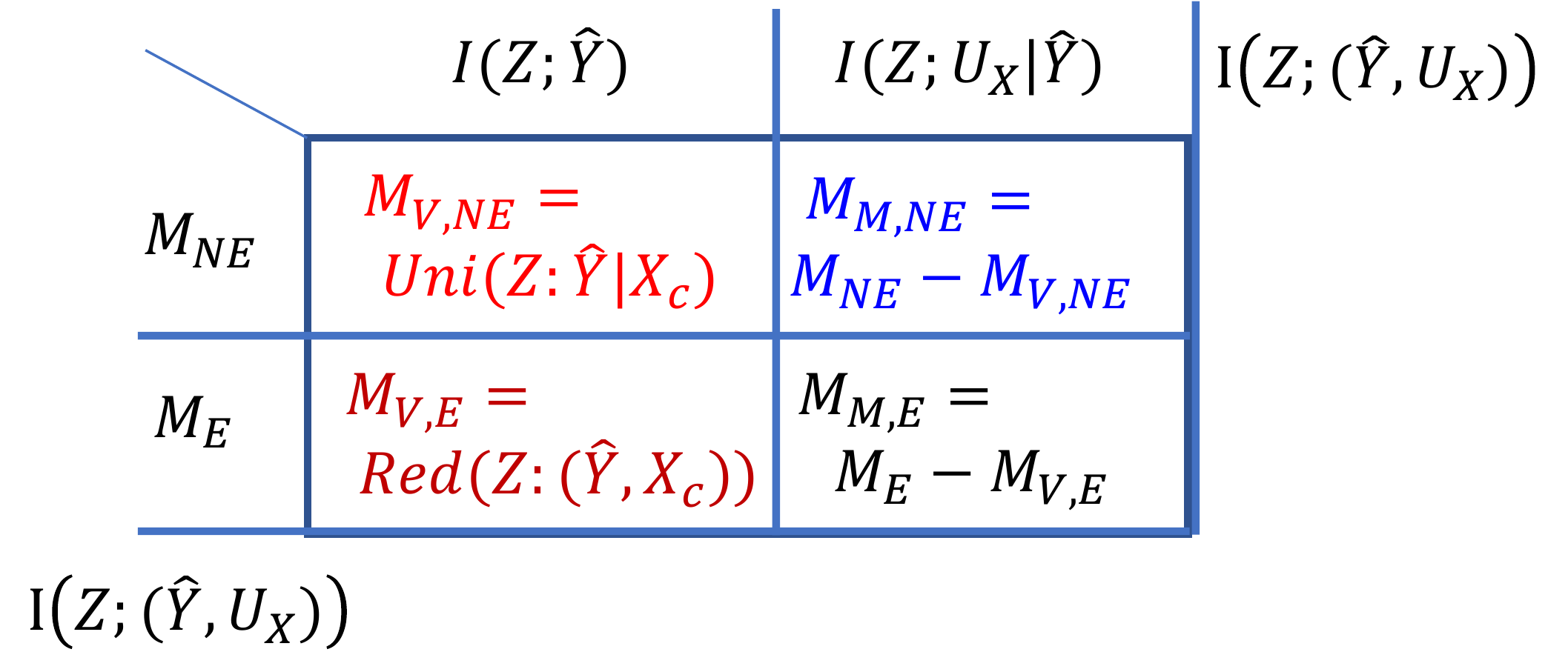}
\caption{Tabular representation of overall decomposition}
\end{subfigure}
\caption{Overall decomposition of total disparity $\mut{Z}{(\hat{Y},X_c)}$ into four non-negative components, namely, non-exempt visible disparity $M_{V,NE}$, exempt visible disparity $M_{V,E}$, non-exempt masked disparity $M_{M,NE}$ and exempt masked disparity $M_{M,E}$. \label{fig:decomp}}
\end{figure*}

\begin{restatable}[Non-negative Decomposition of Total Disparity]{thm}{decompositionthm}  
The total disparity can be decomposed into four components as follows:
\begin{align}
\mathrm{I}(Z; (\hat{Y},U_X))=M_{V,NE} + M_{V,E} + M_{M,NE}+M_{M,E}.
\end{align}
Here $M_{V,NE}=\uni{Z}{\hat{Y}| X_{c}}$ and $M_{V,E}=\rd{Z}{(\hat{Y}, X_{c})}$. These two terms add to form $\mut{Z}{\hat{Y}}$ which is the total statistically visible disparity. Next, $M_{M,NE}= M^*_{NE} - M_{V,NE}$ where $M^*_{NE}$ is our proposed measure of non-exempt disparity (Definition~\ref{defn:non_exempt_disparity}), and $M_{M,E}=\mathrm{I}(Z; \hat{Y},U_X)-\mathrm{I}(Z; \hat{Y})-M_{M,NE}$. All of these components are non-negative.
\label{thm:measure_decomposition}  
\end{restatable}
The decomposition of total disparity into a summation of these four terms is trivial. What remains to be shown is that these four terms are non-negative (details provided in Appendix~\ref{app:measure_decomposition}). 


\noindent \textbf{Interpretation of the four components:} Here $M_{V,NE}=\uni{Z}{\hat{Y}| X_{c}}$ can be interpreted as the non-exempt statistically visible disparity (as also motivated in Section~\ref{subsec:rationale}). The remaining part of the statistically visible disparity (recall Definition~\ref{defn:visible_disparity}), i.e., $\mut{Z}{\hat{Y}}-\uni{Z}{\hat{Y}| X_{c}}=\rd{Z}{(\hat{Y}, X_{c})}$ then becomes the exempt statistically visible disparity ($M_{V,E}$). This also agrees with the intuition that redundant information about $Z$ visible in both $\hat{Y}$ and $Z$ represents the exempt statistically visible disparity. 

Now that we have a measure of non-exempt disparity ($M^*_{NE}$) and a measure of non-exempt statistically visible disparity ($M_{V,NE}$), we can interpret their difference as the non-exempt masked disparity, i.e., $M_{M,NE}= M^*_{NE} - M_{V,NE}=M^*_{NE}-\uni{Z}{\hat{Y}| X_{c}}$. It also agrees with the intuition that non-exempt masked disparity is the part of non-exempt disparity that $\uni{Z}{\hat{Y}| X_{c}}$ alone fails to capture. For instance, recall Canonical Example~\ref{cexample:masking} where $\hat{Y}=Z\oplus U_{X_1}$ and $X_c=U_{X_1}$. Here, $\mut{Z}{\hat{Y}}=0$, implying $M_{V,NE}=\uni{Z}{\hat{Y}| X_{c}}=0$. But, $M^*_{NE}=1$ bit (supporting derivation in Appendix~\ref{app:properties}; see the proof of Theorem~\ref{thm:satisfythm} under Property~\ref{propty:masking}). Therefore, the non-exempt masked disparity $M_{M,NE}= M^*_{NE} - M_{V,NE}=1$ bit here, which is in agreement with our intuition of non-exempt masked disparity.  Lastly, the remaining component $M_{M,E}=\mathrm{I}(Z; \hat{Y},U_X)-\mathrm{I}(Z; \hat{Y})-M_{M,NE}$ is interpreted as the exempt masked disparity. For instance, recall Canonical Example~\ref{cexample:exemption} where $\hat{Y}=X_c=Z+U_{X_1}+U_{X_2}$ with $Z,U_{X_1},U_{X_2}\sim \iid{}$ Bern(\nicefrac{1}{2}). Here, the total disparity $\mathrm{I}(Z; \hat{Y},U_X)=1$ bit, but the statistically visible disparity $\mathrm{I}(Z; \hat{Y})=0.5$ bits which means that there is masked disparity present. Our intuition is that this masked disparity should be entirely exempt because there is no non-exempt disparity in this example. This is in agreement with the value that we obtain, i.e., $M_{M,E}=\mathrm{I}(Z; \hat{Y},U_X)-\mathrm{I}(Z; \hat{Y})-M_{M,NE}=0.5$ bits. This is because $M_{M,NE}$ and $M_{V,NE}$ are both non-negative sub-components of $M^*_{NE}$, and $M^*_{NE}=0$ (from the Markov chain $(Z,U_{X_1},U_{X_2})-X_c-\hat{Y})$).

\begin{rem}[On conditioning to capture masked disparity] 
\label{rem:conditioning1} Conditioning on a random variable $G$ leading to $\mathrm{I}(Z;\hat{Y}\given G)>I(Z;\hat{Y})$ can sometimes detect masked disparity, if conditioning exposes more disparity than what was already visible. For example, $I(Z;\hat{Y}\given X_{c})$ can detect masked disparity if the mask is of the form $g(X_{c})$, e.g., in Canonical Example~\ref{cexample:masking} (a special case of the canonical example of masking with $X_c=U_{X_1}$ and $\hat{Y}=Z\oplus U_{X_1}$). However, conditioning on any random variable $G$ leading to $\mathrm{I}(Z;\hat{Y}\given G)>I(Z;\hat{Y})$  cannot always be interpreted as a case of masked disparity because this can sometimes lead to a false positive conclusion in detecting masked disparity,~e.g., in Canonical Example~\ref{cexample:college} where $\hat{Y}=U_{X_1}$ and $X_c= Z+U_{X_1}$. If $G$ is chosen as $X_c$, then $\mathrm{I}(Z;\hat{Y}\given X_c)>I(Z;\hat{Y})$ even though there is no disparity here at all (recall $\mathrm{CCI}(Z\rightarrow \hat{Y})=0$). For completeness, we therefore include another result here (Lemma~\ref{lem:maskingequivalence}) that clarifies when conditioning can correctly capture masked disparity.
\end{rem}

\begin{restatable}[Conditioning to Capture Masked Disparity]{lem}{maskingequivalence} The following two statements are equivalent:
\begin{itemize}[leftmargin=*,itemsep=0pt,topsep=0pt]
\item Masked disparity $\mathrm{I}(Z; (\hat{Y}, U_X))  - \mathrm{I}(Z;\hat{Y})>0$.
\item  $\exists$ a random variable $G$ of the form $G=g(U_{X})$ such that
$\mathrm{I}(Z;\hat{Y}\given G)- \mathrm{I}(Z;\hat{Y})>0$. 
\end{itemize}
\label{lem:maskingequivalence}
\end{restatable}

Without knowledge of the true causal model, such a $G=g(U_{X})$ may be difficult to determine from observational data alone, because the observational data can be a function of both $Z$ and $U_X$. This serves as the motivation behind our impossibility result on observational measures, that we state next.


\section{Impossibility Result}
\label{sec:impossibility}
\begin{restatable}[Impossibility of Observational Measures]{thm}{impossibilitythm} No observational measure of non-exempt disparity simultaneously satisfies all six desirable properties.
\label{thm:impossibility}
\end{restatable}

\begin{proof}[Proof of Theorem~\ref{thm:impossibility}]
Observe the two examples here:

\begin{example}[A Case of No Disparity] Let $X_{c}=Z \oplus U_{X_1}$, $X_{g}=Z$ and $\hat{Y}=X_{c} \oplus X_{g}= U_{X_1}$ where $Z$ and $U_{X_1}$ are both independent and identically distributed as Bern(\nicefrac{1}{2}).
\label{example:no_discrim}
\end{example}

\begin{example}[A Case of Non-Exempt Disparity] Let $X_{c}= U_{X_1}$, $X_{g}=Z$ and $\hat{Y}=X_{c} \oplus X_{g}= Z \oplus U_{X_1}$ where $Z$ and $U_{X_1}$ are both independent and identically distributed as Bern(\nicefrac{1}{2}).
\label{example:masked_discrim}
\end{example}

In Example~\ref{example:no_discrim}, the influences of $Z$ cancel each other and there is no total disparity. So, the non-exempt disparity should be zero by Property~\ref{propty:cancellation} (Zero Influence). However, Example~\ref{example:masked_discrim} is the canonical example of non-exempt masked disparity where there is non-exempt disparity present, and hence the non-exempt disparity should be non-zero by Property~\ref{propty:masking} (Non-Exempt Masked Disparity). But, for both of these examples, the joint distribution of the observables $(Z,X_{c},X_{g},\hat{Y})$ is the same which means that no observational measure can distinguish between these two cases. This proves the result.
\end{proof}

\begin{rem}[Alternative Examples] In fact, we can show that no observational measure can satisfy Property~\ref{propty:masking}. Consider a scenario of no disparity given by: $X_c=\phi$, $X_g=(Z\oplus U_{X_1}, Z)$ and $\hat{Y}=U_{X_1}$. For this example, the Markov chain $Z-X_c-(\hat{Y},U_{X_1})$ holds implying that $M_{NE}=0$ by Property~\ref{propty:masking}. Alternatively, consider a scenario of non-exempt disparity given by: $X_c=\phi$, $X_g=(U_{X_1}, Z)$ and $\hat{Y}=Z \oplus U_{X_1}$ which is again a variant of the canonical example of non-exempt masked discrimination. Let $Z$ and $U_{X_1}$ be independent and identically distributed as Bern(\nicefrac{1}{2}). Then, no purely observational measure can distinguish between these two scenarios because $(Z,X_c,X_g,\hat{Y})$ have the same joint distribution.
\end{rem}

\begin{rem}[Revisiting Conditional Statistical Parity and Justifiable Fairness]
For both Examples~\ref{example:no_discrim} and \ref{example:masked_discrim}, we observe that conditional mutual information $\mut{Z}{\hat{Y}\given X_c}>0$. Because $\mut{Z}{\hat{Y}\given X_c}$ is an observational measure, it fails to distinguish between whether there is causal influence of $Z$ or not in $\hat{Y}$. Existing observational definitions of fairness, e.g., conditional statistical parity and justifiable fairness would also not be able to distinguish between these two examples. One needs counterfactual measures to be able to distinguish between them, such as the counterfactual measure proposed in this work. 
\label{rem:impossibility}
\end{rem}

Nevertheless, because counterfactual measures are difficult to realize in practice, we examine the following observational measures of non-exempt disparity that satisfy only a few of Properties 1-6. 


\section{Observational Relaxations of our Proposed Counterfactual Measure: Utility and Limitations}
\label{sec:observational}

In this section, we propose three observational measures of non-exempt disparity and discuss their utility and limitations.

\noindent \textbf{Observational Measure 1.} $M_{NE}=\uni{Z}{\hat{Y}| X_{c}}.$

\noindent \textbf{Utility:} This measure satisfies several desirable properties as stated here: 
\begin{restatable}{lem}{unisatisfyproperties}[Fairness Properties of $\uni{Z}{\hat{Y} | X_{c}}$] The measure $\uni{Z}{\hat{Y} | X_{c}}$ satisfies Properties \ref{propty:cancellation}, \ref{propty:synergy}, \ref{propty:nonincreasing}, and \ref{propty:complete_exemption}.
\label{lem:uni_satisfy_properties}
\end{restatable}

The proof is in Appendix~\ref{app:observational}. Importantly, note that, $\uni{Z}{\hat{Y} | X_{c}}$ satisfies Property~\ref{propty:cancellation} which $\mut{Z}{\hat{Y}\given  X_{c}}$ does not (recall Canonical Example~\ref{cexample:college}). Thus, $\uni{Z}{\hat{Y} | X_{c}}$ does not give false positive conclusions in detecting non-exempt disparity if a model is counterfactually fair. 

This measure may be preferred over our other observational measures when one wants to prioritize avoiding false positive quantification of non-exempt disparity when a model is counterfactually fair. Recall that, $\uni{Z}{\hat{Y} | X_{c}}$ is a measure of non-exempt, statistically visible disparity. \emph{It correctly captures the entire non-exempt disparity when non-exempt masked disparity is absent.}

\noindent \textbf{Limitations:}
It does not quantify any non-exempt masked disparity (Property~\ref{propty:masking}). This is because $\uni{Z}{\hat{Y} | X_{c}}$ is a sub-component of the statistically visible disparity $\mut{Z}{\hat{Y}}$, and hence always goes to $0$ whenever the statistically visible disparity $\mut{Z}{\hat{Y}}=0$ (recall Canonical Examples~\ref{cexample:masking} and \ref{cexample:masking_general}). It also does not satisfy Property~\ref{propty:absence} because when $X_c=\phi$, we have $\uni{Z}{\hat{Y}| X_{c}}=\mut{Z}{\hat{Y}}$, which is only the statistically visible disparity but not the total disparity in a counterfactual sense (i.e., $\mut{Z}{\hat{Y},U_X}$).

\noindent \textbf{Observational Measure 2.} $M_{NE}=\mut{Z}{\hat{Y}\given  X_{c}}.$

\noindent \textbf{Utility:}
This measure also satisfies several desirable properties, as stated here:

\begin{restatable}{lem}{cmisatisfyproperties}[Fairness Properties of $\mathrm{I}(Z;\hat{Y} \given  X_{c})$] The measure $\mathrm{I}(Z;\hat{Y} \given  X_{c})$ satisfies Properties~\ref{propty:synergy} and \ref{propty:complete_exemption}. 
\label{lem:conMI1_satisfy_properties}
\end{restatable}

The proof is in Appendix~\ref{app:observational}. We note that, while it does not satisfy Property~\ref{propty:masking} in its entirely, it does capture some scenarios of non-exempt masked disparity. E.g., it can detect the non-exempt masked disparity in Canonical Example~\ref{cexample:masking} which $\uni{Z}{\hat{Y}| X_{c}}$ is not able to, even though they both fail to detect the non-exempt masked disparity in Canonical Example~\ref{cexample:masking_general}. In general, $\mathrm{I}(Z;\hat{Y} \given  X_{c})$ can detect non-exempt masked disparity when the ``mask'' is entirely derived from the critical features, i.e., $G=g(X_{c})$.

\noindent \textbf{Limitations:} It can sometimes lead to false positive conclusion about non-exempt disparity, e.g., in Canonical Example~\ref{cexample:college} (does not satisfy Property~\ref{propty:cancellation}). It also does not satisfy Property~\ref{propty:nonincreasing} because clearly $\mut{Z}{\hat{Y}\given  X_{c}}$ may be greater or less that $\mut{Z}{\hat{Y}}$ (recall Canonical Example~\ref{cexample:masking}). It also does not satisfy Property~\ref{propty:absence} because when $X_c=\phi$, we have $\mut{Z}{\hat{Y}\given  X_{c}}=\mut{Z}{\hat{Y}}$, which is only the statistically visible disparity but not the total disparity in a counterfactual sense (i.e., $\mut{Z}{\hat{Y},U_X}$).

\noindent \textbf{Observational Measure 3.} $M_{NE}=\mut{Z}{\hat{Y} \given  X_{c}, X'}$ where $X'$  consists of certain features in $X_{g}$.

\noindent \textbf{Utility and Limitations:}
This is somewhat of a heuristic relaxation that only satisfies Property~\ref{propty:complete_exemption}. However, while it does not satisfy any of the other properties in their entirety, it can still lead to the desirable quantification in several examples where the previous two measures may not be successful if $X'$ is chosen appropriately. For example, recall Canonical Example~\ref{cexample:masking_general} where $\hat{Y}=Z \oplus U_{X_1}$ with $X_g=(Z,U_{X_1})$. With some partial knowledge or assumption about the SCM, if we choose $X'=U_{X_1}$, then $\mut{Z}{\hat{Y} \given  X_{c}, X'}>0$ for this example even though 
$\mut{Z}{\hat{Y}\given  X_{c}}=0$. Thus, this measure is able to detect some more scenarios of non-exempt masked disparity that $\mut{Z}{\hat{Y}\given  X_{c}}$ cannot, i.e., when the mask is of the form $G=g(X_{c},X')$. It can also sometimes avoid false positive quantification of non-exempt disparity if $X'$ is chosen appropriately, e.g., in Canonical Example~\ref{cexample:college} if $X'=U_{X_1}$. Thus, under partial knowledge or assumption about the true SCM, this measure can correctly capture the non-exempt disparity in many scenarios where the previous two measures may not be successful.





Lastly, one may also consider using various combinations of these measures, e.g., $\uni{Z}{\hat{Y}| X_{c}} + \mut{Z}{\hat{Y}\given X'}$, or $\mut{Z}{\hat{Y}\given X_{c}} + \mut{Z}{\hat{Y}\given X'}$, or $\uni{Z}{\hat{Y}| X_{c}} + \syn{Z}{(\hat{Y}, X')},$ that can also approximate our proposed measure in several scenarios if $X'$ is chosen appropriately based on partial knowledge or assumptions about the true SCM.

\section{Case Studies Demonstrating Practical Application in Auditing and Training}
\label{sec:case_study}

Here, we discuss some case studies to demonstrate application of our proposed techniques on both simulated and real data.

\subsection{Case Study on Simulated Data}
\label{subsec:case_study_sim}

We present our case study on simulated data first. The benefit of using simulated data is that the true causal model (SCM) is known. The knowledge of the SCM enables the following: (i) we can exactly compute our proposed causal measure of non-exempt disparity ($M_{NE}^*$), as well as, demonstrate the decomposition of total disparity into four components during auditing a pre-trained model; (ii) we can also compare the performance of different observational measure of non-exempt disparity when used as a regularizer during training. Assuming the SCM is not available during training (but available during auditing), we examine the tradeoff between accuracy and the actual causal non-exempt disparity ($M_{NE}^*$) when each of these observational measures are used as a regularizer, under various experimental scenarios.

In this case study, an algorithm has to decide whether to show ads for a job using a score generated from internet activity. We will consider four different experimental scenarios, each with a known SCM. To demonstrate application in \textbf{auditing}, we first train a Deep-Neural-Network (DNN) model with no fairness regularizer for each of the four scenarios, and then use our techniques for computing the total disparity ($\mut{Z}{(\hat{Y},U_X)}$), as well as, decompose the total disparity into four components, namely, visible and masked, exempt and non-exempt disparities. We use the \texttt{dit}~\cite{dit} package to compute all of these quantities from the empirical distribution of the test data after the model has been trained, and after appropriately discretizing continuous random variables as required. Note that, to compute unique information, the package solves an optimization problem~\cite{dit}.

To demonstrate application in \textbf{training}, we train a DNN model $\hat{Y}=h(X)$ for classification with different \textbf{observational regularizers} and examine the tradeoff between accuracy and the \textit{actual} non-exempt disparity (as measured by our causal measure of non-exempt disparity $M_{NE}*$), when each of these observational regularizers are used. For simplicity and ease of computation during training, we rely on simple correlation-based estimates (inspired from \cite{kamishima2012fairness}) of mutual information and conditional mutual information. Further, we introduce a novel regularizer for approximating unique information, leveraging a Gaussian approximation for PID in \cite{barrett2015exploration}. We train using the following loss functions:

\begin{itemize}[leftmargin=*]

\item Loss~$L_1$ (\textbf{Statistical Parity} using Mutual Information regularizer $\mut{Z}{\hat{Y}}$ (denoted as MI)):  $$\min_{w,b} L_{\text{Cross Entropy}}(Y,\hat{Y}) + \lambda \widetilde{\mathrm{I}}(Z;\hat{Y}),$$ where (i) $\lambda$ is the regularization constant; and (ii) $\widetilde{\mathrm{I}}(Z;\hat{Y})=-\frac{1}{2}\log{{(1-\rho^2_{Z,\hat{Y}})}}$ is an approximate expression of mutual information where $\rho_{Z,\hat{Y}}$ is the correlation between $Z$ and $\hat{Y}$. This approximation is exact if $Z$ and $\hat{Y}$ are jointly Gaussian~\cite{cover2012elements}.

\item Loss $L_2$ (Proposed Unique Information-based (observational) regularizer $\uni{Z}{\hat{Y}|X_c}$ (denoted as Uniq)): $$\min_{w,b} L_{\text{Cross Entropy}}(Y,\hat{Y}) {+} \lambda \widetilde{\mathrm{Uni}}(Z:\hat{Y}|X_{c}),$$ where  $\widetilde{\mathrm{Uni}}(Z:\hat{Y}| X_{c})$ is given by: 
\begin{align}
\widetilde{\mathrm{Uni}}(Z:\hat{Y}|X_{c})&= \widetilde{\mathrm{I}}(Z;\hat{Y})-\min\{\widetilde{\mathrm{I}}(Z;\hat{Y}),\widetilde{\mathrm{I}}(Z;X_c) \} \nonumber \\
&= -\frac{1}{2}\log{{(1-\rho^2_{Z,\hat{Y}})}}  - \min\{ -\frac{1}{2}\log{{(1-\rho^2_{Z,\hat{Y}})}}, -\frac{1}{2}\log{{(1-\rho^2_{Z,X_c})}} \} .
\end{align}
We note that, in general, $\mathrm{Uni}(Z:\hat{Y}| X_{c} \geq \mut{Z}{\hat{Y}}- \min \{\mut{Z}{\hat{Y}},\mut{Z}{X_c} \} ,$ where the lower bound is tight if all of the random variables are jointly Gaussian~\cite{barrett2015exploration}. Similarly, the correlation-based approximations are also exact under Gaussian assumptions~\cite{cover2012elements}.

\item Loss $L_3$ (Proposed Conditional Mutual Information regularizer $\mut{Z}{\hat{Y}|X_c}$ (denoted as CMI)): $$\min_{w,b} L_{\text{Cross Entropy}}(Y,\hat{Y}) {+} \lambda \widetilde{\mathrm{I}}(Z;\hat{Y}\given X_{c}),$$ where again (i) $\lambda$ is the regularization constant; and (ii) $\widetilde{\mathrm{I}}(Z;\hat{Y}\given X_{c})$ is given by: 
\begin{align}
\widetilde{\mathrm{I}}(Z;\hat{Y}\given X_{c}) &= {\sum_{i=1}^n}{\Pr}(X_{c} \in \text{Bin }i)  \widetilde{\mathrm{I}}(Z;\hat{Y}\given X_{c} \in \text{Bin }i) = {-}\frac{1}{2}\sum_{i=1}^n\Pr(X_{c} \in \text{Bin }i) \log{{(1-\rho^2_{Z,\hat{Y},i}) }},
\end{align}
where the range of $X_c$ is divided into $n$ discrete bins, and $\rho_{Z,\hat{Y},i }$ is the conditional correlation of $\hat{Y}$ and $Z$ given $X_{c}$ is in the $i$-th discrete bin.

\item Loss $L_4$ (Another Proposed Heuristic regularizer $\mut{Z}{\hat{Y}|X_c,X'}$ (denoted as CMI')): $$\min_{w,b} L_{\text{Cross Entropy}}(Y,\hat{Y}) {+} \lambda \widetilde{\mathrm{I}}(Z;\hat{Y}\given X_{c},X'),$$ where again (i) $\lambda$ is the regularization constant; and (ii) $\widetilde{\mathrm{I}}(Z;\hat{Y}\given X_{c},X')$ is given by: 
\begin{align}
\widetilde{\mathrm{I}}(Z;\hat{Y}\given X_{c},X')
&= {\sum_{i=1}^n}{\Pr}(X_{c},X' \in \text{Bin }i)  \widetilde{\mathrm{I}}(Z;\hat{Y}\given X_{c},X' \in \text{Bin }i) \nonumber \\
&= {-}\frac{1}{2}\sum_{i=1}^n\Pr(X_{c},X' \in \text{Bin }i) \log{{(1-\rho^2_{Z,\hat{Y},i}) }},
\end{align}
where the range of the joint random variables $(X_c,X')$ is divided into $n$ discrete bins, and $\rho_{Z,\hat{Y},i }$ is the conditional correlation of $\hat{Y}$ and $Z$ given $(X_{c},X')$ is in the $i$-th discrete bin. Note that, here $X'$ consists of certain features in $X_g$, as discussed in Section~\ref{sec:observational} (Observational Measure 3).

\item Loss $L_5$ (\textbf{Equalized Odds} using regularizer $\mut{Z}{\hat{Y}|Y}$ (denoted as EO)): $$\min_{w,b} L_{\text{Cross Entropy}}(Y,\hat{Y}) {+} \lambda \widetilde{\mathrm{I}}(Z;\hat{Y}\given Y),$$ where again (i) $\lambda$ is the regularization constant; and (ii) $\widetilde{\mathrm{I}}(Z;\hat{Y}\given Y)$ is given by: 
\begin{align}
\widetilde{\mathrm{I}}(Z;\hat{Y}\given Y) &= {\sum_{i=1}^n}{\Pr}(Y \in \text{Bin }i)  \widetilde{\mathrm{I}}(Z;\hat{Y}\given Y \in \text{Bin }i) \nonumber \\
& = {-}\frac{1}{2}\sum_{i=1}^n\Pr(Y \in \text{Bin }i) \log{{(1-\rho^2_{Z,\hat{Y},i}) }}.
\end{align}The range of $Y$ is divided into $n$ discrete bins, and $\rho_{Z,\hat{Y},i }$ is the correlation of $\hat{Y}$ and $Z$ given $Y$ is in the $i$-th bin.  

\end{itemize}

Now, we discuss the four scenarios (SCMs) and the corresponding results.

\noindent \textbf{Experimental Scenario 1 (All four disparities present):} The decision of showing ads for a reporter's job requiring English proficiency, is based on three features $X=(X_1,X_2,X_3)$: (i) $X_1$: a score based on online writing samples  (critical feature $X_c=X_1$); (ii) $X_2$: a score based on browsing history, e.g., interest in English websites as compared to websites of other languages; and (iii) $X_3$: a preference score based on geographical proximity. $Z$ is a protected attribute denoting whether a person is a native English speaker or not, distributed as Bern(\nicefrac{1}{2}). Suppose that the true SCM is as follows:
$X_1=Z + U_{X_1}$, $X_2=Z+U_{X_2}$, and $X_3=U_{X_3},$ where $U_{X_1},U_{X_2},U_{X_3}\sim \iid{}$ $\mathcal{N}(0,\sigma^2)$ denote latent writing ability, interests, and geographical proximity, respectively. The true labels, based on previous candidates, are given by $ Y = \mathbbm{1}(X_1 + X_2 + X_3 \geq 1) $. Here, the critical feature $X_{c}=X_1$ and the general features are $X_{g}= (X_2,X_3)$. The results are provided in Fig.~\ref{fig:audit1} and Fig.~\ref{fig:training1}. \\



\noindent \textbf{Experimental Scenario 2 (Masking by critical feature):} The decision of showing ads for an editor's job in a newspaper company is based on four features: (i) $X_1$: a relevant score based on online writing samples (critical feature $X_c=X_1$); (ii) $X_2$: a score based on browsing history, e.g., awareness of current events; (iii) $X_3$: a score based on proofreading and reviewing experience; and (iv) $X_4$: a preference score based on activity in social media, e.g., political and ideological alignment with the newspaper company. Let the protected attribute $Z$ be political inclination, distributed as Bern(\nicefrac{1}{2}). Suppose the true SCM is as follows: $X_1= U_{X_1} + U_{X_3}$, $X_2=U_{X_2}$,  $X_3=U_{X_3}$, and $X_4=U_{X_2}-Z$, where $U_{X_1}\sim$ Bern(\nicefrac{1}{2}) denotes if the writing ability is above a threshold, and $U_{X_2},U_{X_3}\sim \iid{}$ $\mathcal{N}(0,\sigma^2)$ denote interests and proofreading skill-level, respectively. Suppose that the historic true labels are given by $Y = \mathbbm{1}((X_1+X_4)^2 \geq 0.5)$, i.e., primarily high online-writing scores and high social-media-based-preference scores, but to appear ``facially neutral'' with respect to political inclination, the ad is also shown to candidates with low social-media-based-preference scores and low writing scores for whom the ad may be irrelevant. Here, the critical feature $X_{c}=X_1$ and the general features are $X_{g}= (X_2,X_3,X_4)$. The results are provided in Fig.~\ref{fig:audit2} and Fig.~\ref{fig:training2}.\\

\begin{figure*}
\begin{subfigure}[b]{0.5\linewidth}
\centering
\includegraphics[height=3.4cm]{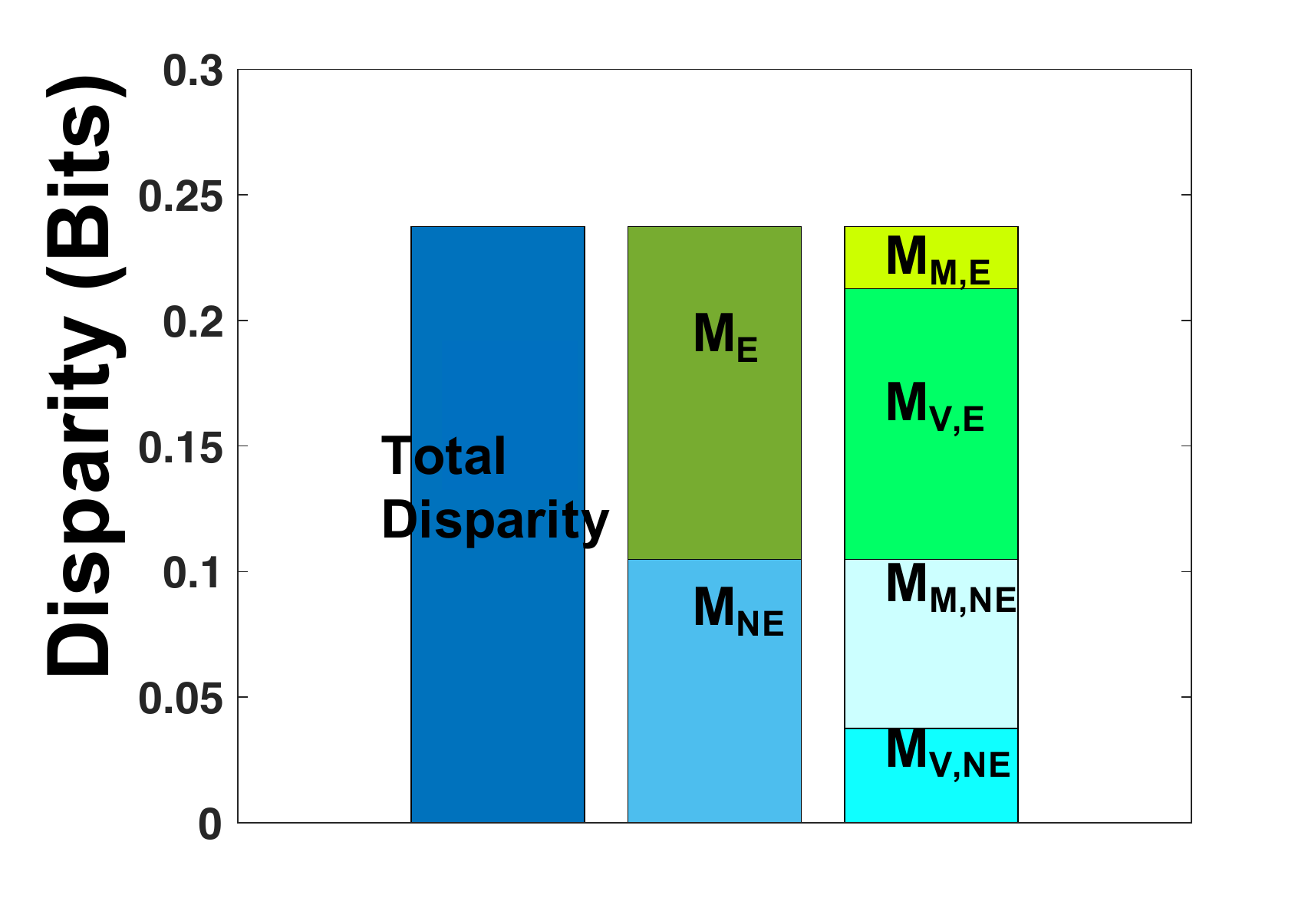}
\subcaption{Experimental Scenario 1 (All four disparities present)\label{fig:audit1} }
\end{subfigure}
\begin{subfigure}[b]{0.5\linewidth}
\centering
\includegraphics[height=3.4cm]{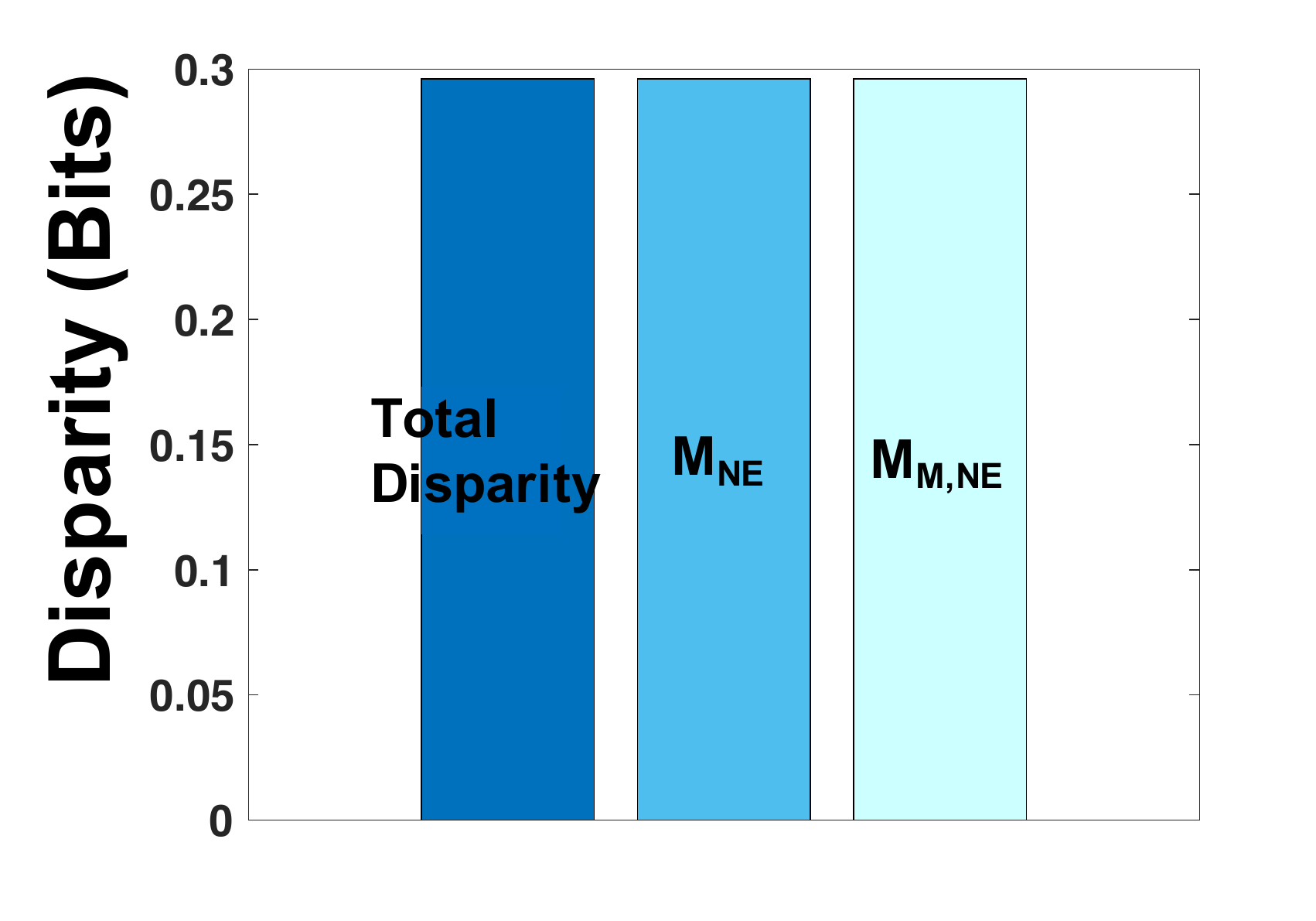}
\caption{Experimental Scenario 2 (Masking by critical feature)\label{fig:audit2} }
\end{subfigure}
\begin{subfigure}[b]{0.5\linewidth}
\centering
\includegraphics[height=3.4cm]{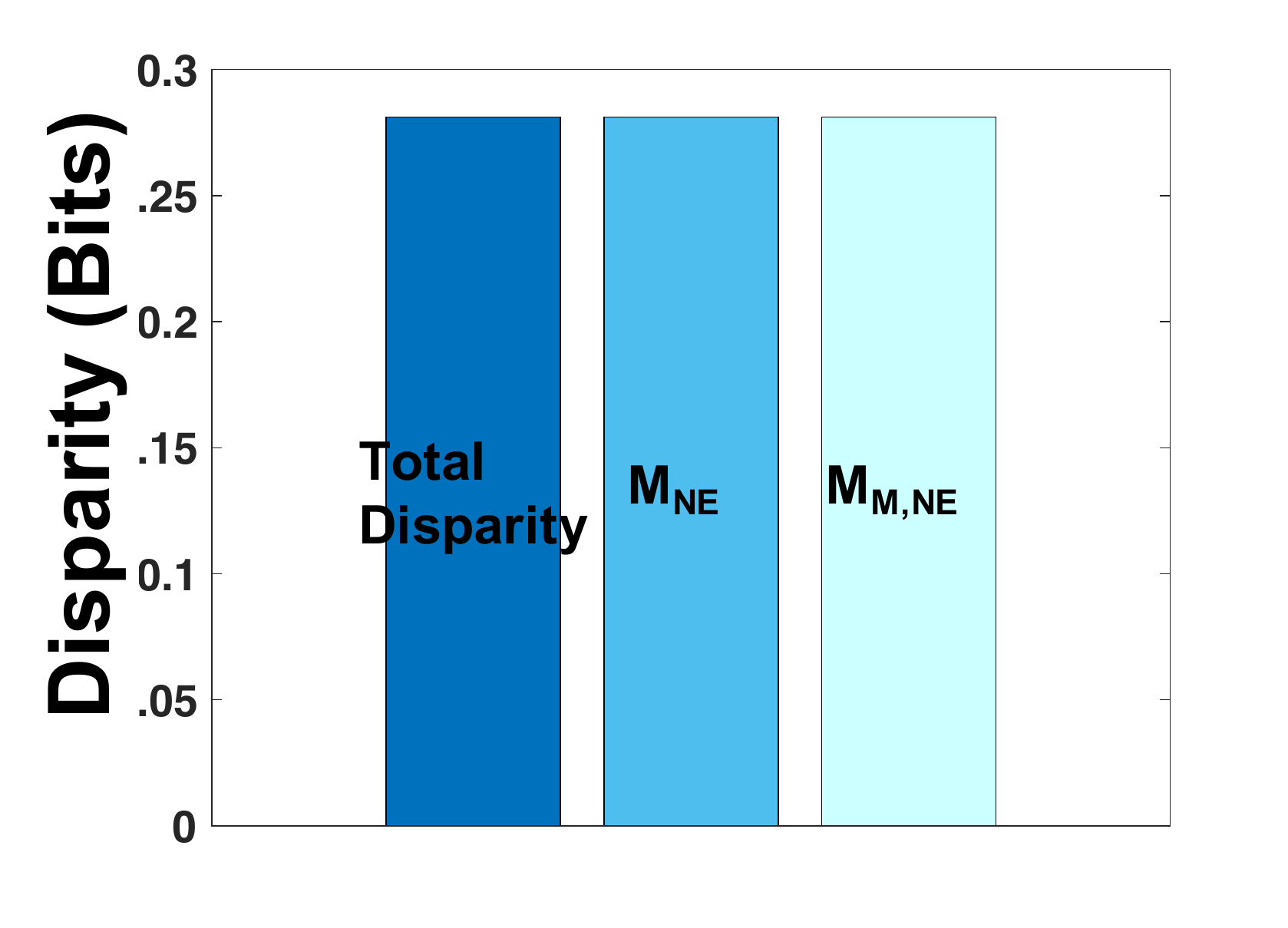}
\caption{Experimental Scenario 3 (Masking by general feature)\label{fig:audit3} }
\end{subfigure}
\begin{subfigure}[b]{0.5\linewidth}
\centering
\includegraphics[height=3.4cm]{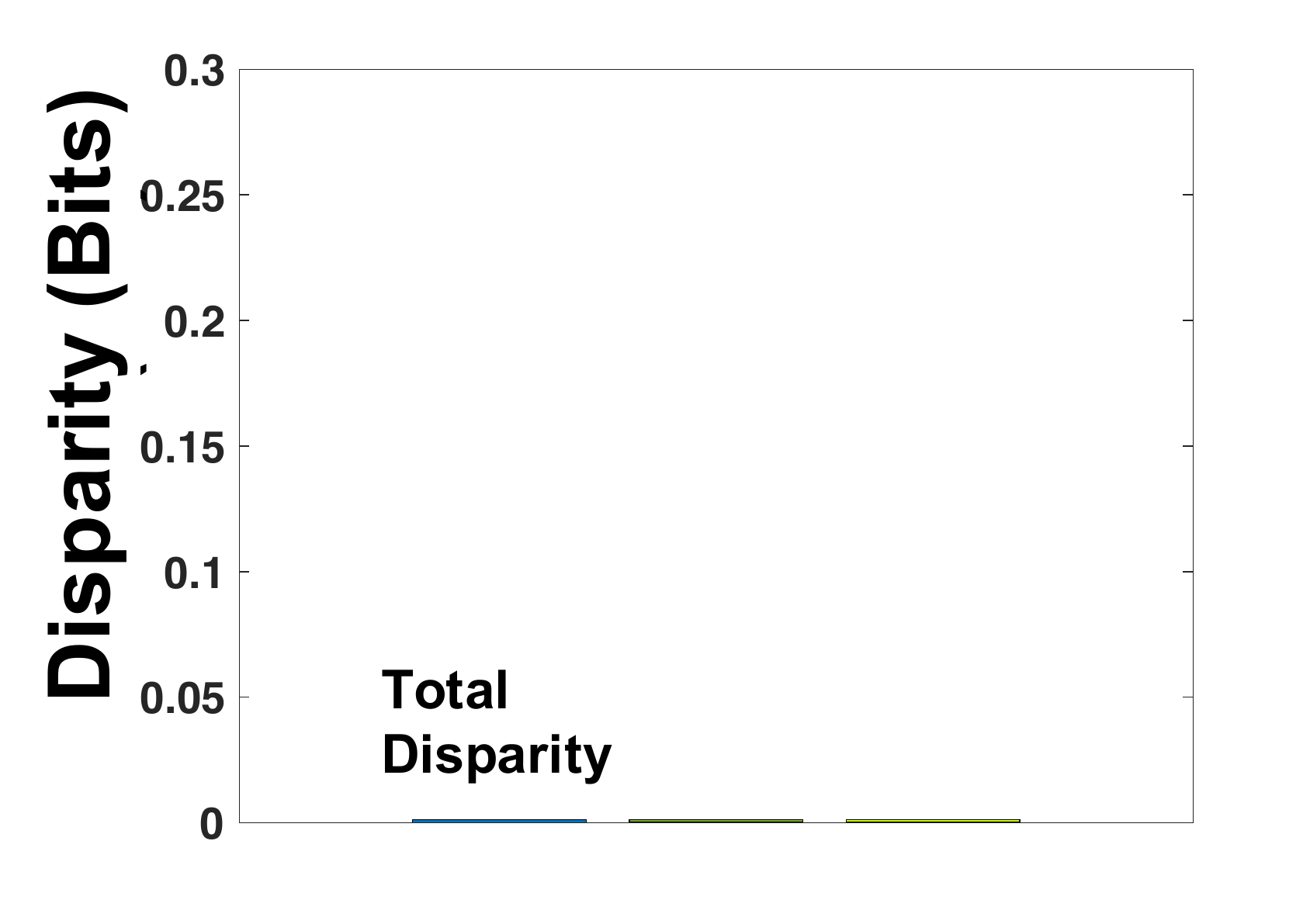}
\caption{Experimental Scenario 4 (No label bias)\label{fig:audit4} }
\end{subfigure}
\caption{\textbf{Observations from Auditing.} The different types of disparities after training a model with no fairness regularizer for all the experimental scenarios: $M_E$ and $M_{NE}(=M_{NE}*)$ denote the exempt and non-exempt disparities, respectively. $M_{V,E}$, $M_{M,E}$, $M_{V,NE}$, and $M_{M,NE}$ denote the visible and masked exempt disparity and visible and masked non-exempt disparity, respectively. Because the SCM is known, all of these quantities can be computed. For each of the four experimental scenarios, the test accuracy is close to $99\%$ (model output is very similar to the true label). We observe that the disparity decomposition for the model output $\hat{Y}$ is also quite similar to what one might intuitively expect for the true label $Y$. In \underline{Experimental Scenario 1}, biased critical and general features are used in the true label. We also observe all four disparities $M_{V,E}$, $M_{M,E}$, $M_{V,NE}$, and $M_{M,NE}$ are present in output $\hat{Y}$. In \underline{Experimental Scenarios 2 and 3}, the disparity in $\hat{Y}$ is dominated by non-exempt, masked disparity $M_{M,NE}$, and the other components are negligible. In \underline{Experimental Scenario 4}, the total disparity is significantly less in comparison to the other three scenarios (intuitively agrees with the fact that the true labels that have no bias at all). \label{fig:audit} }
\end{figure*}
\noindent \textbf{Experimental Scenario 3 (Masking by general feature):} Consider another example similar to the previous one. The decision of showing ads for a website-manager's job in a newspaper company is based on three features, none of them critical: (i) $X_1$: a  score based on online writing samples; (ii) $X_2$: a score based on browsing history, e.g., awareness of current events; and (iii) $X_3$: a preference score based on activity in social media, e.g., political alignment with the newspaper. The protected attribute $Z$ is political inclination, distributed Bern(\nicefrac{1}{2}). Suppose the true SCM is as follows:
$X_1= U_{X_1}+U'_{X_1}$, $X_2=U_{X_2}$, and $X_3=U_{X_2}-Z$, where $U_{X_1}\sim$ Bern(\nicefrac{1}{2}) denotes if writing ability is above a threshold, and $U'_{X_1}, U_{X_2}\sim \iid{}\; \mathcal{N}(0,\sigma^2)$ denote proofreading skill and interests. Suppose that the true labels are given by $Y = \mathbbm{1}( (X_1+X_3)^2 \geq 0.5) $, i.e., primarily high online-writing scores and high social-media-based-preference scores, but to appear ``facially neutral'' with respect to political inclination, the ad is also shown to candidates with low social-media-based-preference scores and low writing scores. Here, all the features are non-critical: $X_{g}= (X_1,X_2,X_3)$. The results are provided in Fig.~\ref{fig:audit3} and Fig.~\ref{fig:training3}.\\

\noindent \textbf{Experimental Scenario 4 (No label bias):} The decision of showing ads for an editor's job is based on four features: (i) $X_1$: a score based on online writing samples (critical feature $X_c=X_1$); (ii) $X_2$: a score based on browsing history, e.g., awareness of current events; (iii) $X_3$: a preference score based on geographical proximity; and (iv) $X_4$: a score based on browsing history, e.g., interest in English websites as compared to websites of other languages. Let $Z\sim$ Bern(\nicefrac{1}{2}) be the protected attribute denoting whether the candidate is a native English speaker. Suppose the true SCM is as follows:
$X_1= Z+U_{X_1}+U'_{X_1}$, $X_2=U_{X_2}$, $X_3=U_{X_3}$, and $X_4=Z+U_{X_2}$, where $U_{X_1} \sim$ Bern(\nicefrac{1}{2}) denotes whether writing skill is above a threshold, and  $U'_{X_1},U_{X_2},U_{X_3}\sim \iid{}$ $\mathcal{N}(0,\sigma^2)$ denote proofreading skill, interests, and proximity. Suppose that the true labels do not have label disparityand are given by $Y = \mathbbm{1}(U_{X_1}+U_{X_2} \geq 0.5) $. Here, the critical feature is $X_c=X_1$ and the general features are $X_{g}= (X_2,X_3,X_4)$. The results are provided in Fig.~\ref{fig:audit4} and Fig.~\ref{fig:training4}.

\begin{figure*}
\begin{subfigure}[b]{0.5\linewidth}
\centering
\includegraphics[height=4.5cm]{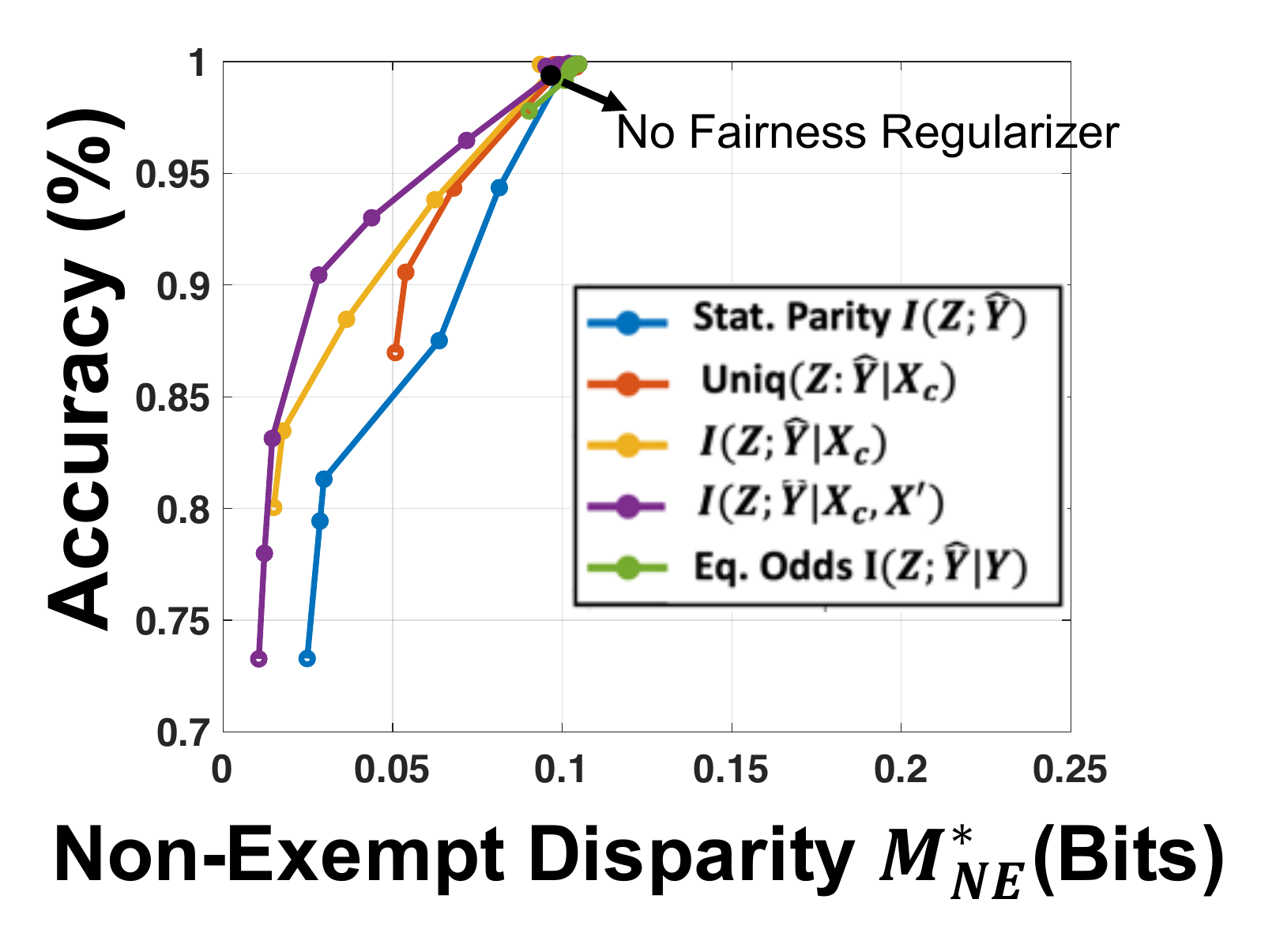}
\caption{Experimental Scenario 1 (All four disparities present)\label{fig:training1} }
\end{subfigure}
\begin{subfigure}[b]{0.5\linewidth}
\centering
\includegraphics[height=4.5cm]{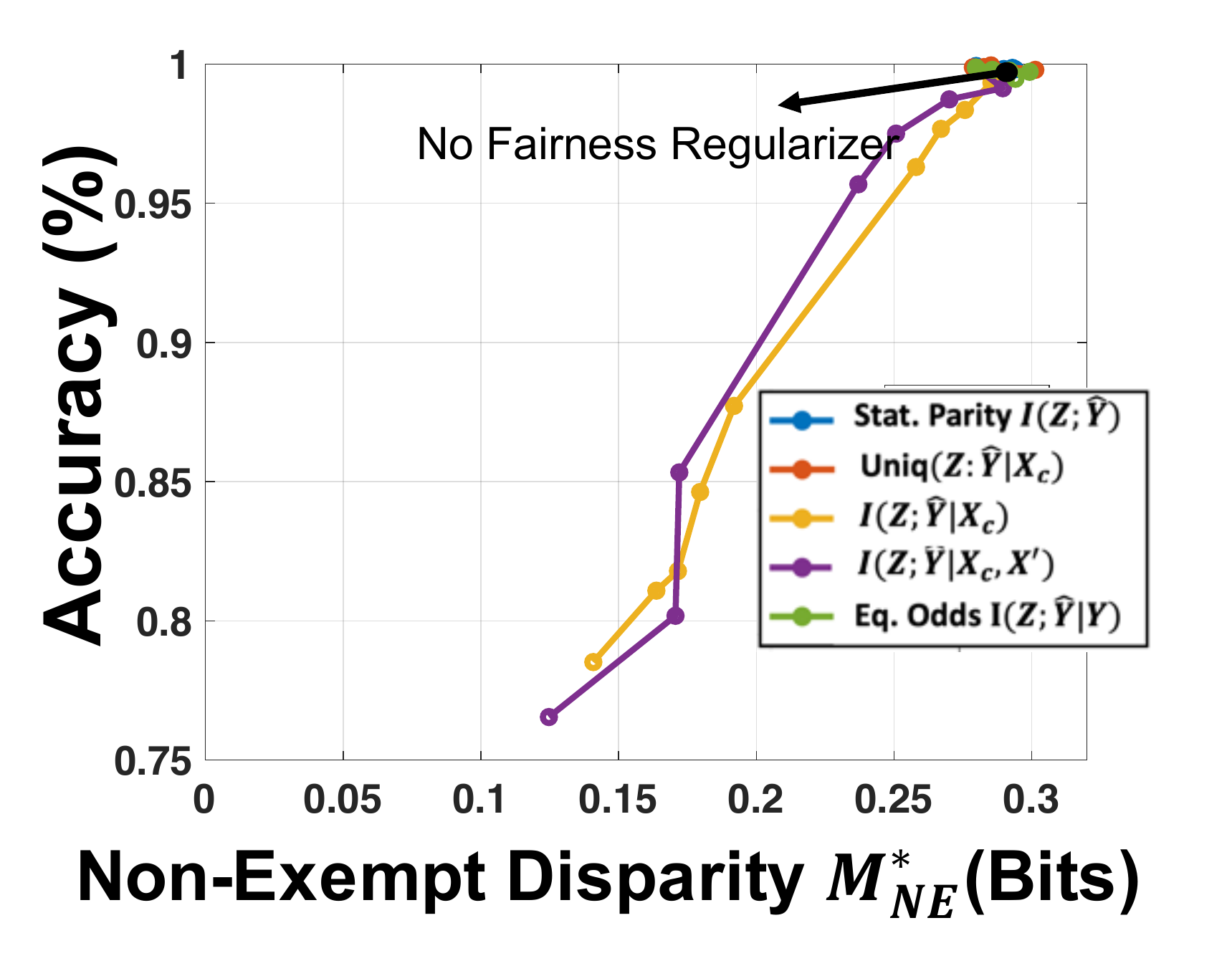}
\caption{Experimental Scenario 2 (Masking by critical feature) \label{fig:training2} }
\end{subfigure}\\
\begin{subfigure}[b]{0.5\linewidth}
\centering
\includegraphics[height=4.5cm]{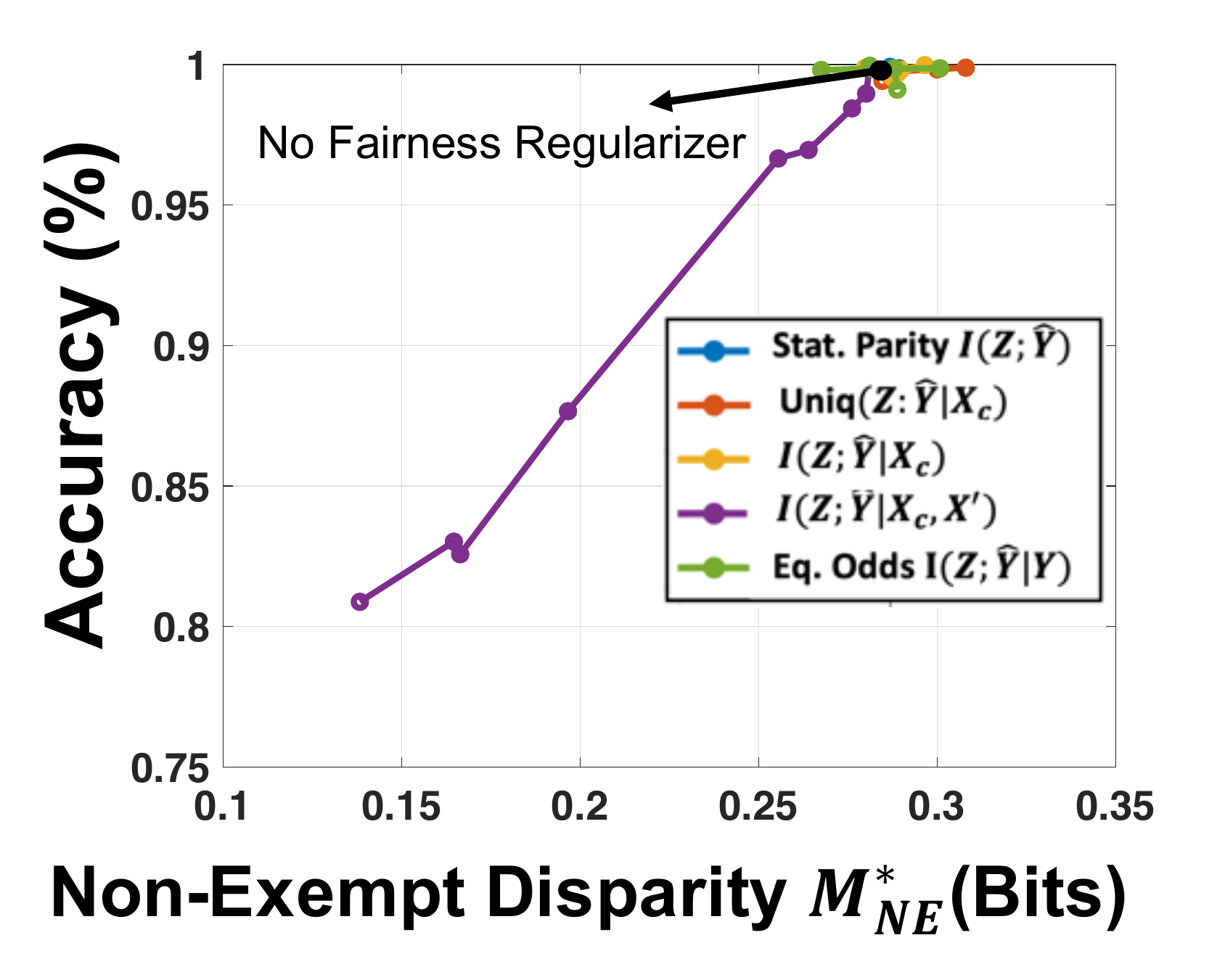}
\caption{Experimental Scenario 3 (Masking by general feature) \label{fig:training3} }
\end{subfigure}
\begin{subfigure}[b]{0.5\linewidth}
\centering
\includegraphics[height=4.5cm]{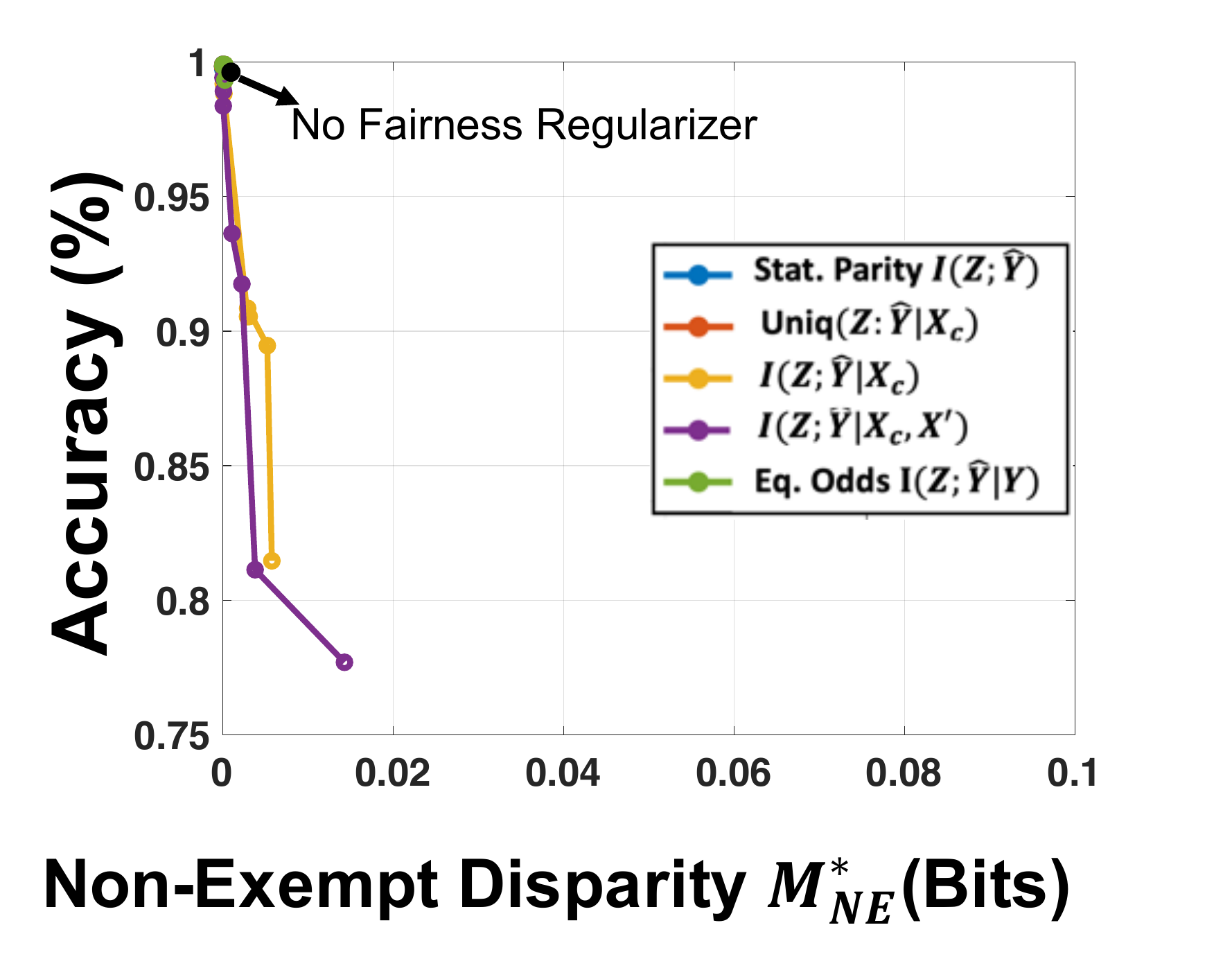}
\caption{Experimental Scenario 4 (No label bias) \label{fig:training4} }
\end{subfigure}
\caption{\textbf{Observations from training:} For each experimental scenario, we train a model using each of the five observational regularizers: MI (Statistical Parity), Uniq, CMI, CMI', and EO (Equalized Odds) for different values of regularization constant $\lambda$. The tradeoff between test accuracy and the actual non-exempt disparity ($M_{NE}*$) computed using the \texttt{dit} package is shown. In \underline{Experimental Scenario 1}, the model output (no fairness) has all four types of disparities, $M_{V,E}$, $M_{M,E}$, $M_{V,NE}$, and $M_{M,NE}$. We observe that, all three of Uniq, CMI, and CMI' attain better tradeoff between accuracy and non-exempt disparity as compared to EO (Equalized Odds) and MI (Statistical Parity). Equalized Odds does not affect the accuracy or the non-exempt disparity much, even for high values of the regularization constant. Statistical Parity attempts to reduce both exempt and non-exempt disparities, and ends up reducing accuracy a lot for same values of non-exempt disparity as compared to Uniq, CMI, and CMI'. CMI and CMI' are slightly better than Uniq because they also partially quantify non-exempt, masked disparity. For CMI'($=\mut{Z}{\hat{Y}|X_c,X'}$), we choose $X'=X_3$ (location, a general feature that has no causal influence of $Z$, but is suspected to ``mask'' $Z$ in the final output) which leads to a better trade-off than CMI. In \underline{Experimental Scenarios 2 and 3}, the disparity in the model output (no fairness) is dominated by non-exempt, masked disparity. This disparity is missed by MI (Statistical Parity), EO (Equalized Odds), and Uniq. Consequently, they do not affect the accuracy or the non-exempt disparity much, even for high values of regularization constant. For Experimental Scenario 2, only CMI and CMI' (with $X'= X_3$) are able to detect the non-exempt, masked disparity, and lead to alternate models with reduced accuracy and also reduced non-exempt disparity. For Experimental Scenario 3, only CMI' (with $X'=X_1$, the general feature that masks $Z$ in the final output) detects the non-exempt disparity, and reduces it. In \underline{Experimental Scenario 4}, the model output (no fairness) has almost negligible non-exempt disparity because the true labels do not have any bias at all. We observe that, MI, EO, and Uniq also do not affect the accuracy much even for high values of regularization constant (which is desirable). However, CMI, and CMI' (with $X'=X_2$) falsely detect disparity here, when there is no non-exempt disparity actually present. In an attempt to reduce the falsely detected disparity, they lead to alternate models with significantly reduced accuracy, and slightly increased non-exempt disparity.\label{fig:training}}
\end{figure*}

\textbf{Summary of Results:} We present results for auditing and training in Fig.~\ref{fig:audit} and Fig.~\ref{fig:training} with detailed explanations. Our proposed regularizers, namely, Uniq, CMI and CMI' attain better trade-off between accuracy and non-exempt disparity than MI (Statistical Parity) and EO (Equalized Odds) in Experimental Scenario 1. CMI and CMI' are also able to detect certain scenarios of non-exempt, masked disparity that Uniq, MI and EO fail to detect, e.g., in Experimental Scenario 2 where the masking is by the critical feature $X_c$. Experimental Scenario 3 demonstrates additional scenarios of non-exempt, masked disparity, e.g., masking by $X_g$, where even CMI is unable to detect this disparity, and only CMI' succeeds (by choosing $X'$ based on certain knowledge/suspicion of the causal model). However, Experimental Scenario 4 denotes a scenario of false detection of disparity by CMI and CMI'. In essence, Uniq is a somewhat conservative measure of non-exempt disparity which can miss non-exempt, masked disparity, but never does false detection of disparity. On the other hand, CMI and CMI' can sometimes detect certain scenarios of non-exempt, masked disparity, but can also sometimes falsely detect disparity. This is expected: these are observational measures attempting to approximate a causal measure, a fundamentally impossible task. However, these examples illustrate how knowledge of aspects of the SCM (e.g., whether the disparity is predominantly masked disparity) can be used to inform the choice of the observational measure.

\subsection{Case Study on Real Data: \texttt{Adult} Dataset}

The Adult dataset~\cite{UCI}, also known as the Census income dataset, consists of $14$ features (e.g. age, educational qualification), and the true labels denote whether the income is greater than $\$50$k. This dataset is widely used in existing fairness literature (e.g., \cite{zemel2013learning}), because it is representative of data used in highly consequential applications, such as, lending, showing expensive ads, etc. Here, we choose gender as the protected attribute $(Z)$ for analyzing the Adult dataset. Our set of input features $(X)$ consists of all the other features except gender, and our critical feature $(X_c)$ is working-hours per-week.

We train a deep neural network (multi-layer perceptron) on this dataset, with all features, except gender, as input (with one hot encoding of all categorical variables). The input layer is followed by three hidden layers, each having 32 neurons with ReLu activation and dropout probability $0.2$. Finally, the output layer consists of a single neuron with sigmoidal activation that produces an output value between $0$ and $1$ (likelihood of income being $>50$k possibly leading to a loan decision). 

Since the true causal model is not known, we cannot compute the exact value of the total disparity or non-exempt disparity $(M_{NE}^*)$ as in the previous case study. However, our observational measures can still provide valuable insights as we demonstrate here (see Fig.~\ref{fig:adult1}). We consider five setups for auditing: (i) No fairness: model trained with no fairness regularizer; (ii) Statistical Parity: model trained with $\mut{Z}{\hat{Y}}$ as regularizer; (iii) CMI Regularizer: model trained with $\mut{Z}{\hat{Y}|X_c}$ regularizer; (iv) Uniq Regularizer: model trained with $\uni{Z}{\hat{Y}|X_c}$ as regularizer; and (iv) Equalized Odds: model trained with $\mut{Z}{\hat{Y}|Y}$ regularizer. For each of these setups, we choose the same value of the regularization constant $\lambda=4$, and similar correlation-based estimates for the regularizers as in the previous case study.

After training these models, we audit/evaluate the trained models by computing the following observational quantities on the empirical distribution of the test data using the \texttt{dit}~\cite{dit} package: MI (statistically visible disparity:$\mut{Z}{\hat{Y}}$), CMI (conditional mutual information $\mut{Z}{\hat{Y}|X_c}$), as well as, the decomposition of CMI into Unique Information (Uniq) given by $\uni{Z}{\hat{Y}|X_c}$) and Synergistic Information (Syn) given by $\syn{Z}{(\hat{Y},X_c)}$). Recall that Uniq is the non-exempt statistically visible disparity, while Syn can correspond to either non-exempt masked disparity or false detection of disparity (recall our impossibility result; one might need some knowledge of the causal model to be certain). As discussed in the caption of Fig.~\ref{fig:adult1}, the correlation-based estimates serve as relatively good approximations and reduce the respective statistical dependences as one would intuitively expect to see. 
\begin{figure*}
\centering
	\includegraphics[height=3.2cm]{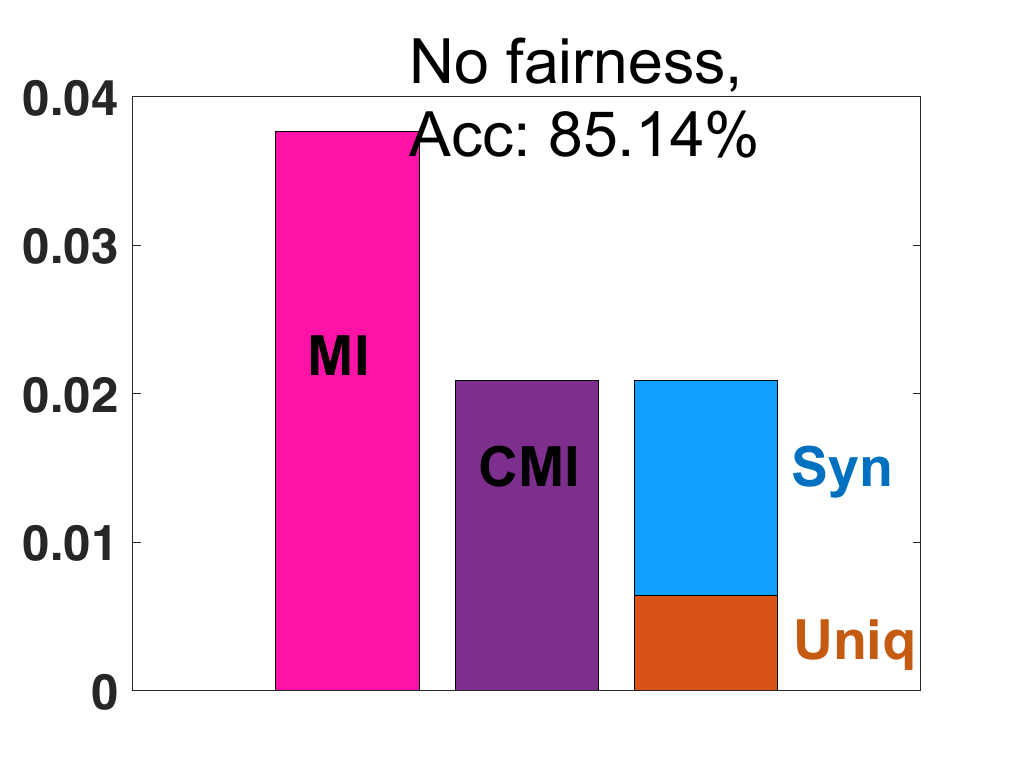}\\
	\includegraphics[height=3.2cm]{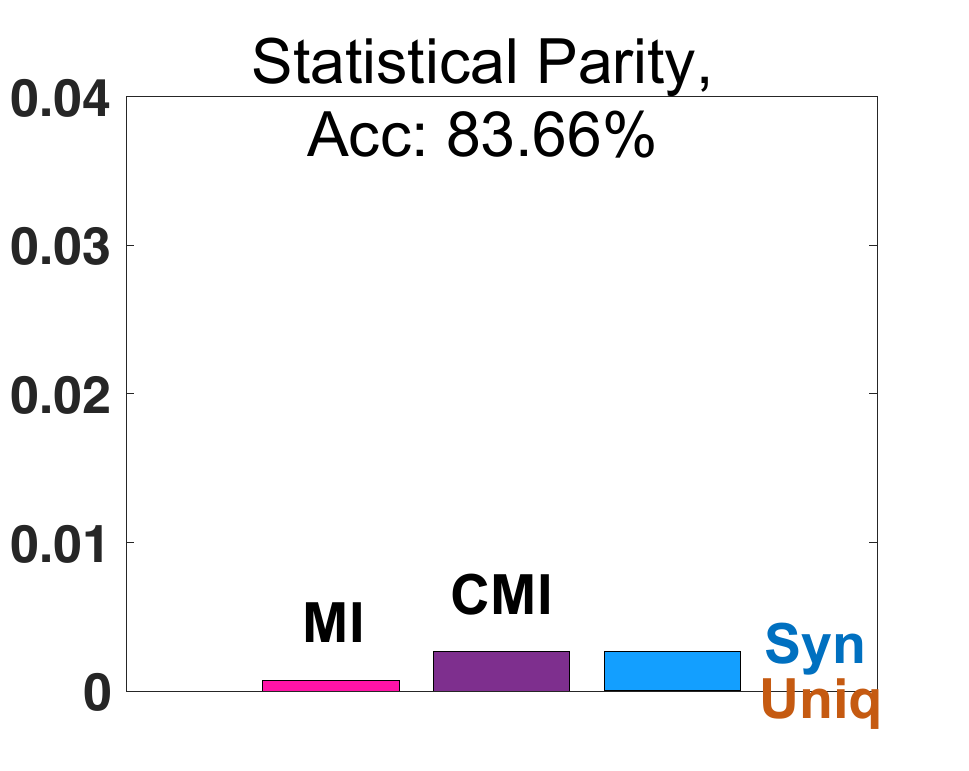}
	\includegraphics[height=3.2cm]{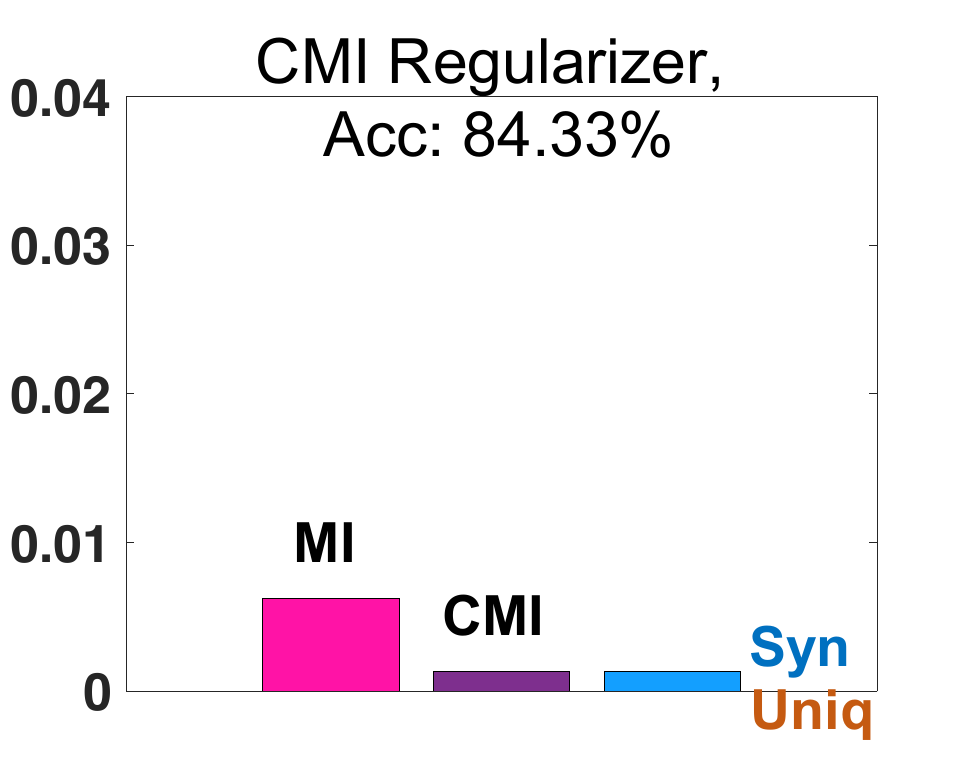}
	\includegraphics[height=3.2cm]{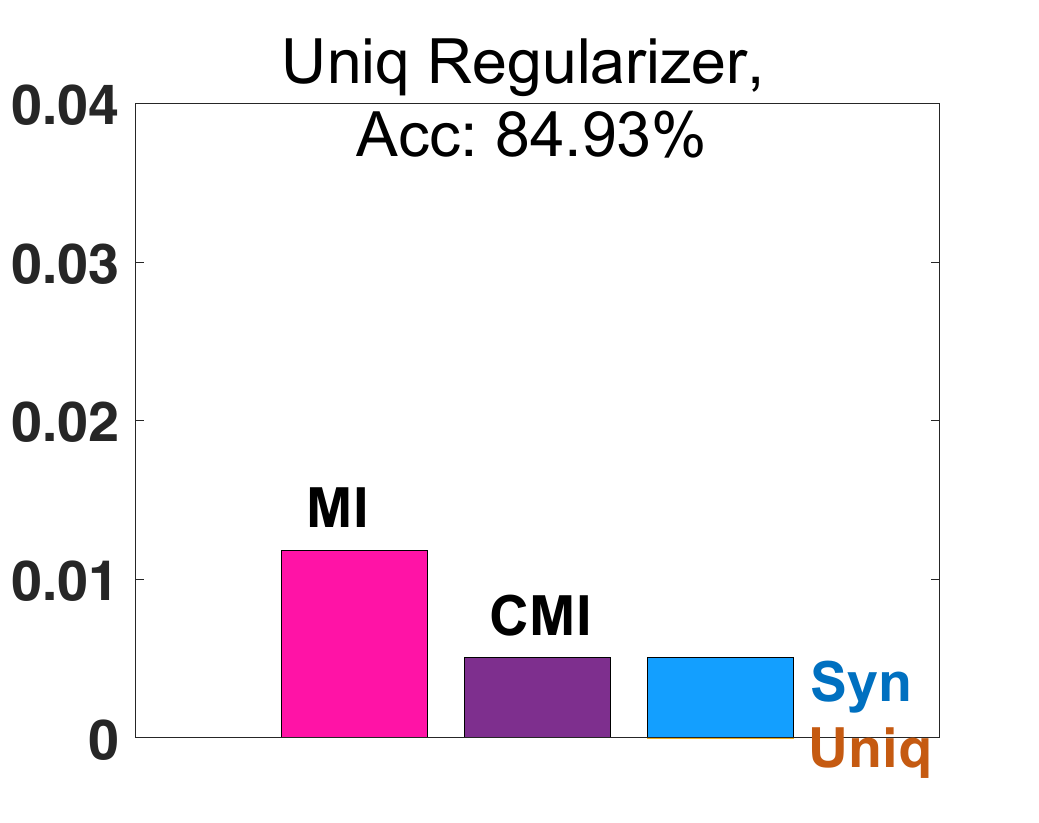}
	\includegraphics[height=3.2cm]{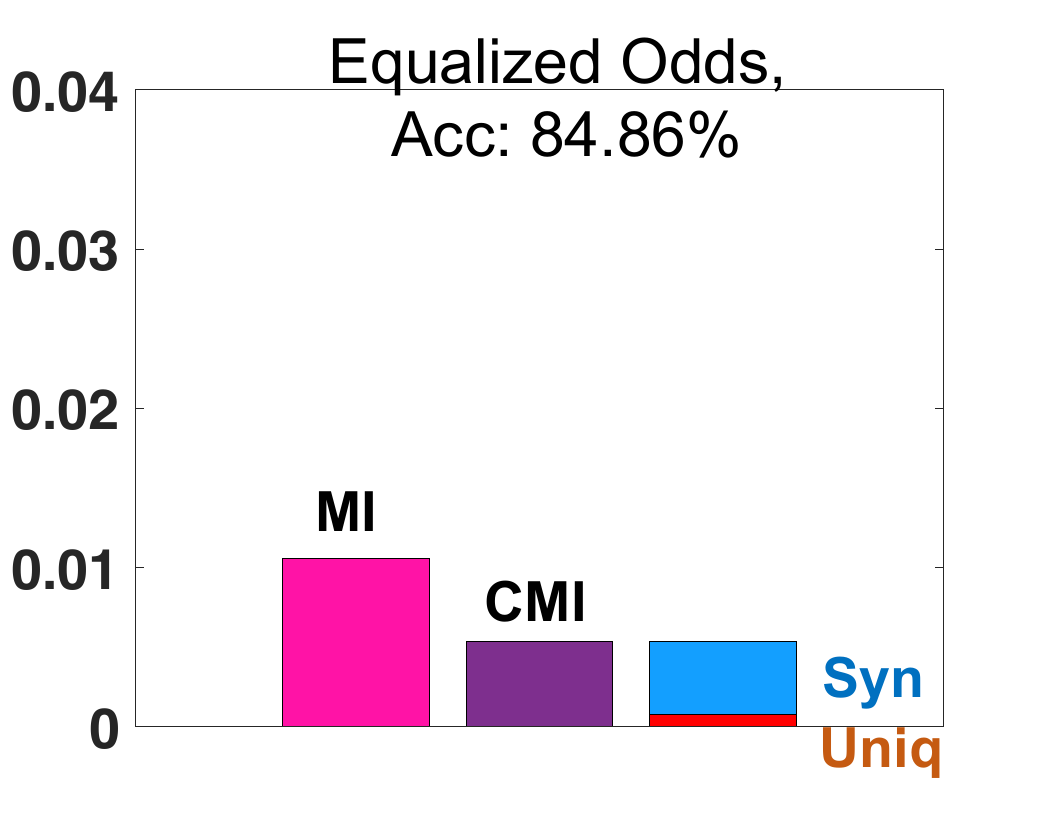}\\
	\includegraphics[height=3cm]{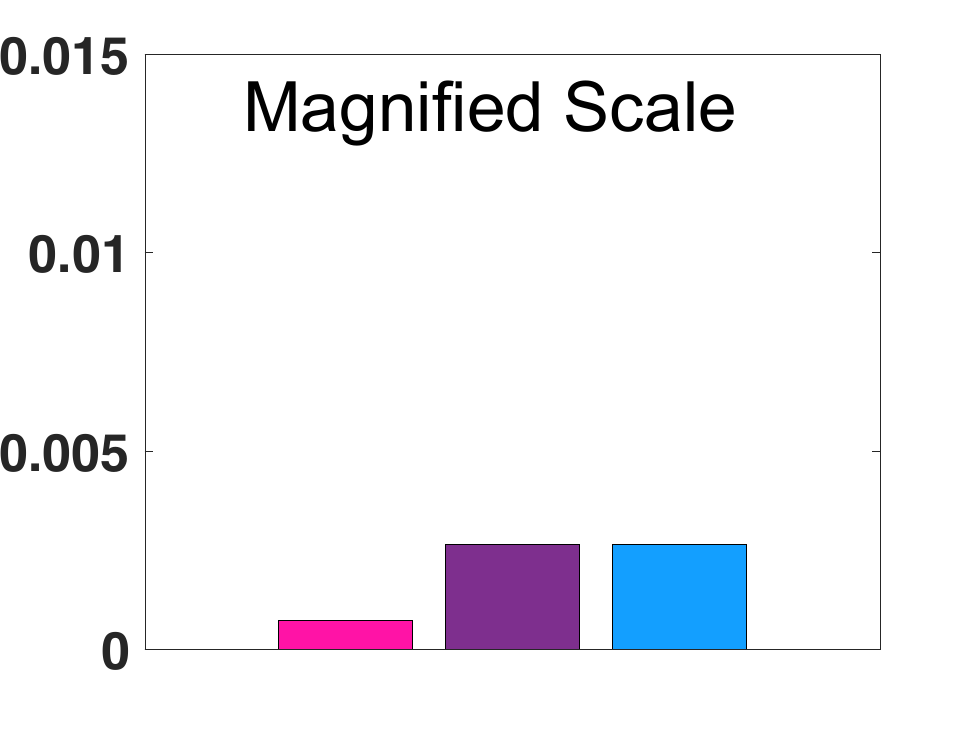}
	\includegraphics[height=3cm]{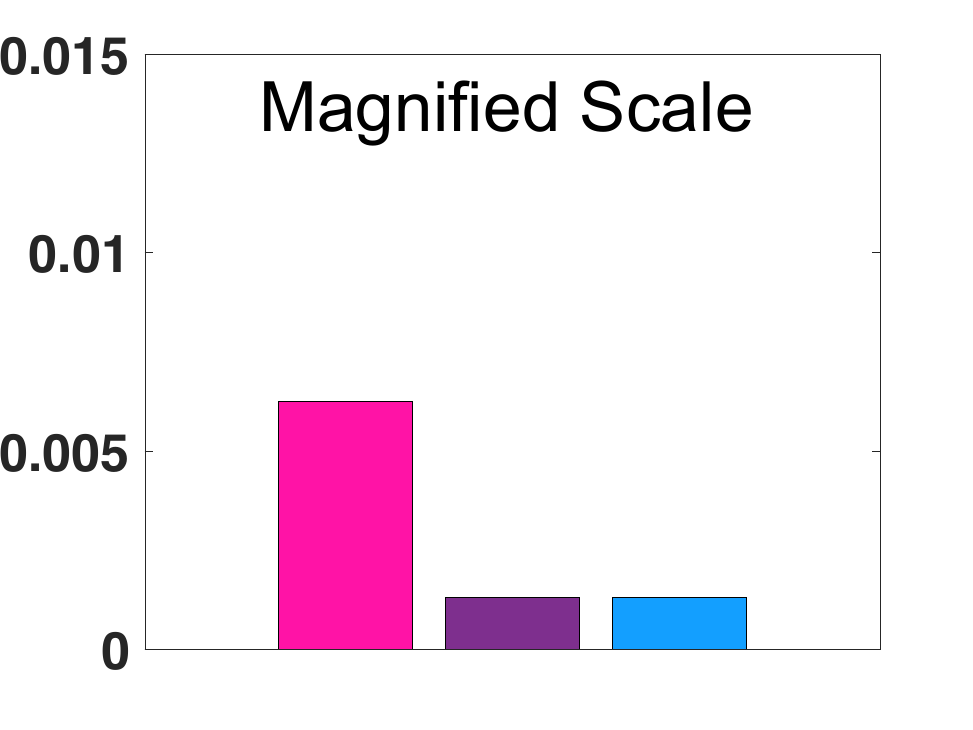}
	\includegraphics[height=3cm]{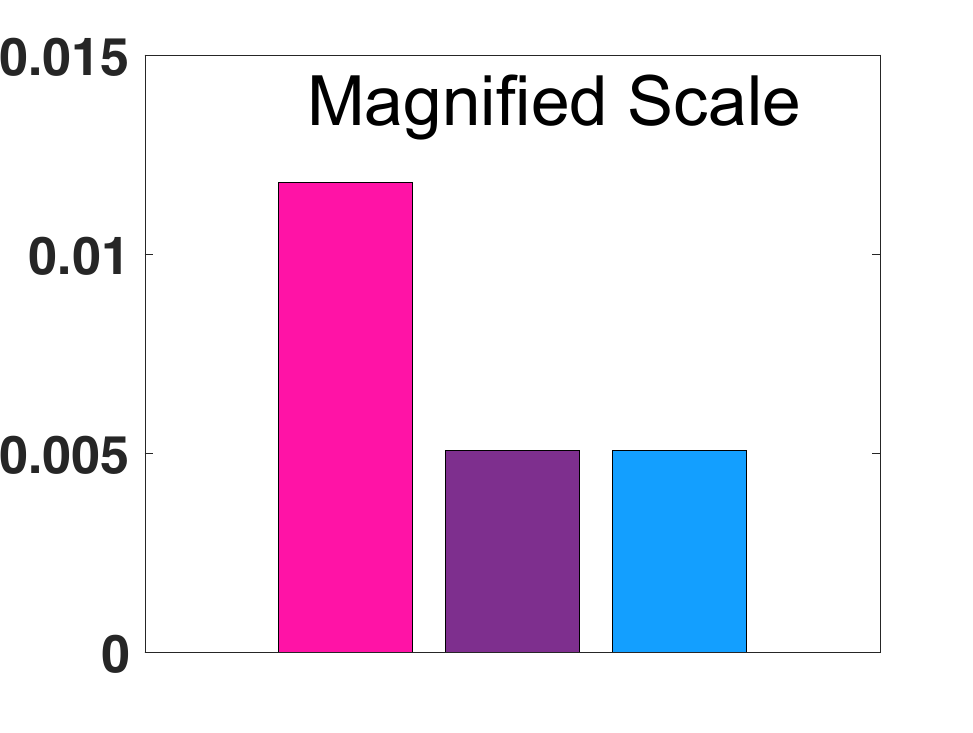}
	\includegraphics[height=3cm]{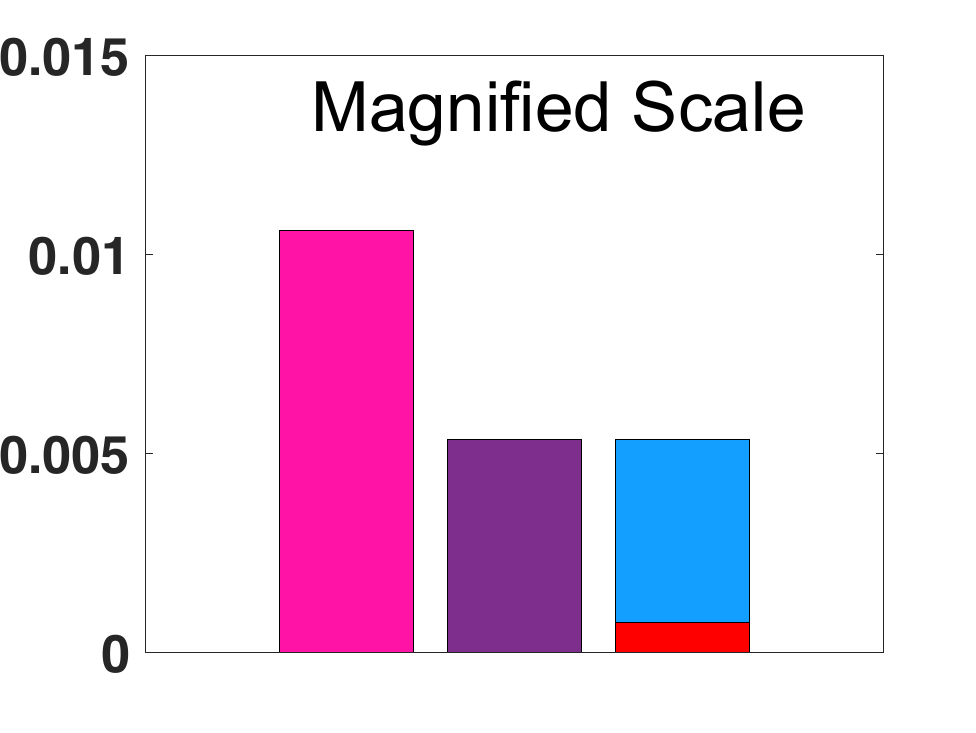}
	\caption{For the model with no fairness, we see a high value of MI as well as CMI (consisting of both Uniq and Syn). When the model is trained for statistical parity, the MI reduces as expected, but interestingly CMI is now higher than MI. Next, when CMI is used as a regularizer, we notice that CMI (and its sub-components Uniq and Syn) reduce as expected, but MI is higher than CMI. For Uniq as a regularizer, we notice that MI or CMI are not reduced that much, but only Uniq is minimized selectively. Lastly, for equalized odds, we observe that the trained model still has some Uniq (non-exempt, visible disparity). These experiments also demonstrate that the correlation-based estimates for the regularizers are relatively good approximations for this real dataset and actually reduce the respective statistical dependences as one would intuitively expect.  \label{fig:adult1}}
\end{figure*}

\subsection{Case Study on Real Data: \texttt{German Credit} Dataset}

We also perform a similar case study on the German Credit Dataset~\cite{UCI}. This dataset consists of $20$ features (e.g., status of a checking account, credit amount, present employment, etc.), and the true labels denote whether a customer is good or bad. Our critical feature ($X_c$) is the number of existing credits at this bank, and the protected attribute ($Z$) is gender. Our set of all features $(X)$ consist of all features except gender and marital status.

We train a deep neural network (multi-layer perceptron) on this dataset, with all features, except gender and marital status as input (with one hot encoding of categorical variables). The input layer is followed by two hidden layers, each having $124$ neurons with ReLu activation and dropout probability $0.5$. Finally, the output layer consists of a single neuron with sigmoidal activation that produces an output value between 0 and 1 (likelihood of being a good customer).

The causal model is again not known, similar to the previous case. However, similar to the case study on the Adult dataset, we train the model using different observational regularizers, and audit/evaluate the trained models. As discussed in the caption of Fig.~\ref{fig:german}, the correlation-based estimates reduce the respective statistical dependences as one would intuitively expect to see.

\begin{figure*}
\centering
	\includegraphics[height=3.5cm]{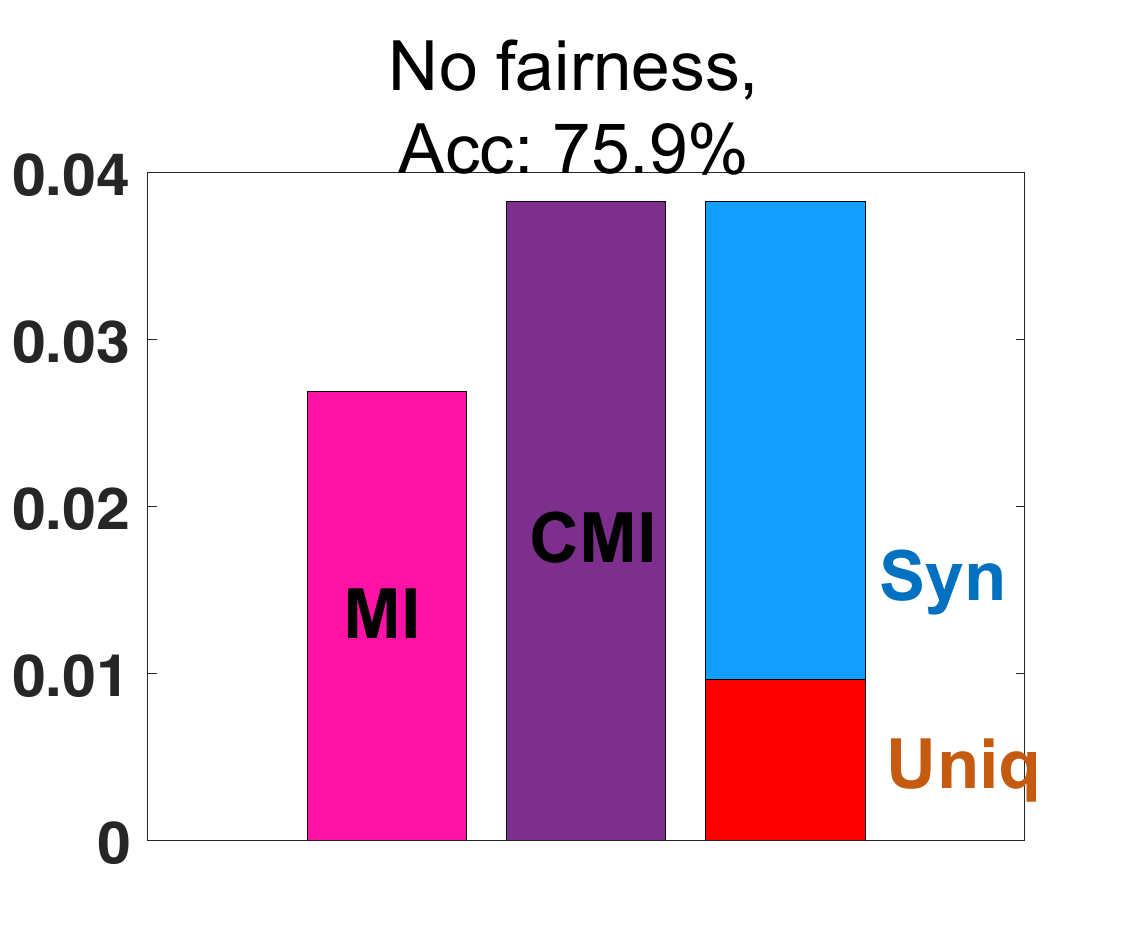}\\
	\includegraphics[height=3.5cm]{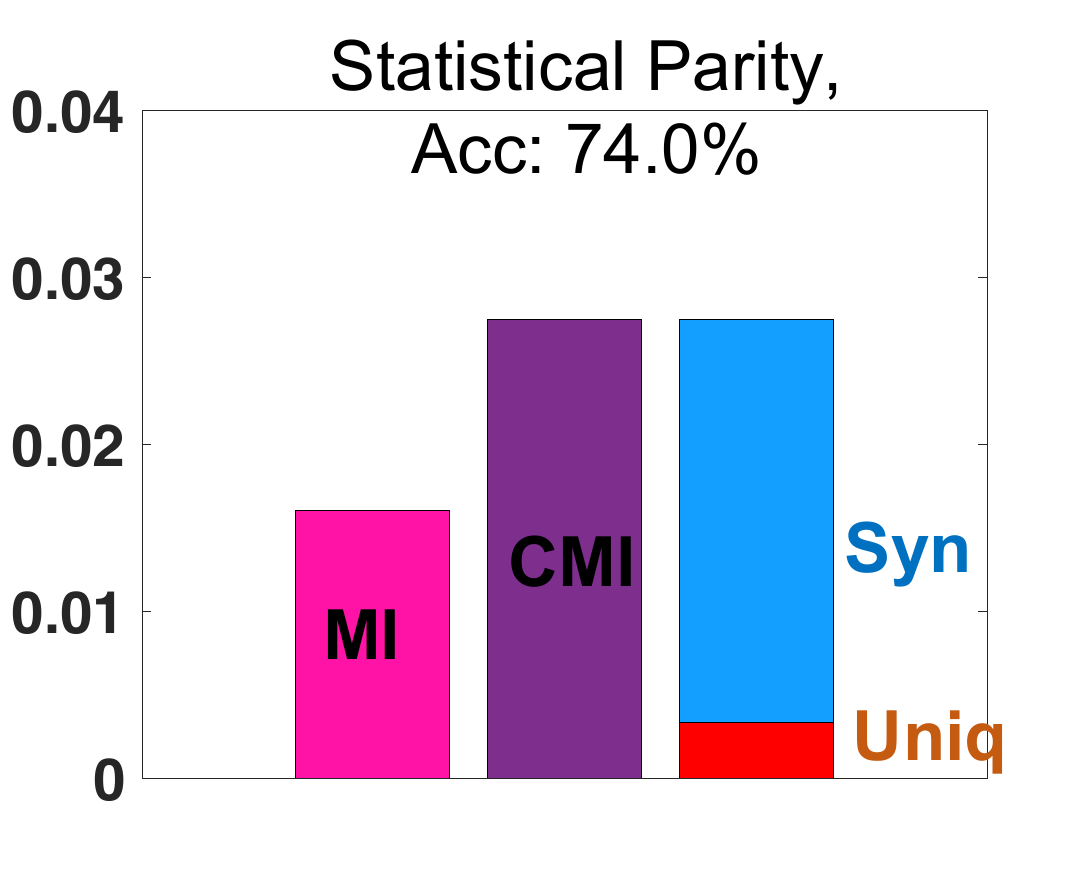}
	\includegraphics[height=3.5cm]{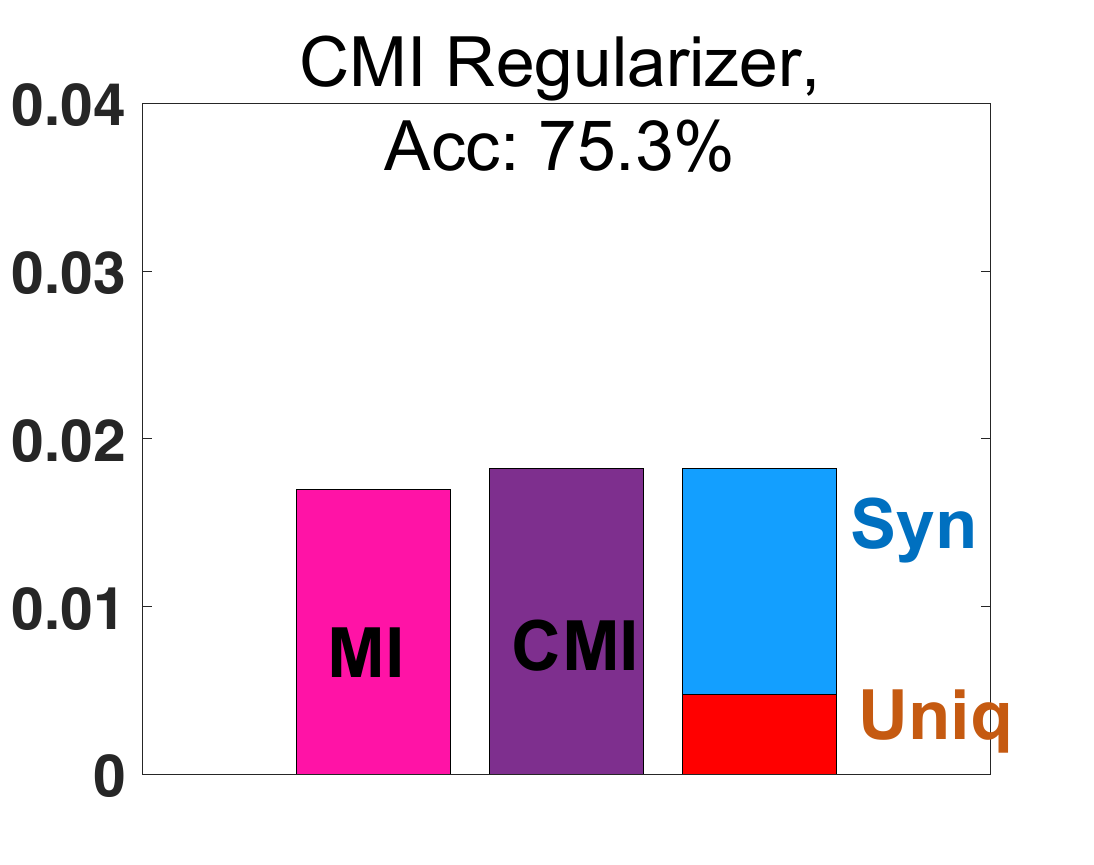}
	\includegraphics[height=3.5cm]{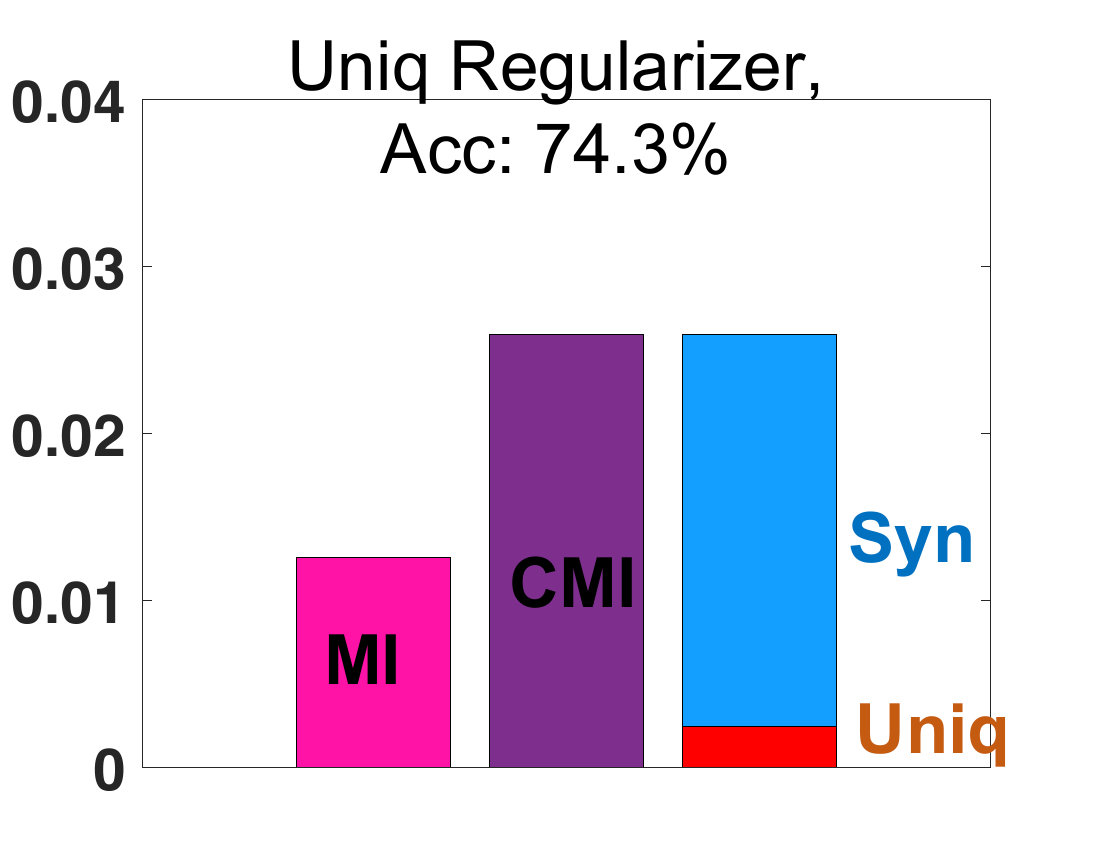}
	\includegraphics[height=3.5cm]{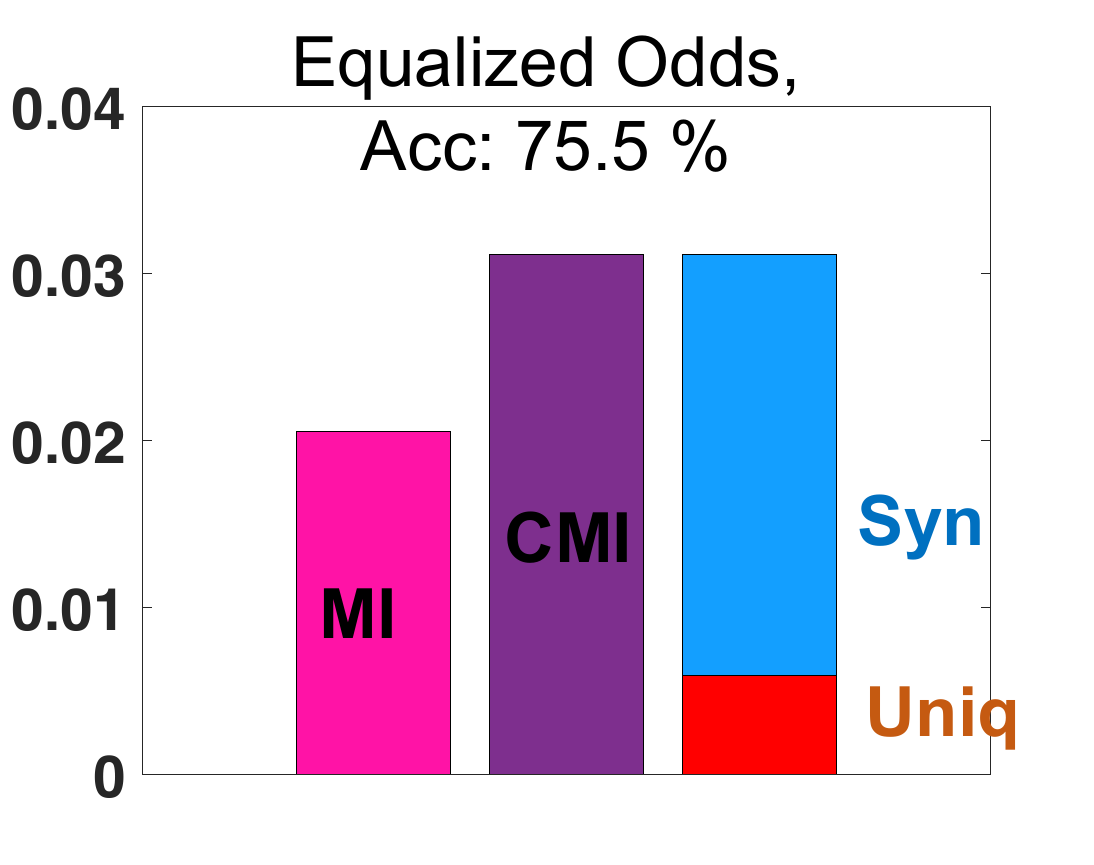}
	\caption{The experimental results demonstrate that the correlation-based estimates for the regularizers behave as expected. When the model is trained for statistical parity, MI reduces as expected without significantly reducing CMI. Next, when CMI is used as a regularizer, we notice that CMI (and its sub-components Uniq and Syn) reduce as expected. For Uniq as a regularizer, we notice that MI or CMI are not reduced that much, but only Uniq is minimized selectively. Lastly, for equalized odds, we observe that the trained model still retains quite a bit of MI, CMI, Uniq and Syn, as compared to the model with no fairness. \label{fig:german}}
\end{figure*}

\section{Discussion and Conclusion}
\label{sec:conclusion}

{\color{black}
\textbf{On Choice of Critical Features and Connections with Explainability:} In this work, as also in some existing works on fairness~\cite{kamiran2013quantifying,kilbertus2017avoiding}, we assume that the critical features are known. We adopt a viewpoint stated in \cite{pandeyblog} which suggests that ``We can't just rely on the math; we still need a human person applying human judgements.'' Since most of these exemptions are embedded in law and social science~\cite{barocas2016big,USEqualPay,grover1995business}, we believe that fairness researchers need to collaborate with social scientists and lawyers in order to determine which set of features can be designated as critical for a particular application. 

This work also shares close connections with the field of \emph{explainability} in machine learning~\cite{datta2016algorithmic,varshney2019trustworthy,kim2018interpretability}, and motivates several related research problems, e.g., how to check or explain if certain features contributed to the disparity in a model, or how to incorporate exemptions in applications, such as, image processing, where certain neurons in an intermediate hidden layer might need to be exempted instead of the input layer because they often have more interpretability~\cite{kim2018interpretability}.

\textbf{On Better Understanding of Observational Measures:} Our proposed counterfactual measure and the desirable properties help in evaluation of observational measures in practice, and understand their utility and limitation, i.e., what they capture and miss. Finally, in applications where when the true SCM is known or can be evaluated from the data \cite[Chapters 4,7]{peters2017elements}, the proposed measure exactly captures the non-exempt disparity.

\textbf{On Uniqueness, Operational Meaning and Further Generalizations:} We acknowledge that we do not prove uniqueness of our measure with respect to the desirable properties, and neither do we show that the properties are exhaustive (recall Remark~\ref{rem:uniqueness} in Section~\ref{subsec:rationale}). This is an interesting direction of future work. However, there may also be value in the fact that the properties do not yield a unique measure: this allows for tuning the measure based on the application. E.g., Shannon established uniqueness on entropy with respect to \textbf{some} properties in \cite{shannon1948mathematical} but subsequent applications have still led to the use of modified measures, e.g. Renyi entropy~\cite{renyi1961measures,liao2019learning,issa2019operational,liao2019robustness}.

Deriving the exact operational meaning of our proposed counterfactual measure is also an interesting direction of future work. Nonetheless, the proposed measure does satisfy our stated desirable properties and capture important aspects of the problem, e.g., statistically visible and masked disparities. Furthermore, our measure can also be modified to account for further functional generalizations. First notice, that our proposed Property~\ref{propty:masking} is a special case of the following statement: 

\textit{If $(Z,f_a(U_X)){-}X_c{-}(\hat{Y},f_b(U_X))$ form a Markov chain for any deterministic functions $f_a(\cdot)$ and $f_b(\cdot)$ such that $f_a(U_X)\independent f_b(U_X)$ and $\mathrm{H}(U_X){=}\mathrm{H}(f_a(U_X))+ \mathrm{H}(f_b(U_X))$, then $M_{NE}{=}0$. } 

To account for this more general property, our proposed measure might be modified as follows:
\begin{align}
\min_{f_a(U_X), f_b(U_X)} \uni{(Z,f_a(U_X))}{(\hat{Y},f_b(U_X)  )| X_c},   
\end{align}
such that $f_a(U_X)\independent f_b(U_X)$ and $\mathrm{H}(U_X)=\mathrm{H}(f_a(U_X))+ \mathrm{H}(f_b(U_X))$. This measure also satisfies all the other desirable properties. In this work, we restrict ourselves to $f_a(U_X)$ and $f_b(U_X)$ being disjoint subsets of $U_X$ for simplicity, computability and ease of understanding. Future work will explore how different assumptions on the SCM restrict the class of $f_a$ and $f_b$.

\textbf{On Understanding Other Forms of Masked Disparity:}  Let us revisit the discussion from Section~\ref{subsec:rationale} that not all forms of masked discrimination are necessarily undesirable. E.g., if  $U_{X_1}$ is a random coin flip in Canonical Example~\ref{cexample:masking_general}, then performing $\hat{Y}=Z \oplus U_{X_1}$ randomizes the race, and can even be regarded as a preventive measure against discrimination.  However, keeping the mathematics of the example same, if  $U_{X_1}$ instead denotes whether one's income is above a threshold, then the model is unfair. It is an interesting future direction to examine how to quantify non-exempt discrimination while allowing the user with more flexibility on what latent factors are allowed to mask $Z$.

\textbf{On Estimation of Mutual Information, Conditional Mutual Information and Unique Information:} In general, it is difficult to directly incorporate these information-theoretic measures as a regularizer with the loss function (see \cite{pal2010estimation,mukherjee2019ccmi} and the references therein). Examining alternate methods of incorporating our proposed measures as regularizer (using or building upon techniques proposed in \cite{mukherjee2019ccmi,fairMI,debiasing,liao2019learning,liao2019robustness,galhotra2020fair,xu2020algorithmic}) is an interesting direction of future work.
}

\ifCLASSOPTIONcaptionsoff
  \newpage
\fi





\appendices

\section{Counterfactual Causal Influence (CCI) and its connection to Counterfactual Fairness}

\subsection{Proof of Lemma~\ref{lem:cci}}
\label{app:cci}
Here, we first provide a proof of Lemma~\ref{lem:cci} which shows that our proposed quantification of total disparity is zero if and only if $\mathrm{CCI}(Z \rightarrow \hat{Y})=0$. For ease of reading, we repeat the statement of the lemma here again.

\cci*

\begin{proof}[Proof of Lemma~\ref{lem:cci}]
From the definition of CCI (Definition~\ref{defn:cci} in Section~\ref{subsec:sys_model}),
\begin{align}
\mathrm{CCI}(Z \rightarrow \hat{Y})  &=\E{Z,Z',U_{X}}{|  h(Z,U_{X})-h(Z',U_{X})| } \nonumber \\
&= \sum_{z_1,z_2,u_x} \Pr(Z=z_1,Z'=z_2, U_{X}=u_x )  | h(z_1,u_x)-h(z_2,u_x)|  \nonumber  \\
& = \sum_{z_1,z_2,u_x} \Pr(Z=z_1)\Pr(Z'=z_2) \Pr(U_{X}=u_x)
 | h(z_1,u_x)-h(z_2,u_x)|. 
\end{align}
Here, the last line holds due to independence. The summation consist of non-negative terms. Therefore, $\mathrm{CCI}(Z \rightarrow \hat{Y})= 0$, \emph{if and only if} all the terms in the summation are zero, \textit{i.e.}, for all $z_1$, $z_2$ and $u_x$ with $\Pr(Z=z_1), \Pr(Z=z_2),\Pr(U_{X}=u_x) >0 $, $|h(z_1,u_x) - h(z_2,u_x)|=0$. This is  equivalent to $h(z,u_x)$ being constant over all possible values of $z$ with $\Pr(Z=z)>0$ given a fixed value of $u_x$, and this should happen over all values of $u_x$ with $\Pr(U_{X}=u_x)$.

Now, observe that,
\begin{align}
\mut{Z}{(\hat{Y},U_X)}  &=\mut{Z}{\hat{Y}\given U_X} + \mut{Z}{U_X} \\
& =\mut{Z}{\hat{Y}\given U_X} &&  [Z \independent U_X]\\
& = \mathrm{H}(\hat{Y}\given U_X) - \mathrm{H}(\hat{Y}\given U_X,Z) && \text{[By Definition]} \\
& = \mathrm{H}(\hat{Y}\given U_X). && \text{[$\hat{Y}$ determined by $Z,U_X$]}
\end{align}
$\mathrm{H}(\hat{Y}\given U_X)$ can be $0$ \emph{if and only if} $h(z,u_x)$ is constant over all possible values of $z$ with $\Pr(Z=z)>0$ given a fixed value of $u_x$, and this should happen over all $u_x$ with $\Pr(U_{X}=u_x)>0$. Thus, $\mathrm{CCI}(Z \rightarrow \hat{Y})= 0$ if and only if $\mut{Z}{(\hat{Y},U_X)}=0.$
\end{proof}

\subsection{Connections to Counterfactual Fairness}
\label{app:cci_additional}

We note that the concept of counterfactual causal influence (often referred to as only ``influence'') is derived from a separate body of work~\cite{breiman2001random,datta2016algorithmic,koh2017understanding,adler2018auditing,henelius2014peek}) outside the fairness literature. The original definition of counterfactual fairness in \cite{kusner2017counterfactual} was stated differently (without using CCI), although the connection with CCI has been hinted at in \cite{russell2017worlds}. Here, for the sake of completeness, we will formally show in Lemma~\ref{lem:cci_imply_cf} that $\mathrm{CCI}(Z\rightarrow \hat{Y})=0$ is equivalent to the counterfactual fairness criterion proposed in \cite{kusner2017counterfactual}. What this means is that, our proposed quantification of total disparity is also $0$ if and only if a model is counterfactually fair.

First, we clarify the differences in notation between our work and \cite{kusner2017counterfactual}. In our work, $X=f(Z,U_{X})$ and $\hat{Y}=r(X)=r\circ f(Z,U_{X})=h(Z,U_{X})$ where $h=r\circ f$. In \cite{kusner2017counterfactual}, $\hat{Y}_{Z\leftarrow z_1} (U)$ denotes the random variable $\hat{Y}$ when the value of $Z$ is fixed as $z_1$ by an intervention, i.e., $\hat{Y}_{Z\leftarrow z_1} (U)=h(z_1,U_X)$. Alongside, we also clarify that the event that $X$ takes the value $x$ when $Z$ is fixed as $z_1$ refers to the event that $U_X$ takes a value from the set $\mathcal{S}(x,z_1)=\{u_x :\  x=f(z_1,u_x),\ \Pr(U_X=u_x)>0\}$ because $X=f(Z,U_X)$. 
\begin{defn}[Counterfactual Fairness given $X=x$ and $Z=z_1$~\cite{kusner2017counterfactual}]
\label{defn:cf}
A predictor $\hat{Y}$ is counterfactually fair given the protected attribute $Z=z_1$ and the observed variable $X=x$, if we have, 
\begin{align}
&\Pr(\hat{Y}_{Z\leftarrow z_1} (U) = y | \text{$X$ takes value $x$ when $Z$ fixed as $z_1$} ) \nonumber \\
&= \Pr(\hat{Y}_{Z\leftarrow z_2}  (U) = y | \text{$X$ takes value $x$ when $Z$ fixed as $z_1$}), 
\end{align}
for all attainable $y$ and $z_2$. In our notations, this definition is equivalent to the following:
Given the sensitive attribute $Z=z_1$ and the observed variable $X=x$, 
\begin{equation}
\Pr(h(z_1,U_{X}) = y \given  \ U_{X} \in \mathcal{S}(x,z_1) )  \\ =\Pr(h(z_2,U_{X})) = y \given  \ \ U_{X} \in \mathcal{S}(x,z_1) ), \label{eq:cf}
\end{equation}
for all attainable $y$ and $z_2$, where $\mathcal{S}(x,z_1)=\{u_x :\  x=f(z_1,u_x),\ \Pr(U_X=u_x)>0\}.$
\end{defn}

Next, we show that $\mathrm{CCI}(Z \rightarrow \hat{Y})=0$ is equivalent to the counterfactual fairness criterion of \cite{kusner2017counterfactual}.

\begin{lem}
\label{lem:cci_imply_cf} $\mathrm{CCI}(Z\rightarrow \hat{Y})=0$ is equivalent to counterfactual fairness (Definition~\ref{defn:cf}) for all $X=x$ and $Z=z_1$ with $\Pr(X=x,Z=z_1)>0$.
\end{lem}

\begin{proof}[Proof of Lemma~\ref{lem:cci_imply_cf}]
Suppose that, $\mathrm{CCI}(Z\rightarrow \hat{Y})=0$.
 Recall from Lemma~\ref{lem:cci}, that $\mathrm{CCI}(Z\rightarrow \hat{Y})=0$ is equivalent to the criterion that $h(z_1,u_x)=h(z_2,u_x)$ for all attainable $z_1$, $z_2$ given a particular value of $u_x$, and this should hold for all $u_x$ with $\Pr(U_{X}=u_x)>0$. Therefore, for any particular $X=x$ and $Z=z_1$ with $\Pr(X=x,Z=z_1)>0$,
\begin{equation}
\Pr(h(z_1,U_{X}) = y \given  \ U_{X} \in \mathcal{S}(x,z_1) )  \\ =\Pr(h(z_2,U_{X})) = y \given  \ \ U_{X} \in \mathcal{S}(x,z_1) ),
\end{equation}
because $h(z_1,u_{x})=h(z_2,u_x)$ for all $u_{x} \in \mathcal{S}(x,z_1)$. Thus, we show that $\mathrm{CCI}(Z\rightarrow \hat{Y})=0$ implies counterfactual fairness. 

Now, we prove the implication in the other direction. Suppose that the counterfactual fairness criterion \eqref{eq:cf} holds for all $X=x$ and $Z=z_1$ with $\Pr(X=x,Z=z_1)>0$.

First consider any particular $X=x$ and $Z=z_1$ with $\Pr(X=x,Z=z_1)>0$. Since $\Pr(X=x,Z=z_1)>0$, there exists at least one $u_x$ with $\Pr(U_X=u_x)>0$ such that $x=f(z_1,u_x)$. So, the set $\mathcal{S}(x,z_1)$ is non-empty. Equation \eqref{eq:cf} implies that,

\begin{equation}
\Pr(h(z_1,U_{X}) = y |  \ U_{X} \in \mathcal{S}(x,z_1) )  \\ =\Pr(h(z_2,U_{X})) = y |  \ \ U_{X} \in \mathcal{S}(x,z_1) ) \forall \text{attainable } y,z_2.
\end{equation}

This leads to, 
\begin{equation}
\Pr(h(z_1,U_{X}) = y , \ U_{X} \in \mathcal{S}(x,z_1) )  \\=\Pr(h(z_2,U_{X}) = y , \ \ U_{X} \in \mathcal{S}(x,z_1) )\ \forall \text{ attainable } y,z_2. 
\end{equation}
Or, 
\begin{equation}
\sum_{u_x \in \mathcal{S}(x,z_1)}\Pr(U_X=u_x) \mathbbm{1}(h(z_1,u_{x}) = y ) 
\\ =\sum_{u_x \in \mathcal{S}(x,z_1)}\Pr(U_X=u_x) \mathbbm{1}(h(z_2,u_{x}) = y ). 
 \label{eq:cf_implies_cci}
\end{equation}

Now, observe that, $f(z_1,u_x)=x$ for all $u_x \in \mathcal{S}(x,z_1)$, and thus $h(z_1,u_{x})=r \circ f(z_1,u_{x})$ takes the same value for all $u_x \in \mathcal{S}(x,z_1)$. Let $h(z_1,u_{x})=\tilde{y}$ for all $u_x \in \mathcal{S}(x,z_1)$. Then, for \eqref{eq:cf_implies_cci} to hold, we need, 
\begin{equation}
\sum_{u_x \in \mathcal{S}(x,z_1)}\Pr(U_X=u_x)(1-\mathbbm{1}(h(z_2,u_{x}) \\ = \tilde{y} )) =0 \ \forall \text{ attainable } z_2.
\end{equation}
This holds if and only if $\mathbbm{1}(h(z_2,u_{x}) = \tilde{y})=1$ for all $u_x \in \mathcal{S}(x,z_1)$ and for all attainable $z_2$. Thus, the counterfactual fairness criterion \eqref{eq:cf} \emph{for a particular} $X=x,Z=z_1$ with $\Pr(X=x,Z=z_1)>0$ implies that for all $u_x \in \mathcal{S}(x,z_1)$,  
\begin{equation}
h(z_2,u_{x})=h(z_1,u_{x}) \ \ \forall \text{ attainable } z_2. \label{eq:cf_cci}
\end{equation}

Because the counterfactual criterion \eqref{eq:cf} holds for all $X=x,Z=z_1$ with $\Pr(X=x,Z=z_1)>0$, we therefore have \eqref{eq:cf_cci} hold for all $$u_x \in \cup_{\{x,z_1: \Pr(X=x,Z=z_1)>0\}} \mathcal{S}(x,z_1).$$

Now, because $U_X$ is independent of $Z$, for any $u^*_x$ with  $\Pr(U_X=u^*_x)>0,$ there always exists some $x^*$ such that $x^*=f(z_1,u^*_x)$, and $\Pr(X=x^*,Z=z_1)\geq \Pr(U_X=u_x^*,Z=z_1)>0$. Thus, $u^*_x \in S(x^*,z_1)$ for some $(x^*,z_1)$ with $\Pr(X=x^*,Z=z_1)>0.$ Thus,
$$ \{u_x:\Pr(U_X=u_x)>0\} \subseteq \cup_{\{x,z_1: \Pr(X=x,Z=z_1)>0\}} \mathcal{S}(x,z_1),$$
implying that $h(z_2,u_{x})=h(z_1,u_{x})$ for all attainable $z_1,z_2$ given a particular value of $u_x$, and this holds for all $u_x$ with $\Pr(U_X=u_x)>0.$ This is equivalent to  $\mathrm{CCI}(Z\rightarrow \hat{Y})=0$ (recall Lemma~\ref{lem:cci}).



\end{proof}

\section{Relevant Information-Theoretic Properties}
\label{app:pid_properties}

\begin{lem}[Conditional DPI]
For all $(A,A',B,X_c)$ such that $(B,X_c)-A-A'$ form a Markov chain, we have the following conditional form of the Data Processing Inequality (DPI): $\mut{A}{B\given X_c}\geq \mut{A'}{B\given X_c}.$
\label{lem:cDPI}
\end{lem}

\begin{proof}[Proof of Lemma~\ref{lem:cDPI}]
From the Markov chain, we have $\mut{A'}{(B,X_c)\given A}=0.$ Because, $\mut{A'}{(B,X_c)\given A}= \mut{A'}{X_c\given A}+ \mut{A'}{B\given A,X_c}$ by chain rule and mutual information is non-negative, we also have $\mut{A'}{B\given A,X_c}=0.$ Now, similar to the proof of DPI, we have:
\begin{align}
\mut{A'}{B\given X_c}+\mut{A}{B\given A',X_c}  =
\mut{A}{B\given X_c}+\mut{A'}{B\given A,X_c} =\mut{A}{B\given X_c},
\end{align}
because $\mut{A'}{B\given A,X_c}=0.$ This leads to $\mut{A}{B\given X_c}\geq \mut{A'}{B\given X_c}.$
\end{proof}

\begin{lem}[Triangle Inequality of Unique Information]
For all $(Z,B,A,X_c)$, we have:
$$\uni{Z}{A| X_c} \leq \uni{Z}{A| B} + \uni{Z}{B| X_c}   .$$
\label{lem:triangle_inequality}
\end{lem}

This result is derived in \cite[Proposition  2]{rauh2019unique}.

\begin{lem}[Monotonicity under local operations on $Z$]
Let $Z'=f(Z)$ where $f(\cdot)$ is a deterministic function. Then, we have:
$$\uni{Z}{B| X_c} \geq \uni{Z'}{B| X_c}   .$$
\label{lem:monotonicity_alice}
\end{lem}

This result is derived in \cite[Lemma 31]{banerjee2018unique}. We include a proof for completeness.
\begin{proof}[Proof of Lemma~\ref{lem:monotonicity_alice}] Let $P'$ be the true joint distribution of $(Z',B,X_c)$ and $P$ be the true joint distribution of $(Z,B,X_c)$. Also let $Q^*= \arg \min_{Q\in \Delta_P} \mathrm{I}_{Q}(Z;B\given X_c)$ where $\Delta_P$ is the set of all joint distributions of $(Z,B,X_c)$ with the same marginals between $(Z,B)$ and $(Z,X_c)$ as the true joint distribution $P$. Let us also define
$$Q'^{*}(z',b,x_c)=\sum_{z}\Pr(z'\given z) Q^*(z,b,x_c),$$ where $\Pr(z'\given z)$ is the true conditional distribution of $Z'=f(Z)$ given $Z$.

Now, observe that,
\begin{align}
 \uni{Z}{B| X_c} 
&= \min_{Q\in \Delta_P} \mathrm{I}_{Q}(Z;B\given X_c) &&[\text{By Definition}]\nonumber \\
 & =  \mathrm{I}_{Q^*}(Z;B\given X_c) &&[\text{By Definition of }Q^*] \nonumber \\
& \overset{(a)}{\geq} \mathrm{I}_{Q'^{*}}(Z';B\given X_c) \nonumber \\
& \overset{(b)}{\geq} \min_{Q' \in \Delta_{P'}}\mathrm{I}_{Q'}(Z';B\given X_c)  \nonumber \\
& = \uni{Z'}{B| X_c} &&[\text{By Definition}].
\end{align}
Here (a) holds using the conditional form of the Data Processing inequality (Lemma~\ref{lem:cDPI}) as follows. Consider the random variables $(Z,B,X_c)$ following distribution $Q^*$ and $Z'=f(Z)$. Then, $(B,X_c)-Z-Z'$ form a Markov chain. Also note that (b) holds because $Q'^*$ belongs to $\Delta_{P'}$ which is the set of all joint distributions of $(Z',B,X_c)$ with the same marginals between $(Z',B)$ and $(Z',X_c)$ as the true joint distribution $P'.$
\end{proof}

\begin{lem}[Monotonicity under local operations on $B$] Let $B'=f(B)$ where $f(\cdot)$ is a deterministic function. Then, we have:
$$\uni{Z}{B| X_c} \geq \uni{Z}{B'| X_c}   .$$
\label{lem:monotonicity_bob}
\end{lem}
This result is derived in \cite[Lemma 31]{banerjee2018unique}. We include a proof for completeness.

\begin{proof}[Proof of Lemma~\ref{lem:monotonicity_bob}] Let $P'$ be the true joint distribution of $(Z,B',X_c)$ and $P$ be the true joint distribution of $(Z,B,X_c)$. Also let $Q^*= \arg \min_{Q\in \Delta_P} \mathrm{I}_{Q}(Z;B\given X_c)$ where $\Delta_P$ is the set of all joint distributions of $(Z,B,X_c)$ with the same marginals between $(Z,B)$ and $(Z,X_c)$ as the true joint distribution $P$. Let us also define
$$Q'^{*}(z,b',x_c)=\sum_{b}\Pr(b'\given b) Q^*(z,b,x_c),$$ where $\Pr(b'\given b)$ is the true conditional distribution of $B'=f(B)$ given $B$.

Now, observe that,
\begin{align}
 \uni{Z}{B| X_c} 
 &= \min_{Q\in \Delta_P} \mathrm{I}_{Q}(Z;B\given X_c) &&[\text{By Definition}]\nonumber \\
 & =  \mathrm{I}_{Q^*}(Z;B\given X_c) &&[\text{By Definition of }Q^*] \nonumber \\
& \overset{(a)}{\geq} \mathrm{I}_{Q'^{*}}(Z;B'\given X_c) \nonumber \\
& \overset{(b)}{\geq} \min_{Q' \in \Delta_{P'}}\mathrm{I}_{Q'}(Z;B'\given X_c)  \nonumber \\
& = \uni{Z}{B'| X_c} &&[\text{By Definition}].
\end{align}
Here (a) holds using the conditional form of the Data Processing inequality (Lemma~\ref{lem:cDPI}) as follows. Consider the random variables $(Z,B,X_c)$ following distribution $Q^*$ and $B'=f(B)$. Then, $(Z,X_c)-B-B'$ form a Markov chain. Also note that (b) holds because $Q'^*$ belongs to $\Delta_{P'}$ which is the set of all joint distributions of $(Z,B',X_c)$ with the same marginals between $(Z,B')$ and $(Z,X_c)$ as the true joint distribution $P'.$
\end{proof}

\begin{lem}[Monotonicity under adversarial side information]
For all $(A,B,X_c,X_c')$, we have:
$$\uni{A}{B| (X_c,X_c')} \leq \uni{A}{B| X_c}   .$$
\label{lem:monotonicity_eve}
\end{lem}

This result is derived in  \cite[Lemma 32]{banerjee2018unique}. 

\begin{lem}[Maximal conditional mutual information] Let $A=f(Z,U_{X})$ where $Z \independent U_{X}$ and $B=g(U_{X})$ for some deterministic functions $f(\cdot)$ and $g(\cdot)$ respectively. Then,
\begin{equation}
\mathrm{I}(Z;A\given U_{X}) \geq \mathrm{I}(Z;A\given B).
\end{equation}
\label{lem:maxCMI}
\end{lem}

\begin{proof}[Proof of Lemma~\ref{lem:maxCMI}]
Observe that,
\begin{align}
& \mathrm{I}(Z;U_{X}\given  A,B) ) \geq 0 && \text{[non-negativity property]} \nonumber \\
& \implies \mathrm{H}(Z\given A,B)-\mathrm{H}(Z\given A,B,U_{X} ) \geq 0 && \text{[by definition]} \nonumber \\
& \implies \mathrm{H}(Z\given A,B)-\mathrm{H}(Z\given A,U_{X} ) \geq 0 && \text{[}B=g(U_{X})\text{]} \nonumber \\
& \implies \mathrm{H}(Z)-\mathrm{H}(Z\given A,U_{X} )  \geq \mathrm{H}(Z)-\mathrm{H}(Z\given A,B) \nonumber   \\
& \implies \mathrm{H}(Z| U_{X})-\mathrm{H}(Z| A,U_{X} ) \geq \mathrm{H}(Z| B)-\mathrm{H}(Z|A,B) && \text{[$Z \independent U_{X}$  and  $Z \independent B$]} \nonumber \\
& \implies \mathrm{I}(Z;A\given U_{X}) \geq \mathrm{I}(Z;A\given B).
\end{align}

\end{proof}

\begin{lem}[Absence of counterfactual causal influence] Let $\hat{Y}=h(Z,U_{X})$ where $Z \independent U_{X}$ and $X_c=g(Z,U_{X})$ for some deterministic functions $h(\cdot)$ and $g(\cdot)$ respectively. Then $\mathrm{CCI}(Z\rightarrow \hat{Y})=0$ implies $\uni{Z}{(\hat{Y},U_X)| X_c}=0$ and also $\uni{Z}{\hat{Y}| X_c}=0$. \label{lem:uni_cci}
\end{lem}

\begin{proof}[Proof of Lemma~\ref{lem:uni_cci}]
$\mathrm{CCI}(Z\rightarrow \hat{Y})=0$  is equivalent to $\mut{Z}{(\hat{Y},U_X)}=0$ (using Lemma~\ref{lem:cci}).
Now, 

$$\uni{Z}{(\hat{Y},U_X)| X_c} \overset{(a)}{\leq} \mut{Z}{(\hat{Y},U_X)} =0,$$
where (a) holds from \eqref{eq:pid2} in Section~\ref{subsec:background} and non-negativity of PID.
Also,
$$\uni{Z}{\hat{Y}| X_c} \overset{(a)}{\leq} \mut{Z}{\hat{Y}}  \overset{(b)}{\leq} \mut{Z}{(\hat{Y},U_X)} =0,$$
where (a) holds from \eqref{eq:pid2} in Section~\ref{subsec:background} and non-negativity of PID terms, and (b) holds from the chain rule and non-negativity of mutual information.

\end{proof}

\begin{lem}[Zero-synergy property of deterministic functions] Let $f(Z)$ be any deterministic function of $Z$, and let $X_c$ be any random variable. Then, 
\begin{equation}
 \mathrm{Syn}(Z:(f(Z), X_c))=\mathrm{Syn}(Z:(X_c, f(Z)))=0.
\end{equation}
This leads to $\mathrm{Uni}(Z:f(Z)| X_c)= \mathrm{I}(Z;f(Z)|X_c)$ and $\mathrm{Uni}(Z:X_c| f(Z))= \mathrm{I}(Z;X_c|f(Z))$.
\label{lem:zero_syn}
\end{lem}

\begin{proof}[Proof of Lemma~\ref{lem:zero_syn}:] Recall from the definition of $\uni{Z}{B| X_c}$ that $\Delta$ denotes the set of all joint distributions of $(Z,B,X_c)$ and 
$\Delta_p$ is the set of all such joint distributions that have the same marginals for $(Z,B)$ and $(Z,X_c)$ as the true distribution, \textit{i.e.}, 
\begin{equation}
\Delta_p=\{ Q \in \Delta : \ q(z,b)=\Pr(Z=z,B=b) \text{ and } \\ q(z,x_c)=\Pr(Z=z,X_c=x_c) \}.
\end{equation}
We first show that if $B=f(Z)$, then $\Delta_p$ is only a singleton set which only consists of the true distribution. Observe that, for any $Q \in \Delta_p$,
\begin{align}
&q(z,b,x_c)=q(z)q(b|z)q(x_c|b,z) &&\text{[chain rule of probability]} \nonumber \\
& =\Pr(Z=z)\Pr(B=b|Z=z)q(x_c|b,z) &&\text{[$q(z,b)=\Pr(Z=z,B=b)$]} \nonumber \\
&= \begin{cases} \Pr(Z=z)q(x_c|b,z), &\text{if } b=f(z) \\
0, &\text{otherwise} 
\end{cases} &&\text{[$\Pr(B=b|Z=z)=1$ only if $b=f(z)$]} \nonumber \\
& = \begin{cases} \Pr(Z=z)q(x_c|z), &\text{if } b=f(z)  \\
0, &\text{otherwise} 
\end{cases} &&\text{[$b$ is entirely determined by $z$]}\nonumber \\
& = \begin{cases} \Pr(Z=z)\Pr(X_c=x_c|Z=z), &\text{if } y=f(z) \\
0, &\text{otherwise} 
\end{cases} &&\text{[}q(x_c|z)=\Pr(X_c=x_c|Z=z)\text{]}  \nonumber \\
& = \Pr(Z=z,B=b,X_c=x_c).
\end{align}

Thus, for $B=f(Z)$,
\begin{equation}
\mathrm{Uni}(Z:B| X_c)  = \min_{Q \in \Delta_p} \mathrm{I}_{Q}(Z;B|X_c) = \mathrm{I}(Z;B|X_c).
\end{equation}
This leads to $\mathrm{Syn}(Z:(f(Z), X_c))=\mathrm{I}(Z;f(Z)|X_c)-\mathrm{Uni}(Z:f(Z)| X_c)=0$ (using \eqref{eq:pid3} in Section~\ref{subsec:background}). Note that, $\mathrm{Syn}(Z:(f(Z), X_c))$ is symmetric between $f(Z)$ and $X_c$.
\end{proof}

\section{Appendix to Section~\ref{sec:properties}}

Here, we provide the proofs of the results as well as additional discussion to supplement Section~\ref{sec:properties}. For convenience, we repeat the statements of the results.

\subsection{Proof of Theorem~\ref{thm:satisfythm} and Lemma~\ref{lem:markov_chain}}
\label{app:properties}

\satisfythm*

\begin{proof}[Proof of Theorem~\ref{thm:satisfythm}]
Here, we formally show that our proposed measure satisfies all the four desirable properties. We restate each of the properties again and then show that they are is satisfied.

\cancellation*
\begin{align}
M^*_{NE} & = \min_{U_a,U_b\text{ s.t. } U_a=U_X\backslash U_b} \uni{(Z,U_a)}{(\hat{Y},U_b)| X_c} \nonumber \\
& \leq \uni{Z}{(\hat{Y},U_X)| X_c} \nonumber \\
& \leq \mut{Z}{(\hat{Y},U_X)}.  &&  \text{[\eqref{eq:pid2} in Section~\ref{subsec:background} and non-negativity of PID terms]}
\end{align}
Thus, $\mut{Z}{(\hat{Y},U_X)}=0$ implies $M_{NE}=0$.

\synergy*

\begin{align}
M^*_{NE}& = \min_{U_a,U_b\text{ s.t. } U_a=U_X\backslash U_b} \uni{(Z,U_a)}{(\hat{Y},U_b)| X_c} \nonumber
\\& = \uni{(Z,U^*_a)}{(\hat{Y},U^*_b)| X_c} && \text{[for some $(U_a^*,U_b^*)$]} \nonumber
\\ & \geq \uni{Z}{(\hat{Y},U^*_b)| X_c} && \text{[Using Lemma~\ref{lem:monotonicity_alice}]} \nonumber \\
& \geq \uni{Z}{\hat{Y}| X_c}. && \text{[Using Lemma~\ref{lem:monotonicity_bob}]}
\end{align}
Thus, $\uni{Z}{\hat{Y}| X_c}>0$ implies that $M_{NE}>0$.

\nonexemptmasking*

First we will show that $M_{NE}^*>0$ for the canonical example of non-exempt disparity where $\hat{Y}=Z \oplus U_{X_1}$ where $Z$ lies in the non-critical/general features and $U_{X_1}$ can be either critical or non-critical.

\textbf{Case 1:} $X_c=U_{X_1}$, $X_g=Z$ and $\hat{Y}=Z\oplus U_{X_1}$ with $Z,U_{X_1}\sim \iid$ Bern(\nicefrac{1}{2}).

We will check the value of $\uni{(Z,U_a)}{(\hat{Y},U_b)| X_c}$ for different choices of $U_a$ to find the minimum.

For $U_a=\phi$ and $U_b=U_{X_1}$, we have
\begin{align}
&\uni{(Z,U_a)}{(\hat{Y},U_b)| X_c}\nonumber \\
&= \uni{Z}{(\hat{Y},U_{X_1})| X_c} &&[\text{Substituting the variables}]\nonumber \\
& = \mut{Z}{(\hat{Y},U_{X_1})} -\rd{Z}{((\hat{Y},U_{X_1}),X_c)} \nonumber &&[\text{Using \eqref{eq:pid2} in Section~\ref{subsec:background}}]\nonumber\\
& \overset{(a)}{=} \mut{Z}{(\hat{Y},U_{X_1})} \nonumber \\
& = 1 \text{ bit}.
\end{align}
Here (a) holds because $\rd{Z}{((\hat{Y},U_{X_1}),X_c)} \leq \mut{Z}{X_c}$ (using \eqref{eq:pid2} in Section~\ref{subsec:background} and non-negativity of PID terms), and here $\mut{Z}{X_c}=0$.

For $U_a=U_{X_1}$ and $U_b=\phi$, we have
\begin{align}
\uni{(Z,U_a)}{(\hat{Y},U_b)| X_c} 
& = \uni{(Z,U_{X_1})}{\hat{Y}| X_c} &&[\text{Substituting the variables}]\nonumber \\
& = \mut{(Z,U_{X_1})}{\hat{Y}\given X_c} && [\text{Lemma~\ref{lem:zero_syn} as $\hat{Y}$ is deterministic in $Z,U_{X_1}$}]\nonumber \\
&= 1 \text{ bit}.
\end{align}

Thus, $M_{NE}^*
=\min_{U_a,U_b\text{s.t. } U_a=U_X\backslash U_b}\uni{(Z,U_a)}{(\hat{Y},U_b)| X_c}=1\text{ bit},$ which is strictly greater than $0$.

\textbf{Case 2:} $X_c=\phi$, $X_g=(Z,U_{X_1})$ and $\hat{Y}=Z\oplus U_{X_1}$ with $Z,U_{X_1}\sim \iid$ Bern(\nicefrac{1}{2}).

Since $X_c=\phi$, we can use Property~\ref{propty:absence} (proved above) to compute 
$$M_{NE}^*= \mut{Z}{(\hat{Y},U_X)}=1 \text{ bit},$$
which is strictly greater than $0$.
Thus, our proposed measure is non-zero in the canonical example of non-exempt masked disparity. Now, we move on to the proof of the next part of this property.

Suppose that $(Z,U_a)-X_c-(\hat{Y},U_b)$ form a Markov chain for some subsets $U_a,U_b \subseteq U_X$ such that $U_a=U_X\backslash U_b$. Then, $\mut{(Z,U_a)}{(\hat{Y},U_b)\given  X_c}=0$, implying that $\uni{(Z,U_a)}{(\hat{Y},U_b)| X_c}=0$ for those subsets $U_a,U_b \subseteq U_X$ because unique information is a sub-component of conditional mutual information. Therefore, 
$$M^*_{NE} = \min_{U_a,U_b\text{ s.t. } U_a=U_X\backslash U_b} \uni{(Z,U_a)}{(\hat{Y},U_b)| X_c} \leq 0.$$
Again, using the fact that unique information is non-negative, we have,
$$M^*_{NE} = \min_{U_a,U_b\text{ s.t. } U_a=U_X\backslash U_b} \uni{(Z,U_a)}{(\hat{Y},U_b)| X_c} \geq 0.$$
Thus, $M^*_{NE}=0$.

\absence*

When $X_c=\phi$, we have $\uni{Z,U_a}{\hat{Y},U_b| X_c}=\mut{Z,U_a}{\hat{Y},U_b}.$ We are required to show that $$\min_{U_a,U_b\text{ s.t. } U_a=U_X\backslash U_b}\mut{Z,U_a}{\hat{Y},U_b}$$ is equal to $\mut{Z}{(\hat{Y},U_X)}.$ Note that,
\begin{align}
& \mut{Z,U_a}{\hat{Y},U_b} = \mathrm{H}(\hat{Y},U_b)-\mathrm{H}(\hat{Y},U_b\given Z,U_a) && \text{[By Definition]} \nonumber \\
& = \mathrm{H}(\hat{Y}\given U_b)+ \mathrm{H}(U_b)-\mathrm{H}(U_b\given Z,U_a)-\mathrm{H}(\hat{Y}\given U_b,Z,U_a) && \text{[Chain Rule]} \nonumber \\
& = \mathrm{H}(\hat{Y}\given U_b)+ \mathrm{H}(U_b)-\mathrm{H}(U_b\given Z,U_a) && \text{[$\hat{Y}$ is entirely determined by $Z,U_a,U_b$]} \nonumber \\
& = \mathrm{H}(\hat{Y}\given U_b) && \text{[$Z,U_a,U_b$ are mutually independent]} \nonumber \\
& \geq \mathrm{H}(\hat{Y}\given U_X) && \text{[conditioning reduces entropy]} \nonumber \\
& = \mathrm{H}(\hat{Y}\given U_X) - \mathrm{H}(\hat{Y}\given Z,U_X) + \mut{Z}{U_X} && \text{[$\hat{Y}$ entirely determined by $Z,U_X$, and $Z\independent U_X$]} \nonumber \\
& = \mut{Z}{\hat{Y}\given U_X} + \mut{Z}{U_X} && \text{[By Definition]} \nonumber \\
& = \mut{Z}{(\hat{Y},U_X)}. && \text{[By Chain Rule]}
\end{align}
Thus, 
$\mut{Z,U_a}{\hat{Y},U_b}\geq \mut{Z}{(\hat{Y},U_X)}$ with equality when $U_b=U_X,U_a=\phi$.

\nonincreasing*

Let $X_c'$ denote the additional feature that is to be removed from $X_g$ and is to be added to $X_c$.
From Lemma~\ref{lem:monotonicity_eve}, we have,
\begin{align}
\uni{(Z,U_a)}{(\hat{Y},U_b)| (X_c,X_c')}  \leq \uni{(Z,U_a)}{(\hat{Y},U_b)| X_c},
\end{align}
for any $U_a,U_b$. Thus,
\begin{align}
\min_{U_a,U_b\text{ s.t. } U_a=U_X\backslash U_b}\uni{(Z,U_a)}{(\hat{Y},U_b)| (X_c,X_c')} 
\leq \min_{U_a,U_b\text{ s.t. } U_a=U_X\backslash U_b} \uni{(Z,U_a)}{(\hat{Y},U_b)| X_c}.
\end{align}

\completeexemption*

Observe that, when $X=X_c$,
\begin{align}
M^*_{NE}& = \min_{U_a,U_b\text{ s.t. } U_a=U_X\backslash U_b} \uni{(Z,U_a)}{(\hat{Y},U_b)| X} \nonumber
\\&\leq \uni{Z,U_X}{\hat{Y}| X} \nonumber\\
& \leq \mut{Z,U_X}{\hat{Y}\given X}  && \text{[\eqref{eq:pid3} in Section~\ref{subsec:background} and non-negativity of PID terms]}\nonumber\\
& = \mathrm{H}(\hat{Y}\given X)- \mathrm{H}(\hat{Y}\given  Z,U_X,X) && \text{[By Definition]}\nonumber\\
& = 0. && \text{[$\hat{Y}$ is a deterministic function of $X$]}
\end{align}

\end{proof}

\markovchain*

\begin{proof}[Proof of Lemma~\ref{lem:markov_chain}]
We note that the terms $\mut{Z}{\hat{Y}\given X_c}$, $\mut{Z}{(\hat{Y},U_b)\given X_c}$ and $\mut{(Z,U_a)}{\hat{Y}\given X_c}$ are all less than or equal to $\mut{(Z,U_a)}{(\hat{Y},U_b)\given X_c} $ using the chain rule and non-negativity of conditional mutual information. 

Thus, if $\mut{(Z,U_a)}{(\hat{Y},U_b)\given X_c}=0$, then all those three terms are also $0$.
\end{proof}

\subsection{Supporting Derivations}
\label{app:properties_supporting}

Here, we include the supporting derivations for some of our statements in Section~\ref{subsec:main_results} and Section~\ref{subsec:rationale}.

\noindent \textbf{Supporting Derivation 1: $\uni{Z}{\hat{Y}| X_c}>0$ for Canonical Example~\ref{cexample:equalized_odds} (discrimination in admissions).}

\begin{proof}
Recall that for this example, $X_c=U_{X_1}$, $X_g=Z\oplus U_{X_2}$,  and $\hat{Y}=U_{X_1}+Z+U_{X_2}$ with $Z,U_{X_1},U_{X_2}\sim$ \iid{} Bern(\nicefrac{1}{2}).
The claim can be verified as follows:
\begin{align*}
\uni{Z}{\hat{Y}| X_c}&=\mut{Z}{\hat{Y}} -\rd{Z}{(\hat{Y}, X_c)} &&[\text{using \eqref{eq:pid2} in Section~\ref{subsec:background}}]\\
&\overset{(a)}{\geq}\mut{Z}{\hat{Y}} - \mut{Z}{X_c} \\
& \overset{(b)}{=} \mut{Z}{\hat{Y}}\\
& \overset{(c)}{>} 0,
\end{align*}
where (a) holds because $\rd{Z}{(\hat{Y}, X_c)} \leq \mut{Z}{X_c}$ (using \eqref{eq:pid2} in Section~\ref{subsec:background} and non-negativity of all PID terms) and (b) holds because $\mut{Z}{X_c}=0$. Lastly, (c) holds because $\hat{Y}$ and $Z$ are not independent of each other for this specific example.
\end{proof}

\noindent \textbf{Supporting Derivation 2: $\uni{Z}{\hat{Y}| X_c}>0$ for Canonical Example~\ref{cexample:synergy} (discrimination by unmasking).}

\begin{proof}
Recall that for this example, $X_c=Z\oplus U_{X_1}$, $X_g=U_{X_1}$ and $\hat{Y}=Z$ with $Z,U_{X_1}\sim$ \iid{} Bern(\nicefrac{1}{2}).

The claim can be verified as follows:
\begin{align*}
\uni{Z}{\hat{Y}| X_c}&=\mut{Z}{\hat{Y}} -\rd{Z}{(\hat{Y}, X_c)} &&[\text{using \eqref{eq:pid2} in Section~\ref{subsec:background}}]\\
&\overset{(a)}{\geq}\mut{Z}{\hat{Y}} - \mut{Z}{X_c} \\
& \overset{(b)}{=} 1 \text{ bit},
\end{align*}
where (a) holds because $\rd{Z}{(\hat{Y}, X_c)} \leq \mut{Z}{X_c}$ (using \eqref{eq:pid2} in Section~\ref{subsec:background} and non-negativity of all PID terms) and (b) holds because $\mut{Z}{X_c}=0$. 
\end{proof}

\noindent \textbf{Supporting Derivation 3: $\uni{Z}{(\hat{Y},U_X)| X_c}>0$ in Canonical Example~\ref{cexample:exemption}.}

\begin{proof}
Consider Canonical Example~\ref{cexample:exemption}.
\begin{align*}
\uni{Z}{(\hat{Y},U_X)| X_c} 
&= \uni{Z}{(Z+U_{X_1}+U_{X_2},U_X)| Z+U_{X_1}} &&[\text{Substituting the variables}]\nonumber \\
& \overset{(a)}{\geq} \uni{Z}{Z| Z+U_{X_1}} \nonumber \\
& \overset{(b)}{=} \mut{Z}{Z\given Z+U_{X_1}}\nonumber \\
& \overset{(c)}{>}0.
\end{align*}
Here, (a) holds because $Z$ is a deterministic function of $(Z+U_{X_1}+U_{X_2},U_X)$ and unique information is non-increasing under local operations of $B$ (see Lemma~\ref{lem:monotonicity_bob} in Appendix~\ref{app:pid_properties}). Next, (b) holds because if we consider $\Delta_p$, the set of joint distributions of $(Z,Z,Z+U_{X_1})$, such that the marginals $(Z,Z)$ and $(Z,Z+U_{X_1})$ are the same as the marginals of the true joint distribution, we find that there is only one distribution in this set, which is exactly the true distribution. Thus, $\uni{Z}{Z| Z+U_{X_1}}=\min_{Q\in \Delta_p}\mathrm{I}_{Q}(Z;Z\given Z+U_{X_1})=\mut{Z}{Z\given Z+U_{X_1}}.$ Lastly (c) holds because,  
\begin{align}
 \mut{Z}{Z\given Z+U_{X_1}}
& =\mathrm{H}(Z|Z+U_{X_1})-\mathrm{H}(Z|Z,Z+U_{X_1})\nonumber \\
& =\mathrm{H}(Z|Z+U_{X_1})\nonumber \\
& =\sum_{t=0,1,2}\mathrm{H}(Z|Z+U_{X_1}=t)\Pr(Z+U_{X_1}{=}t).
\end{align}
Using the fact that $Z,U_{X_1}\sim \iid{}$  Bern(\nicefrac{1}{2}), we can compute $\mathrm{H}(Z|Z+U_{X_1}=0)=0$, $\mathrm{H}(Z|Z+U_{X_1}=1)=h_b(\nicefrac{1}{2})=1$, and $\mathrm{H}(Z|Z+U_{X_1}=2)=0$. Here, $h_b(\cdot)$ is the binary entropy function~\cite{cover2012elements} given by $h_b(p)=-p\log_2(p)-(1-p)\log_2(1-p)$. Also note that, $\Pr(Z+U_{X_1}=1)=\nicefrac{1}{2}.$ So, $\mut{Z}{Z\given Z+U_{X_1}}=0.5$ bits.

\end{proof}

\noindent \textbf{Supporting Derivation 4: Exact computation of $\uni{Z}{\hat{Y} | X_c}$ and $M_{NE}^*$ for Canonical Example~\ref{cexample:equalized_odds}.}

\begin{align}
\uni{Z}{\hat{Y} | X_c} 
& \overset{(a)}{=} \mut{Z}{\hat{Y}} \nonumber \\
& = \mathrm{H}(Z)-\mathrm{H}(Z|\hat{Y}) \nonumber \\
& = \mathrm{H}(Z)-\mathrm{H}(Z|U_{X_1}+Z+U_{X_2}) \nonumber \\
& = \mathrm{H}(Z)-\sum_{t=0,1,2,3}\mathrm{H}(Z|U_{X_1}+Z+U_{X_2}=t)\Pr(U_{X_1}+Z+U_{X_2}=t) \nonumber \\
& \overset{(b)}{=} 1-\nicefrac{3}{4}h_b(\nicefrac{1}{3}) \text{ bits}.
\end{align}
Here (a) holds because $\mut{Z}{U_{X_1}}=0$, implying $\rd{Z}{(\hat{Y},U_{X_1})}=0$ as well (using \eqref{eq:pid2} in Section~\ref{subsec:background} and non-negativity of PID terms). Lastly, (b) holds because $Z,U_{X_1},U_{X_2}\sim \iid{} $  Bern(\nicefrac{1}{2}). So, we can exactly compute $\mathrm{H}(Z|U_{X_1}+Z+U_{X_2}=0)=0$, $\mathrm{H}(Z|U_{X_1}+Z+U_{X_2}=1)=h_b(\nicefrac{1}{3})$, $\mathrm{H}(Z|U_{X_1}+Z+U_{X_2}=2)=h_b(\nicefrac{1}{3})$, and $\mathrm{H}(Z|U_{X_1}+Z+U_{X_2}=3)=0$. Here, $h_b(\cdot)$ is the binary entropy function~\cite{cover2012elements} given by $h_b(p)=-p\log_2(p)-(1-p)\log_2(1-p)$. Also note that, $\Pr(U_{X_1}+Z+U_{X_2}=1)=\Pr(U_{X_1}+Z+U_{X_2}=2)=\nicefrac{3}{8}.$

Now, we will examine the value of $\uni{(Z,U_a)}{(\hat{Y},U_b)| X_c}$ for different choices of $U_a$ to find the minimum.

Let $U_a=\phi$ (and $U_b=U_X$). Then, 
\begin{align}
& \uni{(Z,U_a)}{(\hat{Y},U_b) | X_c} = \uni{Z}{(\hat{Y},U_{X_1},U_{X_2}) | U_{X_1}} \nonumber \\
& \overset{(a)}{=} \mut{Z}{U_{X_1}+Z+U_{X_2},U_{X_1},U_{X_2}} \nonumber \\
& = \mut{Z}{U_{X_1},U_{X_2}} + \mut{Z}{U_{X_1}+Z+U_{X_2}\given U_{X_1},U_{X_2}} [\text{Chain Rule}] \nonumber \\
& = \mut{Z}{U_{X_1}+Z+U_{X_2}\given U_{X_1},U_{X_2}} &&[Z \text{ is independent of } U_{X_1},U_{X_2}] \nonumber \\
& = \mathrm{H}(U_{X_1}+Z+U_{X_2}\given U_{X_1},U_{X_2}) \nonumber \\
&\hspace{2cm}- \mathrm{H}(U_{X_1}+Z+U_{X_2}\given Z, U_{X_1},U_{X_2}) &&[\text{By Definition}] \nonumber \\
& = \mathrm{H}(U_{X_1}+Z+U_{X_2}\given U_{X_1},U_{X_2}) &&[\text{Deterministic Function}] \nonumber \\
& = \sum_{u_1,u_2 \in \{0,1\}} \mathrm{H}(U_{X_1}+Z+U_{X_2}\given U_{X_1}=u_1,U_{X_2}=u_2)\Pr(U_{X_1}=u_1,U_{X_2}=u_2) \nonumber \\
& = \sum_{u_1,u_2 \in \{0,1\}}h_b(\nicefrac{1}{2})\Pr(U_{X_1}=u_1,U_{X_2}=u_2) \nonumber \\
& = 1 \text{ bit}.
\end{align}
Here (a) holds again because $\mut{Z}{U_{X_1}}=0$, implying the redundant information is $0$ as well (using \eqref{eq:pid2} in Section~\ref{subsec:background}).

Next, for $U_a=U_{X_2}$ (and $U_b=U_{X_1}$), we have,
\begin{align}
&\uni{(Z,U_a)}{(\hat{Y},U_b) | X_c} = \uni{(Z,U_{X_2})}{(\hat{Y},U_{X_1}) | U_{X_1}} \nonumber \\
&\overset{(a)}{=} \mut{(Z,U_{X_2})}{(\hat{Y},U_{X_1})} \nonumber \\
& = \mut{(Z,U_{X_2})}{U_{X_1}} + \mut{(Z,U_{X_2})}{\hat{Y}\given U_{X_1}} &&[\text{Chain Rule}] \nonumber \\
& = \mut{(Z,U_{X_2})}{\hat{Y}\given U_{X_1}} &&[Z,U_{X_2} \text{ is independent of } U_{X_1}] \nonumber \\
& = \mathrm{H}(U_{X_1}+Z+U_{X_2}\given U_{X_1})-\mathrm{H}(U_{X_1}+Z+U_{X_2}\given U_{X_1},(Z,U_{X_2})) && [\text{By Definition}] \nonumber \\
& = \mathrm{H}(U_{X_1}+Z+U_{X_2}\given U_{X_1}) && [\text{Deterministic Function}] \nonumber \\
& = \sum_{u_1=0,1}\mathrm{H}(U_{X_1}+Z+U_{X_2}\given U_{X_1}=u_1)\Pr(U_{X_1}=u_1)\nonumber \\
& = \nicefrac{1}{4}\log_2{4}+\nicefrac{1}{2}\log_2{2}+\nicefrac{1}{4}\log_2{4}\nonumber \\
&= \nicefrac{3}{2} \text{ bit}.
\end{align}
Here (a) holds again because $\mut{(Z,U_{X_2})}{U_{X_1}}=0$, implying the redundant information is $0$ as well (using \eqref{eq:pid2} in Section~\ref{subsec:background}).

Next, for $U_a=U_{X_1}$ (and $U_b=U_{X_2}$), we have,
\begin{align}
\uni{(Z,U_a)}{(\hat{Y},U_b) | X_c} 
&= \uni{(Z,U_{X_1})}{(\hat{Y},U_{X_2}) | U_{X_1}}\nonumber \\
& \overset{(b)}{=} \mut{(Z,U_{X_1})}{(\hat{Y},U_{X_2})\given U_{X_1}} \nonumber \\
& =\mut{(Z,U_{X_1})}{U_{X_2}\given U_{X_1}}+ \mut{(Z,U_{X_1})}{\hat{Y}\given U_{X_1},U_{X_2}} && [\text{Chain Rule}] \nonumber \\
& = \mut{(Z,U_{X_1})}{\hat{Y}\given U_{X_1},U_{X_2}} && [\text{Mutual Independence}] \nonumber \\
& = \mathrm{H}(\hat{Y}\given U_{X_1},U_{X_2}) - \mathrm{H}(\hat{Y}\given (Z,U_{X_1}), U_{X_1},U_{X_2}) &&[\text{By Definition}] \nonumber \\
& = \mathrm{H}(\hat{Y}\given U_{X_1},U_{X_2}) &&[\text{Deterministic Function}] \nonumber \\
& = \mathrm{H}(U_{X_1}+Z+U_{X_2}\given U_{X_1},U_{X_2}) \nonumber \\
&= 1 \text{ bit}.
\end{align}
Here (b) holds because $\syn{(Z,U_{X_1})}{(A,B)}=0$ if one of the terms $A$ or $B$ is a deterministic function of $(Z,U_{X_1})$ (using Lemma~\ref{lem:zero_syn} in Appendix~\ref{app:pid_properties}) and hence unique information becomes equal to the conditional mutual information (see \eqref{eq:pid3} in Section~\ref{subsec:background}).

Lastly, for $U_a=U_{X}$ (and $U_b=\phi$), we have,
\begin{align}
\uni{(Z,U_a)}{(\hat{Y},U_b) | X_c} & = \uni{(Z,U_{X_1},U_{X_2})}{\hat{Y} | U_{X_1}} \nonumber \\
&\overset{(b)}{=} \mut{(Z,U_{X_1},U_{X_2})}{\hat{Y} \given U_{X_1}}\nonumber \\
& = \mathrm{H}(\hat{Y}\given U_{X_1}) - \mathrm{H}(\hat{Y}\given (Z,U_{X_1},U_{X_2}),U_{X_1}) && [\text{By Definition}] \nonumber \\
& = \mathrm{H}(\hat{Y}\given U_{X_1}) && [\text{Deterministic Function}] \nonumber \\
& = \nicefrac{1}{4}\log_2{4}+\nicefrac{1}{2}\log_2{2}+\nicefrac{1}{4}\log_2{4} \nonumber \\
&= \nicefrac{3}{2} \text{ bit}.
\end{align}
Here (b) holds again using Lemma~\ref{lem:zero_syn} in Appendix~\ref{app:pid_properties}. 

Thus, we obtain that, 
\begin{align}
M_{NE}^* = \min_{U_a,U_b\text{ s.t. } U_a=U_X\backslash U_b}\uni{(Z,U_a)}{(\hat{Y},U_b) | X_c} =1\text{ bit}.
\end{align}

This is strictly greater than $\uni{Z}{\hat{Y} | X_c} =1-\frac{3}{4}h_b(\nicefrac{1}{3}) \text{ bits},$ accounting for both non-exempt statistically visible and non-exempt masked disparities.

\subsection{Discussion on Other Candidate Measures}
\label{app:properties_others}

\noindent \textbf{Why the product of the two measures $\mut{Z}{\hat{Y}\given  X_c}$ and $\mut{Z}{(\hat{Y},U_X)}$ does not work?}

One might recall that the measure $\mut{Z}{\hat{Y}\given  X_c}$ resolved most of the examples except in Canonical Example~\ref{cexample:college} where the output $\hat{Y}$ had no counterfactual causal influence of $Z$ and yet this measure gave a false positive conclusion about non-exempt disparity. This leads us to examine another candidate measure, i.e., product of $\mut{Z}{\hat{Y}\given  X_c}$ and $\mut{Z}{(\hat{Y},U_X)}$ where the latter is always $0$ whenever there is no counterfactual causal influence of $Z$ on $\hat{Y}$.

\begin{candmeas} $M_{NE}=\mut{Z}{\hat{Y}\given  X_c}\times \mut{Z}{(\hat{Y},U_X)}$.
\label{candmeas:measure_mult}
\end{candmeas}

\begin{cexample}
Let $Z=(Z_1,Z_2)$, $X_c=(Z_1 \oplus U_{X_1},Z_2)$, $X_g=(Z_1,U_{X_2})$ and $\hat{Y}=(U_{X_1},Z_2 \oplus U_{X_2} )$ where $Z_1,Z_2,U_{X_1},U_{X_2}$ are  \iid{} Bern(\nicefrac{1}{2}). 
\label{cexample:measure_mult}
\end{cexample}

This example should be exempt because $Z_2$ already appears in $X_{c}$, and is hence exempt. However, both $\mut{Z}{(\hat{Y},U_X)}$ and $\mut{Z}{\hat{Y}\given X_c}$ are non-zero for this example. This leads us to examine another candidate measure, which is essentially the common information-theoretic volume between  $\mut{Z}{(\hat{Y},U_X)}$ and $\mut{Z}{\hat{Y}\given X_c}$, i.e., a measure of the common reason that can make both $\mut{Z}{(\hat{Y},U_X)}>0$ and  $\mut{Z}{\hat{Y}\given X_c}>0$ (overlapping volume). \\

\noindent \textbf{Measure proposed in \cite{dutta2020information}: Information-theoretic sub-volume of the intersection between $\mut{Z}{\hat{Y}\given  X_c}$ and $\mut{Z}{(\hat{Y},U_X)}$:}

The previous Canonical Example demonstrates that both these measures $\mut{Z}{\hat{Y}\given  X_c}$ and $\mut{Z}{(\hat{Y},U_X)}$ can be non-zero for different reasons leading to a false positive conclusion using Candidate Measure~\ref{candmeas:measure_mult}. Intuitively, we need to identify the common reason that makes them non-zero, if any. This motivates us to examine another candidate (Candidate Measure~\ref{candmeas:intersection}) which is the information-theoretic sub-volume of the intersection between these two measures, as shown in Fig.~\ref{fig:inter}.

\begin{candmeas}
\label{candmeas:intersection}
 $M_{NE}=\uni{Z}{(\hat{Y},U_X)| X_c}-\uni{Z}{(\hat{Y},U_X)| (X_c,\hat{Y})}.$
 \end{candmeas}

\begin{figure}
\centering
\includegraphics[width=5cm]{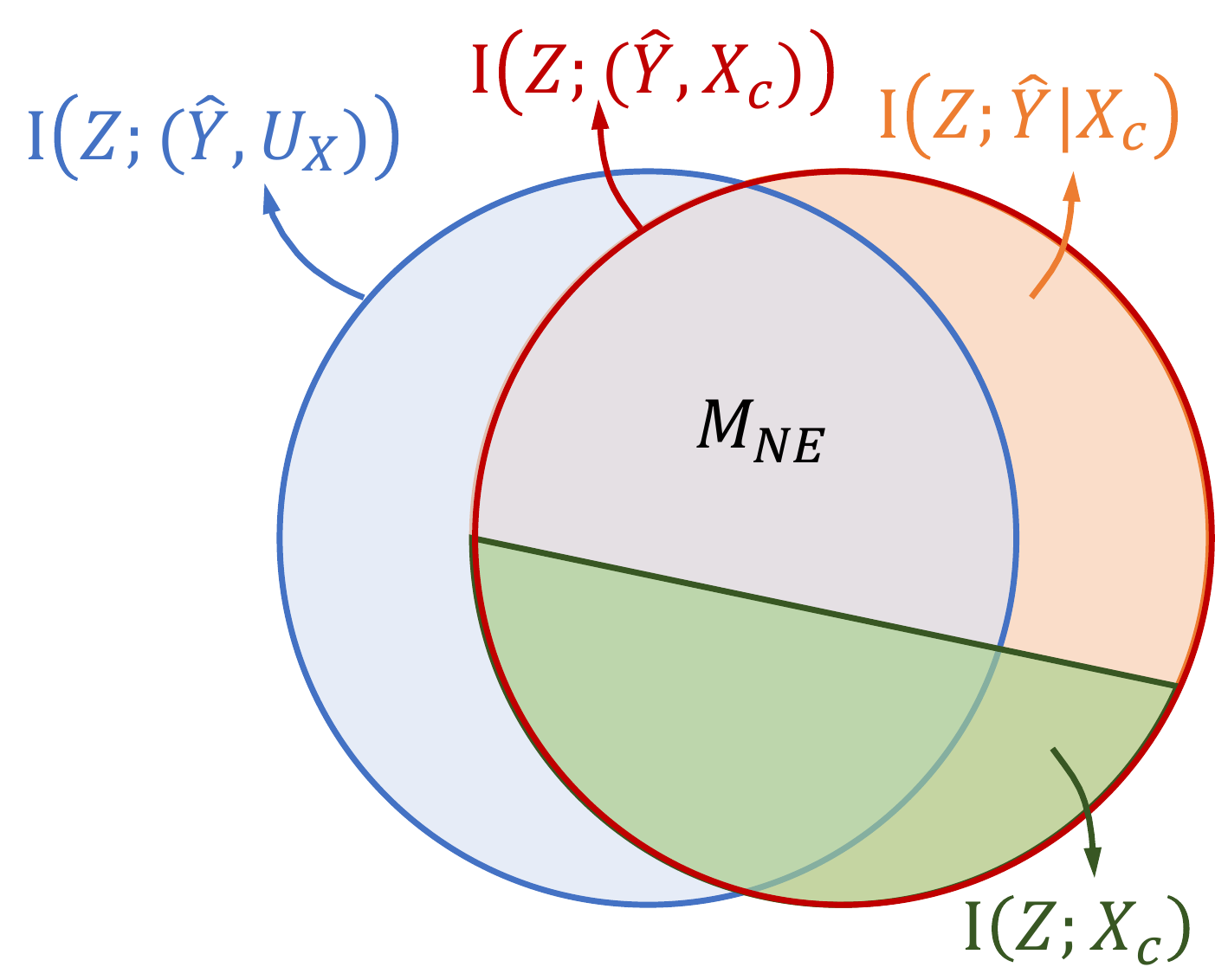}
\includegraphics[width=5cm]{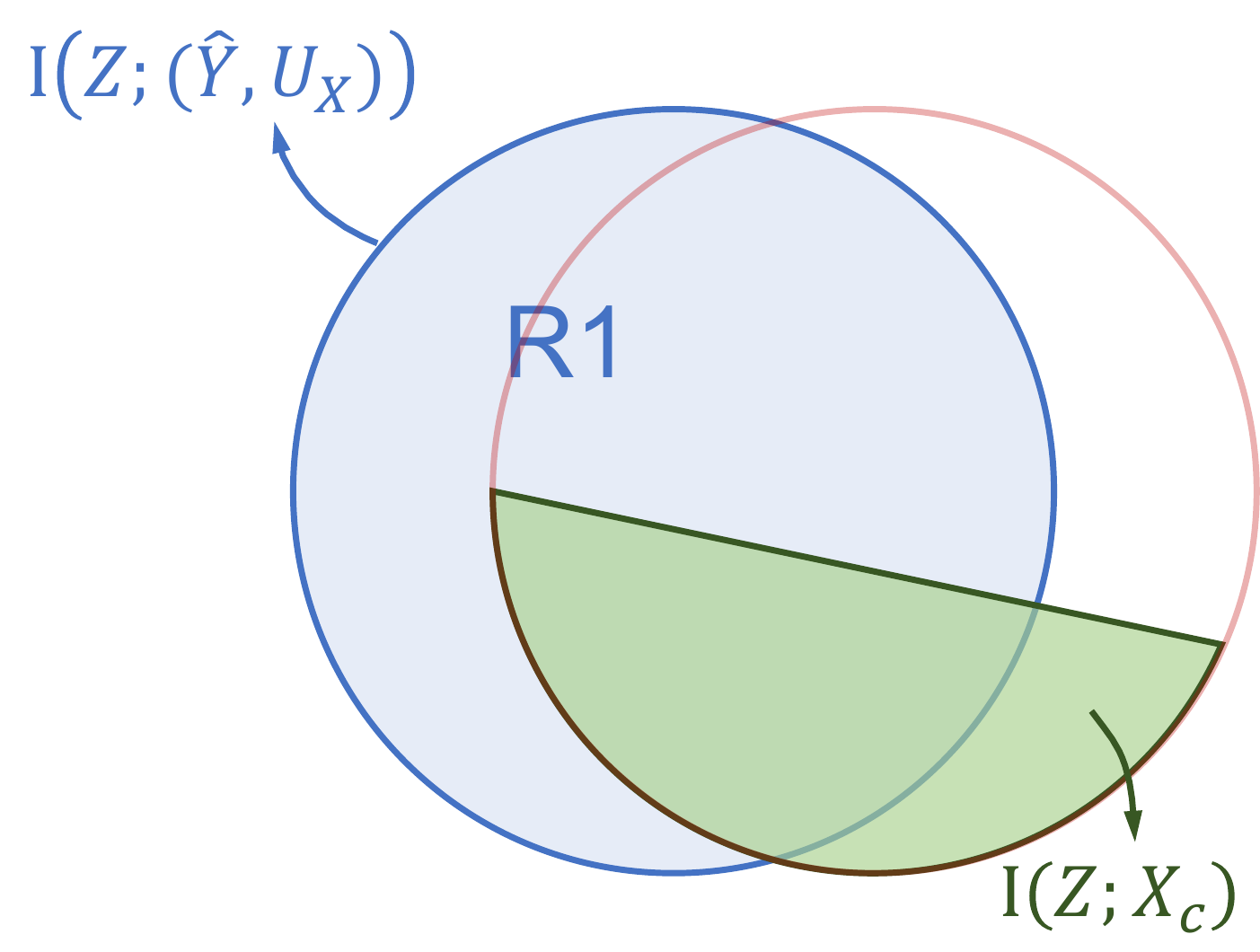}
\includegraphics[width=5cm]{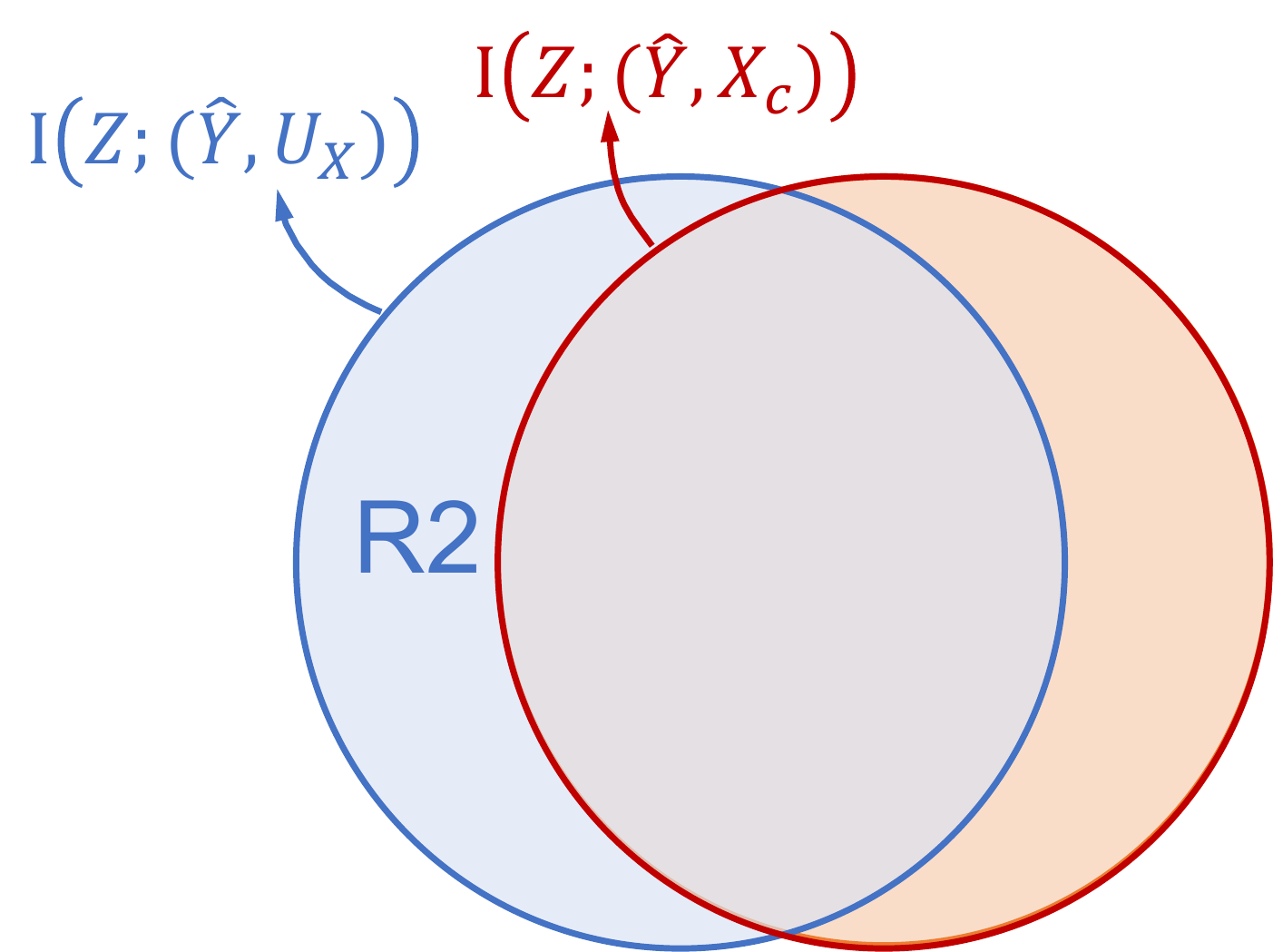}
\caption{(Top) Notice that the blue full-circle denotes $\mut{Z}{(\hat{Y},U_X)}$ and the red full-circle denotes $\mut{Z}{(\hat{Y},X_c)}$. The term $\mut{Z}{(\hat{Y},X_c)}$ is equal to the sum of $\mut{Z}{X_c}$ (green half-circle) and $\mut{Z}{\hat{Y}\given X_c}$ (orange half-circle). The candidate measure $(M_{NE})$ is the intersecting volume between $\mut{Z}{(\hat{Y},U_X)}$ and $\mut{Z}{\hat{Y}\given X_c}$. Next, we show pictorially that this intersecting volume is given by $R1-R2$ where $R1$ is shown in the middle figure and $R2$ is shown in the rightmost figure. (Middle) Notice that $R1= \uni{Z}{(\hat{Y},U_X)| X_c}$. (Bottom) Notice that $R2=\uni{Z}{(\hat{Y},U_X)| (\hat{Y},X_c)}.$ \label{fig:inter} }
\end{figure}

\noindent 
\textbf{Limitations of Candidate Measure~\ref{candmeas:intersection}:} This measure does resolve many of the examples and satisfies several desirable properties (discussed more in \cite{dutta2020information}). However, it fails to capture certain types of non-exempt masked disparity when the mask arises from $X_g$, e.g., scenarios like Canonical Example~\ref{cexample:masking_general} in Section~\ref{subsec:rationale}, where non-exempt masked disparity is present even though $Z-X_c-\hat{Y}$ form a Markov chain.


\section{Appendix to Section~\ref{sec:measure_decomposition}}

\subsection{Proof of Theorem~\ref{thm:measure_decomposition} and Lemma~\ref{lem:maskingequivalence}}
\label{app:measure_decomposition}
\decompositionthm*

\begin{proof}[Proof of Theorem~\ref{thm:measure_decomposition}] 

First consider $M_{V,NE}=\uni{Z}{\hat{Y}| X_c}$ and $M_{V,E}=\rd{Z}{(\hat{Y}, X_c)}$. Because all PID terms are non-negative by definition, both $M_{V,NE}$ and $M_{V,E}$ are non-negative. 

Now, consider $M_{M,E}$. Observe that,
\begin{align}
 &M_{M,E} = \mut{Z}{(\hat{Y},U_X)} - \mut{Z}{\hat{Y}}-M_{M,NE} \nonumber \\
& = \mut{Z}{\hat{Y}}+ \mut{Z}{U_X\given \hat{Y}}- \mut{Z}{\hat{Y}}-M_{M,NE} \nonumber && \text{[Chain Rule for mutual information]} \nonumber \\
& =  \mut{Z}{U_X\given \hat{Y}}-M_{M,NE} \nonumber \\
& = \mut{Z}{U_X\given \hat{Y}}- M^*_{NE}+M_{V,NE} && \text{[By Definition]}\nonumber \\
& = \mut{Z}{U_X\given \hat{Y}}  -\min_{U_a,U_b\text{ s.t. } U_a=U_X\backslash U_b}\mathrm{Uni}((Z,U_a):(\hat{Y},U_b)| X_{c}) + \uni{Z}{\hat{Y} | X_{c}}  && \text{[By Definition]} \nonumber \\
& \geq \mut{Z}{U_X\given \hat{Y}} -\mathrm{Uni}(Z:(\hat{Y},U_X)| X_{c})  + \uni{Z}{\hat{Y} | X_{c}}    \nonumber \\
& \geq \mut{Z}{U_X\given \hat{Y}} -\mathrm{Uni}(Z:(\hat{Y},U_X)| \hat{Y})    && \text{[Triangle Inequality (Lemma~\ref{lem:triangle_inequality})]} \nonumber \\
& \geq \mut{Z}{U_X\given \hat{Y}} -\mut{Z}{(\hat{Y},U_X)\given  \hat{Y}} \ \text{[\eqref{eq:pid3} in Section~\ref{subsec:background}]} \nonumber \\
& = \mut{Z}{U_X\given \hat{Y}} -\mut{Z}{U_X\given  \hat{Y}} - \mut{Z}{\hat{Y}\given U_X, \hat{Y}}  && \text{[Chain Rule for mutual information]} \nonumber \\
& =0.
\end{align}

Lastly, we consider $M_{M,NE}$.

\begin{align}
&M_{NE} = \min_{U_a,U_b\text{ s.t. } U_a=U_X\backslash U_b} \uni{(Z,U_a)}{(\hat{Y},U_b)| X_c} -\uni{Z}{\hat{Y}| X_c}\nonumber
\\& = \uni{(Z,U^*_a)}{(\hat{Y},U^*_b)| X_c}-\uni{Z}{\hat{Y}| X_c} && \text{[for some $(U_a^*,U_b^*)$]} \nonumber
\\ & \geq \uni{Z}{(\hat{Y},U^*_b)| X_c}-\uni{Z}{\hat{Y}| X_c} && \text{[Using Lemma~\ref{lem:monotonicity_alice}]} \nonumber \\
& \geq \uni{Z}{\hat{Y}| X_c}-\uni{Z}{\hat{Y}| X_c} && \text{[Using Lemma~\ref{lem:monotonicity_bob}]}\nonumber \\
& =0.
\end{align}

\end{proof}




\maskingequivalence*


\begin{proof}[Proof of Lemma~\ref{lem:maskingequivalence}]

Before proceeding, note that, $\mathrm{I}(Z; \hat{Y},U_X)=\mathrm{I}(Z; U_X) + \mathrm{I}(Z; \hat{Y}\given U_X) = \mathrm{I}(Z; \hat{Y}\given U_X)$ because $Z$ is independent of $U_X$. This also leads to the masked disparity being equal to $\mathrm{I}(Z; \hat{Y}\given U_X)  - \mathrm{I}(Z;\hat{Y})$.

First, we show that the first statement implies the second statement. Suppose that, masked disparity $\mathrm{I}(Z; \hat{Y}\given U_X)  - \mathrm{I}(Z;\hat{Y})>0$. Then, we can choose the function $G=U_X$ such that $\mathrm{I}(Z; \hat{Y}\given G)  - \mathrm{I}(Z;\hat{Y})>0$. Thus, the implication holds.

We will now show that the second statement also implies the first statement. First note that, using Lemma~\ref{lem:maxCMI}, for any deterministic $g(\cdot)$, we always have $\mathrm{I}(Z;\hat{Y}\given U_X)\geq \mathrm{I}(Z;\hat{Y}\given g(U_{X})).$ Now, suppose there exists a $G=g(U_{X})$ such that $\mathrm{I}(Z;\hat{Y}\given G)>\mathrm{I}(Z;\hat{Y})$. Then, 
$\mathrm{I}(Z;\hat{Y}\given U_X)\geq \mathrm{I}(Z;\hat{Y}\given g(U_{X})) >\mathrm{I}(Z;\hat{Y}),$ implying masked disparity is present.

Thus, we prove that the first and second statements are equivalent.

\end{proof}


\section{Appendix to Section~\ref{sec:observational}}
\label{app:observational}

\unisatisfyproperties*

\begin{proof}[Proof of Lemma~\ref{lem:uni_satisfy_properties}]For Property~\ref{propty:cancellation}, observe that, 
\begin{align}
& \mathrm{CCI}(Z\rightarrow \hat{Y})=0  \nonumber\\
&\implies \mathrm{I}(Z;\hat{Y})=0  \nonumber\\
&\implies \uni{Z}{\hat{Y} | X_{c}}+\mathrm{Red}(Z:(\hat{Y}, X_{c})) = 0 && \text{[Using \eqref{eq:pid2} in Section~\ref{subsec:background}]}  \nonumber \\
&\implies \uni{Z}{\hat{Y} | X_{c}}=0 && \text{[Non-negativity of PID terms]}.
\end{align}

Property~\ref{propty:synergy} is trivially satisfied because the property itself requires that $\uni{Z}{\hat{Y} | X_{c}}>0$.

Property~\ref{propty:nonincreasing} is satisfied using Lemma~\ref{lem:monotonicity_eve} in Appendix~\ref{app:pid_properties} (originally derived in  \cite[Lemma 32]{banerjee2018unique}).

Property~\ref{propty:complete_exemption} is satisfied because $\hat{Y}$ is a deterministic function of the entire $X$, and hence the Markov chain $Z-X-\hat{Y}$ holds. Thus $\mut{Z}{\hat{Y}\given X_c}=0$, also implying $\uni{Z}{\hat{Y}| X_c}=0$.

\end{proof}

\cmisatisfyproperties*

\begin{proof}[Proof of Lemma~\ref{lem:conMI1_satisfy_properties}]

For Property~\ref{propty:synergy}, observe that
\begin{align}
& \uni{Z}{\hat{Y} | X_{c}} > 0  \nonumber \\
& \implies \mathrm{I}(Z;\hat{Y}\given X_{c}) > 0 && \text{[Using \eqref{eq:pid3} in Section~\ref{subsec:background} and non-negativity of PID terms]}.
\end{align}

Property~\ref{propty:complete_exemption} is satisfied because $\hat{Y}$ is a deterministic function of the entire $X$, and hence the Markov chain $Z-X-\hat{Y}$ holds.

\end{proof}

\section*{Acknowledgment}

This work was supported by an NSF Career Award and NSF grant CNS-1704845 as well as by DARPA and the Air Force Research Laboratory under agreement number FA8750-15-2-0277. The U.S. Government is authorized to reproduce and distribute reprints for Governmental purposes not withstanding any copyright notation thereon. The views, opinions, and/or findings expressed are those of the author(s) and should not be interpreted as representing the official views or policies of DARPA, the Air Force Research Laboratory, the National Science Foundation, or the U.S. Government.  S. Dutta was supported by the Cylab Presidential Fellowship 2020, K\&L Gates Presidential Fellowship in Ethics and Computational Technologies 2019 and the Axel Berny Graduate Fellowship 2019. P. Venkatesh was supported by a Fellowship in Digital Health from the Center for Machine Learning and Health at Carnegie Mellon University.

\bibliographystyle{IEEEtran}
\bibliography{main}

\end{document}